\documentclass[a4paper,reqno,onecolumn,11pt,accepted=2024-30-11]{quantumarticle}
\pdfoutput=1

\usepackage[utf8]{inputenc}
\usepackage{amsmath,bbm,physics,hyperref}
\hypersetup{citecolor=blue}
\usepackage{cleveref}
\usepackage{url}

\usepackage[style=numeric,maxbibnames=99,giveninits=true]{biblatex}
\DefineBibliographyExtras{UKenglish}{}
\usepackage{import}
\usepackage{packagesQ}
\usepackage{definitionsQ}

\title[Energy preserving evolutions over Bosonic systems]{Energy preserving evolutions over Bosonic systems}
\author{Paul Gondolf}
\affiliation{Department of Mathematics, Eberhard Karls University T\"{u}bingen, 72074 T\"{u}bingen, Germany}
\email{paul.gondolf@uni-tuebingen.de}

\author{Tim M\"obus}
\affiliation{Department of Mathematics, Technical University of Munich, 80333 M\"unchen, Germany}
\affiliation{Munich Center for Quantum Science and Technology (MCQST), 80799 M\"unchen, Germany}
\affiliation{Department of Mathematics, Eberhard Karls University T\"{u}bingen, 72074 T\"{u}bingen, Germany}
\email{moebustim@gmail.com}

\author{Cambyse Rouz\'e}
\affiliation{Department of Mathematics, Technical University of Munich, 80333 M\"unchen, Germany}
\affiliation{Munich Center for Quantum Science and Technology (MCQST), 80799 M\"unchen, Germany}
\email{rouzecambyse@gmail.com}
\affiliation{Inria, T\'{e}l\'{e}com Paris - LTCI, Institut Polytechnique de Paris, 91120 Palaiseau, France}

\begin{document}

\begin{abstract}
    The exponential convergence to invariant subspaces of quantum Markov semigroups plays a crucial role in quantum information theory. One such example is in bosonic error correction schemes, where dissipation is used to drive states back to the code-space --- an invariant subspace protected against certain types of errors. In this paper, we investigate perturbations of quantum dynamical semigroups that operate on continuous variable (CV) systems and admit an invariant subspace. First, we prove a generation theorem for quantum Markov semigroups on CV systems under the physical assumptions that (i) the generator is in GKSL form with corresponding jump operators defined as polynomials of annihilation and creation operators; and (ii) the (possibly unbounded) generator increases all moments in a controlled manner. Additionally, we show that the level sets of operators with bounded first moments are admissible subspaces of the evolution, providing the foundations for a perturbative analysis. Our results also extend to time-dependent semigroups and multi-mode systems. We apply our general framework to two settings of interest in continuous variable quantum information processing. First, we provide a new scheme for deriving continuity bounds on the energy-constrained capacities of Markovian perturbations of quantum dynamical semigroups. Second, we provide quantitative perturbation bounds for the steady state of the quantum Ornstein-Uhlenbeck semigroup and the invariant subspace of the photon dissipation used in bosonic error correction.
\end{abstract}

\maketitle

\newpage

\tableofcontents

\vspace*{\fill}
\emph{Acknowledgments:} {\footnotesize We would like to thank Yu-Jie Liu for the fruitful discussions that initiated this project. Moreover, we would like to thank Marius Lemm, Angela Capel-Cuevas, Simone Warzel, Michael M. Wolf, Libor Caha, Vjosa Blakaj, Shin Ho Choe, Robert Salzmann, Simon Becker, Lauritz van Luijk, Niklas Galke for the valuable feedback and discussions on the topic. Moreover, we would like to thank Jochen Schmid for his help with questions on evolution systems. T.M.~and C.R.~acknowledge the support of the Munich Center for Quantum Sciences and Technology, C.R.~that of the Humboldt Foundation, and P.G., T.M., and C.R.~that of the Deutsche Forschungsgemeinschaft (DFG, German Research Foundation) – Project-ID 470903074 – TRR 352.}

\newpage

\section{Introduction}\label{sec:introduction}
	\addtocontents{toc}{\protect\setcounter{tocdepth}{1}}
	Quantum processes are often described in physics by their infinitesimal action during short time intervals. When the process at a given time $t$ can be assumed to be independent of previous times, an assumption often referred to as the memorylessness condition, the resulting dynamics can be formally described via a so-called master equation. The latter is an initial value problem of the form
	\begin{equation*}
		\frac{\partial}{\partial t} \rho(t)=\cL(\rho(t))\,,\qquad \rho(0)=\rho_0\,.
	\end{equation*}
	In the case of uniformly continuous dual dynamics over a von Neumann algebra, i.e.~a $*$-subalgebra of $\mathcal{B}(\cH)$ closed under involution (see \cite[Sec.~2.4.2]{Bratteli.1987} for more details), the seminal results by \citeauthor{Lindblad.1976} and \citeauthor{Gorini.1976} \cite{Lindblad.1976,Gorini.1976} classify the generators of a quantum dynamical semigroup in terms of the following so-called GKSL form: $\cL$ is a bounded operator on the space $\cT_1(\cH)$ of trace-class operators that satisfies
	\begin{equation}\label{eqGKSL}
		\cL(\rho) = - i[H, \rho] + \sum\limits_{j = 1}^K L_j \rho L_j^\dagger - \frac{1}{2}\{L_j^\dagger L_j, \rho\}\, 
	\end{equation}
	for a Hamiltonian $H=H^\dagger\in\cB(\cH)$ and so-called Lindblad operators $L_j\in\cB(\cH)$, where $\{A,B\}\coloneqq AB+BA$ denotes the anticommutator of two bounded operators $A,B\in\cB(\cH)$. In other words, the state at time $t\ge 0$ is described as $e^{t\cL}(\rho)$, where the exponential can be defined e.g.~in terms of its converging Taylor series. In that case, the set $(e^{t\cL})_{t\ge 0}$ defines a quantum Markov semigroup (QMS), which is a time-continuous family of completely positive, trace-preserving maps. 
	
	Since their introduction, QMSs have become a standard tool and have been extensively studied in various areas of mathematical physics and quantum information processing. Unfortunately, an extension of the GKSL form \eqref{eqGKSL} is known to fail for strongly continuous evolutions in general and thus requires additional assumptions. Conversely, unbounded operators satisfying an equation like \eqref{eqGKSL} on a suitable domain can fail at generating quantum Markovian dynamics. Simple counterexamples can be constructed in the context of continuous variable (CV) quantum systems over $\cH=L^2(\cH)$ as follows: denoting the creation and annihilation operators associated with a harmonic oscillator by $a^\dagger$ and $a$, respectively, the $2$-photon pure birth process formally defined as in \eqref{eqGKSL} with $K=1$, $H=0$ and $L_1=(a^\dagger)^2$ leads to a semigroup satisfying the master equation but failing to preserve the trace \cite[Example 3.3.]{Davies.1977}. Similar problems were encountered later on by Fagnola et al.~\cite{Chebotarev.1998, Fagnola.1999}, who considered the problem in the Heisenberg instead of the Schrödinger picture. They solved the appearing issues by imposing additional technical conditions on the generators in question. A thorough analysis of semigroups that have GKSL form but fail to be trace-preserving, can be found in \cite{Siemon.2017}, where the authors also discuss the possibility of generators deviating from the GKSL form. Beyond generation theory, a priori estimates on quantum Markov semigroups are key for perturbation theory of $C_0$-semigroups, which have been considered for example in \cite{Chebotarev.2003,Chebotarev.1998}.
	
	In contrast, recent years have seen remarkable progress in the use of CV quantum systems. These systems encompass a wide range of applications in various areas of quantum information, including quantum communication \cite{Braunstein.2005,Holevo.2001,Wolf.2007,Takeoka.2014,Pirandola.2017,Wilde.2017,Rosati.2018,Lami.2023}, sensing \cite{Aasi.2013,Zhang.2018,Meyer.2001,McCormick.2019} or simulation \cite{Flurin.2017}, enabled by advancements in non-classical radiation sources \cite{Ourjoumtsev.2006,Kurochkin.2014,Huang.2015,Reimer.2016,Eichler.2011,Zhong.2013}.
	Given the technological and experimental relevance of CV systems, there is a pressing need for a rigorous mathematical theory of quantum dynamical semigroups over such systems.
	
	One specific area where CV systems governed by a Lindblad master equation like \eqref{eqGKSL} have recently gained significant attention on theoretical as well as experimental grounds is in the field of bosonic quantum error correcting codes
	\cite{Gottesman.2001,Mirrahimi.2014,Guillaud.2019,Joshi.2021,Chamberland.2022,Ofek.2016,Michael.2016,Leghtas.2015,Rosenblum.2018,CampagneIbarcq.2020,Berdou.2023}. In particular, a certain class of CV codes known as CAT qubit codes has focused the attention of the community for their property of dynamically preserving quantum information through the action of a class of suitably engineered QMS which, loosely referred to as dissipative CAT-qubit dynamics in the present introduction. However, a mathematically rigorous analysis of these codes has only gotten little attention to the best of our knowledge, with the notable exception of \cite{Azouit.2016}.
	
	Much theoretical work has focused on the more tractable generators of Gaussian dynamics semigroups, where the generator $\cL$ is expressed as a quadratic form in the creation and annihilation operators \cite{Hudson.1984, Fagnola.1994, Cipriani.2000, Agredo.2021,Heinosaari.2009}. For those generators, the Feller property as well as properties of the spectrum and convergence results are known \cite{Cipriani.2000,Carbone.2003,Carbone.2007,Carlen.2017,DePalma.2018}. Since generators of dissipative CAT-qubit dynamics typically involve higher order monomials in $a$ and $a^\dagger$, the establishment of a more general theory of CV quantum Markov semigroups including them is timely.

	\subsection{Framework}
	
		Note that the technical details are introduced in Section \ref{sec:preliminaries}. Here, we just give an overview of the framework and subsequently the results: In this paper, we consider an operator $\cL$ on the space $\cT_f$ of finite linear combinations of rank-one operators of the form $\ketbra{k}{l}$, where given $k\in\mathbb{N}$, $\ket{k}\in L^2(\mathbb{R})$ denotes the $k$-photon Fock state. We further assume that $\cL$ satisfies the following two conditions: (i) $\cL$ has a GKSL structure \eqref{eqGKSL}, where the Hamiltonian $H$ as well as the jump operators $L_j$ are polynomials of the annihilation and creation operators; (ii) the following condition is satisfied: for a divergent sequence $\{k_r\}_{r \in \N}$ in $\R_+$, there exist real coefficients $\{w_{k_r}\}_{r \in \N}$ such that for all states $\rho\in \cT_f$:
		\begin{equation*}
			\tr[\cL(\rho)(N+\1)^{k_r/2}]\leq w_{k_r}\tr[\rho(N+\1)^{k_r/2}]\,.
		\end{equation*}
		Above, $N\coloneqq a^\dagger a=\sum_{n\in \mathbb{N}}n\ketbra{n}{n}$ denotes the photon number operator.
		This assumption implies not only that $\cL$ defines a quantum dynamical semigroup, but also a quasi-contractive semigroup on the weighted Banach spaces $(\cD(\cW^k),\|\cW^k(\cdot)\|_{1})$ defined through the operator
		\begin{equation*}
			\cW(\cdot)\coloneqq (N+\1)^{1/4}\,(\cdot)\,(N+\1)^{1/4}\,.
		\end{equation*}
		Here $\cD(\cW^k)$ denotes the domain of the operator $\cW^k$. In the latter, we refer to these spaces as Sobolev spaces and denote them by $W^{k,1}$ in analogy with their classical analogues (see also \cite{Becker.2021}). Note that all mentioned definitions and results are extended to multi-mode systems later on. Next, we call operators $\cL$ that satisfy both conditions (i) and (ii) generators of Sobolev preserving quantum dynamical semigroups. Indeed, in our first main result, we show that such operators generate QMSs with the extra property that the latter preserves Sobolev spaces. More precisely:
		
		\begin{thm*}[Generation of bosonic semigroups, see Theorem \ref{thm:generation-theorem}]
			Let $(\cL, \cD(\cL))$ be an operator defined on the Banach space $\cT_{1,\operatorname{sa}}$ of self-adjoint, trace-class operators. If $(\cL, \cD(\cL))$ satisfies conditions (i) and (ii) above, then its closure $\overline{\cL}$ generates a strongly continuous, positivity preserving semigroup $(\cP_t)_{t\ge 0}$ on $W^{k, 1}$ for all $k \in \R_+$ with 
			\begin{equation*}
				\norm{\cP_t}_{W^{k, 1} \to W^{k, 1}} \le e^{\omega_k t} \,\quad \forall t\ge 0\, . 
			\end{equation*}
			where $\omega_k = \frac{k_{r_1} - k}{k_{r_1} - k_{r_0}}\omega_{k_{r_0}} + \frac{k - k_{r_0}}{k_{r_1} - k_{r_0}}\omega_{k_{r_1}}$ for an $r$ such that $k_{r_0}\leq k <k_{r_1}$.
			Finally, for $k = 0$, the semigroup is contractive and trace-preserving.
		\end{thm*}
		
		Additionally, our setup is directly suited to the establishment of a perturbation analysis akin to 
		the result reported in \cite{Szehr.2013} in the finite dimensional setting. Moreover, in some cases, our analysis allows us to conclude the existence of adherence points for the dynamics in the large time limit. We manage to prove the requirements (i)-(ii) of the generation theorems as well as rigorous perturbation analysis for several examples including dissipative CAT-qubit dynamics as well as Gaussian and quantum Ornstein Uhlenbeck generators. For the latter, we show for instance the following perturbation bound for all $t\ge 0$ (see Proposition \ref{propqOUperturb} and Corollary \ref{corECDN}):
		
		\begin{prop*}
			Let $(\cL_{\operatorname{qOU}},\cT_f)$ be the generator of the quantum Ornstein Uhlenbeck semigroup with jump operators ${\lambda}a$ and $\mu a^\dagger$, $\lambda>\mu\geq0$ and $(\cL_G,\cT_f)$ a Gaussian perturbation with unique jump $\gamma a+\eta a^\dagger$ with $\gamma,\eta\in\mathbb{R}$, and $\varepsilon>0$. Then, assuming $\lambda^2-\mu^2+|\gamma|^2-|\eta|^2> 0$, $\cL_{\operatorname{qOU}}+\varepsilon\cL_G$ generates a positivity and Sobolev preserving semigroup on $W^{k,1}$ for $k\geq1$, and there exist uniformly bounded functions $C(\varepsilon),D(\varepsilon)$ depending on $\lambda,\mu,|\eta|,|\gamma|$ such that, for all $t\ge 0$ and all state $\rho\in W^{k,1}$,
			\begin{equation}
				\Big\|\left(e^{t \cL_{\operatorname{qOU}}}-e^{t(\cL_{\operatorname{qOU}}+\varepsilon\cL_G)}\right)(\rho)\Big\|_{1}\leq \varepsilon\, C(\varepsilon)\, \max\Big\{\|\rho\|_{W^{2,1}},D(\varepsilon)\Big\}\,.
			\end{equation}
			In particular, for all $t\ge 0$
			\begin{equation}
				\Big\|e^{t \cL_{\operatorname{qOU}}}-e^{t(\cL_{\operatorname{qOU}}+\varepsilon\cL_G)}\Big\|_{\diamond}^E\leq\,(1+E) \varepsilon\, C(\varepsilon)\, \max\Big\{1,D(\varepsilon)\Big\}\,,
			\end{equation}
			where $\|.\|_\diamond^E$ denotes the energy-constrained diamond norm defined in \Cref{ECnorm}.
		\end{prop*}
		
		We also note that our theory extends to the case of a time-dependent generator as well as to the multi-mode setting $\cH=L^2(\mathbb{R}^m)$, $m\ge 1$, see Section \ref{sec:timedependentgeneration} and \ref{sec:multi-mode-extension}. To prove our generation theorems, the compactly embedded Sobolev spaces play a crucial role and provide an interesting proof strategy, which follows the original method of Davies \cite{Davies.1977} by an explicit reduction to the seminal theorems by Hille, Yosida \cite{Hille.2012} and Feller, Myadera, Lumer and Phillips \cite[Thm.~II.3.8]{Engel.2000}.

	\subsection{Dissipative CAT-qubit dynamics}
	
		As mentioned before, the interest in continuous variable QMS has been reignited by the modelling capabilities of a certain class of error-corrected universal quantum computing architectures. In \cite{Azouit.2015} and later in \cite{Azouit.2016}, Azouit, Sarlette, and Rouchon prove the well-posedness of the dynamics that stabilises an $l$ dimensional code-space, with a generator given for a fixed $\alpha\in\mathbb{R}$ by 
		
		\begin{align}\label{deflphotondissip}
			\cL_l(\rho) = L_l\rho L_l^\dagger-\frac{1}{2}\big\{L_l^\dagger L_l,\rho\big\}\qquad\text{with}\qquad L_l\coloneqq a^l-\alpha^l\1\,.
		\end{align}
		
		In addition, they identified invariant operators of the dynamic and further showed that the semigroup exponentially drives states towards the code-space spanned by these invariants. By constructing a Banach space from composites of the generator, i.e.~$L_l=a^l-\alpha^l\1$, compactly embedded in the self-adjoint trace class operators, they judiciously circumvented the problems previously encountered when trying to take limits of the minimal semigroups. This procedure was very much tailored towards the simple structure of the generator and also relied on a favourable commutation relation of $a^l-\alpha^l\1$ and $(a^l-\alpha^l\1)^\dagger$, which one cannot hope for in general. In contrast, here we do not use parts of the generator to create our compactly embedded spaces, but instead use the most natural candidate at hand, namely the number operator $N$. Generalising the idea of Azouit et al.~we take limits of sequences of semigroups for which our CV Sobolev spaces are admissible subspaces to prove our generation theorems. Combining their exponential dynamical convergence, stated as
		\begin{equation*}
			\tr[L_l\left(e^{t\cL_l}(\rho)-\overline{\rho}\right) L_l^\dagger]\leq e^{-l!t}\tr[L_l|\rho-\overline{\rho}| L_l^\dagger]\,,
		\end{equation*}
		where $\overline{\rho}$ is a $\rho$-dependent state in the code-space, with our generation and perturbation theory, we can, for example, show that any $l$-photon dissipation perturbed by a Hamiltonian admits the following large-time perturbation bounds (see Theorem \ref{thm:l-diss-hamiltonian-perturbation}):
		\begin{thm*}
			Let $\cL_l$ be the $l$-photon dissipation defined in \Cref{deflphotondissip}    
			and $p_H\in\C[X,Y]$ with $\deg(p_H)=d_H\leq2(l-1)$ such that $H=p_H(a,\ad)$ is a symmetric operator. Then, there exist constants $c,\gamma>0$ depending on $\alpha$ and $l$  such that for $\varepsilon\geq0$ and all states $\rho\in W^{6l-4,1}$
			\begin{equation*}
				\begin{aligned}
					\Big|\tr[L\left(e^{t\cL_l}(\rho)-e^{t(\cL_l+\varepsilon\cH[H])}(\rho)\right)L^\dagger]\Big|\leq\varepsilon c\left(1-e^{-l!t}\right)\max\{\gamma,\|\rho\|_{W^{6l-4,1}}\}
				\end{aligned}
			\end{equation*}
			where $\cH[H](\rho)\coloneqq -i[H,\rho]$.
		\end{thm*}
		
		The same idea can be extended to more general setups and thereby extends the result by \citeauthor{Szehr.2013} from finite dimensions to the case of strongly continuous semigroups over infinite dimensional systems. 

	\subsection{Outline of the paper}
	
		In Section \ref{sec:preliminaries}, we begin with an introduction to basic Banach space and operator theory, followed by a short overview of this theory in the context of Hilbert spaces and their associated bounded, compact and trace-class operator spaces. Building upon that we then introduce in Section \ref{subsec:weighted-norms-compact-embedding} the notion of compact embeddings and weighted Banach spaces, followed by basic semigroup theory in Section \ref{subsec:c0-semigroups}. More specific to our application, we then briefly recapitulate Bosonic Hilbert spaces, and relevant operators thereon, and introduce our Bosonic Sobolev spaces. We prove that they are compactly embedded into one another and provide an interpolation theorem in the spirit of the Stein-Weiss theorem for weighted $L_p$ spaces. In Section \ref{sec:polynomial-generators}, we begin by showing the generation theorem in the time-independent case and then employ this theorem in Section \ref{sec:timedependentgeneration} to prove a generation theorem for generators composed of polynomials in $a$ and $a^\dagger$ with coefficients that are continuous functions of time. We extend our analysis to the multi-mode setting in Section \ref{sec:multi-mode-extension} where we lift the generation theorems from the chapter before.\par
		Section \ref{sec:examples-sobolev-preserving-semigroup} begins with a short proposition making better use of tighter input-output moments of the generator and showing the existence of adherence points in the asymptotic time regime for semigroups that admit such bounds. We then proceed to prove the generation theorem for the quantum Ornstein Uhlenbeck generator as well as for a family of dissipative CAT-qubit dynamics in Section \ref{sec:cat-qubits}. This section is then followed by large time perturbation bounds for both the quantum Ornstein Uhlenbeck semigroup as well as the dissipative CAT-qubit dynamics in Section \ref{sec:example-perturbation-bounds}.

\section{Preliminaries}\label{sec:preliminaries}
\addtocontents{toc}{\protect\setcounter{tocdepth}{2}}
	
	We begin with a short review of valuable tools from Banach space theory in Section \ref{subsec:banach-space} and build upon them to prove a compact embedding theorem for a class of weighted spaces in Section \ref{subsec:weighted-norms-compact-embedding}. We then recall standard results from the theory of strongly continuous semigroups as well as evolution systems in Section \ref{subsec:c0-semigroups}. These will play an essential role in Section \ref{sec:polynomial-generators}. Finally, we introduce continuous variable quantum systems, provide some valuable properties of polynomials of annihilation and creation operators, and introduce the notion of a \textit{Sobolev preserving semigroup}, which are the main objects of study in the remainder of the paper.

	\subsection{Basic Banach space theory}\label{subsec:banach-space}

		We start with a brief recap on notions from the theory of Banach spaces that will be needed in this paper, and refer to \cites[Chap.~III]{Kato.1995}[Chap.~IV]{Conway.2007}[Chap.~2-3]{Simon.2015}[Chap.~2-3]{Hille.2012} for more details. Let $(\cX,\|\cdot\|_{\cX})$ be a {Banach space}. We denote the space of bounded operators between two Banach spaces $\cX$ and $\cY$ by $\cB(\cX,\cY)$, with $\cB(\cX,\cX)=\cB(\cX)$. The identity map in $\cB(\cX)$ is denoted by $\1_{\cX}$, or simply $\1$ when the underlying space is clear from the context. The operator norm is denoted by
		\begin{equation}\label{eq:operator-norm}
			\|A\|_{\cX\rightarrow\cY}=\|A:\cX\to \cY\|\coloneqq\sup_{\|x\|_{\cX}=1}\|A(x)\|_{\cY}.
		\end{equation}
		We recall that the linear space $\cB(\cX,\cY)$ equipped with the operator norm is a Banach space since $(\cY,\|\cdot\|_{\cY})$ is a Banach space. An operator $A:\cX\rightarrow\cY$ is compact if the image sequence $\{Ax_n\}_{n\in\N}\subset\cY$ of any bounded sequence $\{x_n\}_{n\in\N}\subset\cX$ has a converging subsequence. In particular, every operator which can be approximated by a sequence of finite rank operators is compact \cites[Thm.~3.1.9]{Simon.2015}.
		
		More generally, an unbounded operator $A$ is a linear map $A:\cD(A)\subset\cX\rightarrow\cY$ defined on its domain $\cD(A)\subset\cX$. If the domain is dense in $\cX$, the operator is said to be densely defined. In this paper, all unbounded operators are densely defined. Note that the addition and concatenation of two unbounded operators $(A,\cD(A))$ and $(B,\cD(B))$ is defined on $\cD(A+B)=\cD(A)\cap\cD(B)$ and $\cD(AB)=B^{-1}(\cD(A))$ (cf.~\cite[Sec.~III§5.1]{Kato.1995}).
		An operator $(A,\cD(A))$ is closed iff its graph $\{(x,A(x)):x\in\cD(A)\}$ is a closed set in the product space $\cX\times\cY$. A bounded operator is closed iff its domain is closed. By convention, we extend all densely defined and bounded operators by the {bounded linear extension theorem} to bounded operators on $\cX$ \cite[Thm.~2.7-11]{Kreyszig.1989}. An operator is called {closable} if there exists a closed {extension}, where $\overline{A}$ is an extension of $A$ if $\cD(A)\subset\cD(\overline{A})$ and $Ax=\overline{A}x$ for all $x\in\cD(A)$. The closure of $A$ is denoted by $\overline{A}$ \cite[Sec.~7.1]{Simon.2015}. We also recall that for an unbounded operator $(A,\cD(A))$  on $\cX$, a {core} for $A$ is a subset $\cD_0\subset\cD(A)$ which is dense in $\cD(A)$ w.r.t.~the {graph norm} $\|\cdot\|_A\coloneqq\|A\cdot\|_{\cX}+\|\cdot\|_{\cX}$ of $A$ (cf.~\cite[Def. 1.6]{Simon.2015}). Given two linear operators $(A,\cD(A)))$ and $(B,\cD(B))$ on $\cX$, the operator $(B,\cD(B))$ is relatively $A$-bounded if $\cD(B)\subseteq \cD(A)$ and there are $a,b \geq 0$ for all $x\in\cD(\cL)$ such that
		\begin{equation}\label{eq:relative-bounded}
			\|B(x)\|_{\cX}\leq a\|A(x)\|_{\cX}+b\|x\|_{\cX}.
		\end{equation}
		\smallskip
		For a closed linear operator $(A, \cD(A))$ on a Banach space $\cX$ we call 
		\begin{equation*}
			\rho(A) \coloneqq \{\lambda \in \C: \lambda - A:\cD(A) \to \cX \text{ is bijective}\}
		\end{equation*}
		the resolvent set of $(A, \cD(A))$. For $\lambda \in \rho(A)$ we call the inverse $$R(\lambda, A) \coloneqq (\lambda - A)^{-1}$$ the resolvent, which is, by the closed graph theorem, a bounded operator on $\cX$.
		
		Besides the convergence w.r.t.~the operator norm (i.e.~uniform convergence), a sequence of operators $\{A_k\}_{k\in\N}$ defined on a common domain $\cD(A)$ converges strongly if $\lim_{k\rightarrow\infty}\|A_kx-Ax\|_{\cY}=0$ for all $x\in\cD(A)$. Based on the underlying topologies associated with these two convergences, one can define the Bochner integral of vector and operator-valued maps on a compact interval equipped with the Lebesgue measure, e.g.~$f:[a,b]\rightarrow\cX$ and $F:[a,b]\rightarrow \cB(\cX)$ with $a<b$. Under the assumption that the function $f$ or $F$ can be approximated by a step function and that the real-valued integral
		\begin{equation*}
			\int_{a}^b\|f(s)\|_{\cX}\,ds\qquad\text{or}\qquad\int_{a}^b\|F(s)\|_{\cX\rightarrow\cX}\,ds
		\end{equation*}
		is bounded, the Bochner integrals are defined by standard approximation with step functions. Since all the vector-valued maps considered in this work are continuous, the Bochner integral is always well-defined and coincides with the Riemann and Pettis integrals (more details can be found in \cite{Gordon.1991}). Similar to the real-valued case, the integral satisfies the triangle inequality w.r.t.~the norm, is invariant under closed linear transformations, and satisfies the fundamental theorem of calculus on continuous functions \cites[Sec.~3.7-8]{Hille.2012}.
		
		Two special cases of Banach spaces that we will consider are Hilbert spaces and bounded operators defined on Hilbert spaces. We denote the latter by $\cB(\cH)$ and use $\norm{\cdot}_\infty$ for their norm. Given a separable Hilbert space $\cH$ and $A\in\cB(\cH)$, its {adjoint} $A^\dagger$ is uniquely defined by
		\begin{equation}\label{eq:adjoint}
			\braket{A\phi\,,\varphi}=\braket{\phi\,,A^\dagger\varphi}
		\end{equation}
		for all $\phi,\varphi\in\cH$. The space of all bounded, {self-adjoint} operators, i.e.~$A=A^\dagger$, is denoted by $\cB_{\operatorname{sa}}(\cH)$. A special case of self-adjoint operators is those with finite support for a fixed orthonormal basis $\{|n\rangle\}_n$, whose set we denote by $\cT_f\equiv \cT_f(\cH) \coloneqq \{A\in\cB_{\operatorname{sa}}(\cH):\exists M\in\N\,:\,A = \sum_{n,m}^M a_{nm} \ketbra{n}{m}\}$. 
		By a slight abuse of notations, we denote the formal adjoint of an unbounded operator $(A,\cD(A))$ on $\cH$ as $A^\dagger:\cD(A^\dagger)\rightarrow\cH$, where the latter satisfies \Cref{eq:adjoint} for all $\phi\in\cD(A)$ and for all $\varphi$ in the maximally defined domain 
		\begin{equation*}
			\cD(A^\dagger)\coloneqq\{\ket{\varphi}\in\cH:\ket{\phi}\mapsto\braket{A\phi\,,\varphi}\text{ is bounded}\}.
		\end{equation*}
		The operator $A$ is called {symmetric} if for all $\ket{\phi},\ket{\varphi}\in\cD(A)$, $\braket{A\phi\,,\varphi}= \braket{\phi\,,A\varphi}$. $A$ is called self-adjoint if $\cD(A)=\cD(A^\dagger)$ and $A^\dagger=A$. An operator $(A,\cD(A))$ is positive if $\langle A\phi,\phi\rangle\ge 0$ for all $\phi\in\cD(A)$. In this case, we write $A\ge 0$. More generally we write $A\ge B$ if $A-B\ge 0$ with $\cD(A)\subset\cD(B)$. We defined the domain of the subtraction by $\cD(A-B)=\cD(A)\cap\cD(B)$. 
		The trace of a positive operator $A$ is defined by $\tr[A]=\sum_{n\in\N}\bra{n}A\ket{n}$. When $A\in\cB(\cH)$, its {trace-norm} is defined by $\|A\|_1\coloneqq\tr\big[|A|\big]$, where $|A|\coloneqq \sqrt{A^\dagger A}$ \cite[Thm.~2.4.4]{Simon.2015}. Bounded operators with finite trace-norm are called trace-class and their class we denote by $\cT_1\equiv \cT_1(\cH)$. We will most often consider the Banach space of self-adjoint trace-class operators denoted by $\cT_{1,\, \operatorname{sa}} \equiv \cT_{1,\, \operatorname{sa}}(\cH) \coloneqq \{A \in \cB_{\operatorname{sa}}(\cH):\norm{A}_1<\infty\}$. We define the set of density operators by $\cS\equiv \cS(\cH)\coloneqq\{\rho:\rho\geq0\text{ and }\tr[\rho]=1\}$.

	\subsection{Weighted norms and compact embeddings}\label{subsec:weighted-norms-compact-embedding}

		For a Banach space $(\cX,\|\cdot\|_{\cX})$ and an invertible operator $(\cW,\cD(\cW))$ on $\cX$, a natural way of defining a new norm out of $\|\cdot\|_{\cX}$ is via the following procedure:
		\begin{defi}[Weighted normed space]
			Let $(\cW,\cD(\cW))$ be an invertible linear operator. Then, $\|X\|_{\cW}\coloneqq\|\cW(X)\|_{\cX}$ for $X\in\cD(\cW)$ defines a norm on $\cD(\cW)$. In the following, we denote by $\|P\|_{\cW\to\cW}\coloneqq \sup_{\|X\|_\cW\le 1}\|P(X)\|_\cW$.
		\end{defi}
		
		In the next lemma, we prove that completeness of the space $(\cX,\|\cdot\|_{\cX})$ is preserved by the closedness of $(\cW,\cD(\cW))$:
		\begin{lem}\label{thm:weighted-banach-space}
			Let $(\cW,\cD(\cW))$ be an invertible linear operator on the Banach space $\cX$. Then, the weighted normed space $(\cD(\cW),\|\cdot\|_{\cW})$ is a Banach space and the norm $\|\cdot\|_{\cW}$ is equivalent to the {graph norm} of $\cW$.
		\end{lem}
		\begin{proof}
			First, it is clear that $\|\cdot\|_{\cW}$ defines a norm because the linearity of $\cW$ directly implies homogeneity and the triangle inequality, while the injectivity of $\cW$ gives positive definiteness. By the closed graph theorem, we can conclude that $\cW^{-1}$ is bounded, and therefore $\cW$ is closed (see \cite[Sec.~III.2]{Kato.1995} and \cite[Thm.~2.11.5]{Hille.2012}). Moreover,
			\begin{equation*}
				\|X\|_{\cW}=\|\cW(X)\|_{\cX}\leq\|X\|_{\cX}+\|\cW(X)\|_{\cX}\leq(\|\cW^{-1}\|_{\cX\to\cX}+1)\|X\|_{\cW}
			\end{equation*}
			shows that the graph norm of $\cW$ is equivalent to $\|\cdot\|_{\cW}$. By definition, the graph is a closed operator on the Banach space $\cX\times\cX$ such that the vector space $\cD(\cX)$ equipped with the graph norm is a Banach space. The statement hence follows.
		\end{proof}
		Next, we introduce a compact embedding \cite[Sec.~5.7]{Evans.2010} for weighted normed spaces, i.e.~we want to reduce the compact embedding to a relation between the defining operators:
		\begin{defi}\label{def:compact-embedding}
			Let $(\cX_1,\|\cdot\|_{\cX_1})$ and $(\cX_2,\|\cdot\|_{\cX_2})$ be Banach spaces such that $\cX_1\subset\cX_2$. We say that $\cX_1$ is compactly embedded in $\cX_2$, and denote this condition as $\cX_1\Subset\cX_2$ iff
			\begin{itemize}
				\item[$-$] $\exists c\geq0$ such that $\|\cdot\|_{\cX_2}\leq c\|\cdot\|_{\cX_1}$ and
				\item[$-$] any bounded sequence in $\cX_1$ has a converging subsequence in $\cX_2$ (i.e.~is precompact).
			\end{itemize}
		\end{defi}
		
		\begin{lem}[Compact embedding]\label{thm:compact-embedding-weighted-spaces}
			Let $(\cW_1,\cD(\cW_1))$ and $(\cW_2,\cD(\cW_2))$ be invertible linear operators on $\cX$ with bounded inverses and $\cD(\cW_1)\subset\cD(\cW_2)$. Then $(\cD(\cW_1),\|\cdot\|_{\cW_1})$ is compactly embedded in $(\cD(\cW_2),\|\cdot\|_{\cW_2})$ iff the extension of $\cW_2\cW_1^{-1}$ is a compact operator on $\cX$. 
		\end{lem}
		\begin{proof}
			First, we prove that the first condition in Definition \ref{def:compact-embedding} is equivalent to the boundedness of $\cW_2\cW_1^{-1}$, which is defined on $\cX$ by the {bounded linear extension theorem} \cite[Thm.~2.7-11]{Kreyszig.1989}. By assuming that $\cW_2\cW_1^{-1}$ is a bounded operator, which is implied by compactness, 
			\begin{equation*}
				\|X\|_{\cW_2}=\|\cW_2\cW_1^{-1}\cW_1(X)\|_{\cX}\leq\|\cW_2\cW_1^{-1}\|_{\cX\to\cX}\|X\|_{\cW_1}.
			\end{equation*}
			Now, if there is a $c\geq0$ such that $\|X\|_{\cW_2}\leq c\|X\|_{\cW_1}$ then boundedness is given by
			\begin{equation*}
				\|\cW_2\cW_1^{-1}(X)\|_{\cX}\leq c\|\cW_1\cW_1^{-1}(X)\|_{\cX}=c\|X\|_{\cX}.
			\end{equation*}
			Let $\{X_k\}_{k\in\N}\subset\cX$ be a bounded sequence and assume the second condition in Definition \ref{def:compact-embedding}. Then, $\cW_1^{-1}(X_k)$ is a bounded sequence in $(\cD(\cW_1),\|\cdot\|_{\cW_1})$ which admits a converging subsequence in $(\cD(\cW_2),\|\cdot\|_{\cW_2})$ by assumption. Therefore, $\cW_2\cW^{-1}_1(X_k)$ has a converging subsequence in $(\cX,\|\cdot\|_\cX)$. Conversely, assume $\cW_2\cW_1^{-1}$ is compact and $X_k$ a bounded sequence in $(\cD(\cW_1),\|\cdot\|_{\cW_1})$. By definition $\cW_1(X_k)$ is a bounded sequence in $(\cX,\|\cdot\|_\cX)$ so that $\cW_2\cW_1^{-1}(\cW_1X_k)=\cW_2(X_k)$ has a converging subsequence.
		\end{proof}
		
		\begin{rmk}
			A simple example of compact embedding is provided by classical Sobolev spaces $W^{k,p}(\mathbb{R})$, $k\in\mathbb{N}$ and $1\le p\le \infty$, with the norm $\|f\|_{k,p}\coloneqq \left(\sum_{i=0}^k\|f^{(i)}\|_p^p\right)^{{1}/{p}}$, where $\|f\|_p$ denotes the $L^p$ norm of $f$ with respect to the Lebesgue measure. Quantum extensions of these spaces recently appeared in \cite{Lafleche.2024,Becker.2021}. Here we will use the latter extension which we recall in Section \ref{subsec:bosonic-systems}.  
		\end{rmk}

	\subsection{Strongly continuous semigroups and evolution systems}\label{subsec:c0-semigroups}

		The evolution of a quantum system is often described by a formal differential equation called the master equation. To rigorously study solutions of a {master equation}, the theory of $C_0$-semigroups constitutes an essential toolbox. While detailed expositions to this theory can be found e.g.~in the books \cite[Chap.~II]{Engel.2000}\cite[Chap.~9]{Kato.1995} or \cite[Chap.~X-XIII]{Hille.2012}, here we provide a short overview and introduce concepts that are relevant for the present paper. A family of operators $(P_t)_{t\geq0}\subset\cB(\cX)$ is called a {$C_0$-semigroup} if it satisfies the following properties: 
		\begin{itemize}
			\item[$-$] $P_tP_s=P_{t+s}$ for all $t,s\geq0$;
			\item[$-$] $P_0=\1$, the identity map on $\cX$; and
			\item[$-$] $t\mapsto P_t$ is strongly continuous at $0$.
		\end{itemize}    
		To every $C_0$-semigroup one can associate a linear operator that is in the most general case unbounded, densely defined and closed. This operator determines the semigroup uniquely and is called its generator \cite[Thm.~II.1.4]{Engel.2000}. We will typically denote it by $(\cL, \cD(\cL))$, where
		\begin{equation}
			\cD(\cL) = \{x \in \cX : t \mapsto P_t(x) \text{ differentiable on } \R_+\}
		\end{equation}
		and for $x \in\cD(\cL)$
		\begin{equation}
			\cL(x) = \lim\limits_{t \to 0^+} \frac{1}{t}(P_t(x) - x) \, , 
		\end{equation}
		where the limit is with respect to the topology induced by $\cX$. The semigroup leaves the domain of its generator invariant and further commutes with it on its domain allowing us to justify the well-posedness of the following differential equation on $\cX$
		\begin{equation}\label{eq:time-indep-master-equation}
			\frac{d}{dt} x(t) = \cL(x(t)) \quad x(0) \in \cD(\cL) \quad\text{and}\quad t \ge 0 \, . 
		\end{equation}
		From the above considerations, this equation has a strongly continuous solution given by the semigroup, i.e.~$P_t(x(0)) = x(t)$. Indeed the semigroup is the unique solution (asking for a continuously differentiable map $t \mapsto x(t)$) to this so-called master equation \cite[Prop.~II.6.2]{Engel.2000}. The formulation as a master equation or initial value problem also reveals the origin of the term ``generator'', since for bounded linear operators the solution to these problems is just given by the semigroup involving the exponential of the generator (i.e.~$(e^{t\cL})_{t \ge 0}$). When the operator $\mathcal{L}$ is bounded, the conditions of existence and uniqueness are immediately satisfied using the series expansion of the exponential. For unbounded operators, the answer is no longer straightforward and requires a different representation of the exponential. One possible choice involves the resolvent of the generator.
		
		The well-known generation theorems by Lumer and Phillips, Hille and Yosida and Feller, Miyadera and Phillips all rely on the resolvent satisfying specific bounds, either directly, or indirectly e.g.~by dissipativity of the underlying operator $\cL$. Below we recall the theorems by Hille and Yosida and Lumer and Phillips, as they are going to be used frequently throughout this paper. It is noteworthy that the first two theorems are generalized by the third with the last allowing for the generation of semigroups that satisfy the bound
		\begin{equation*}
			\norm{e^{t \cL}}_{\cX \to \cX} \le c \,e^{\omega t}
		\end{equation*}
		with  $\omega \in \R$ and $c \ge 0$. If $c = 1$, we call the semigroup $\omega$-quasi contractive, and if further $\omega \le 0$ we call it contractive. We start with the generation theorem by Hille and Yosida that gives necessary and sufficient conditions on an operator to be the generator of a contractive $C_0$-semigroups by imposing constraints on its resolvent.
		
		\begin{thm}[Hille-Yosida]\label{thm:hille-yosida}
			A linear operator $(\cL,\cD(\cL))$ on $\cX$ generates a strongly continuous $\omega$-quasi contraction semigroup iff $(\cL,\cD(\cL))$ is closed, densely defined, the resolvent set contains $(\omega, \infty)$ and for all $\lambda \in (\omega, \infty)$ one has 
			\begin{equation*}
				\| R(\lambda,\cL)\|_{\cX\rightarrow\cX}\leq\frac{1}{\lambda - \omega}\,.
			\end{equation*}
		\end{thm}
		
		The other seminal result, which we use in the present paper is a modified formulation of Theorem \ref{thm:hille-yosida} due to Lumer and Phillips \cite[Thm.~II.3.15]{Engel.2000} which, instead of asking for a certain bound on the resolvent, requires certain properties for $\cL$ among which $\omega$-{dissipativity}: 
		
		\begin{defi}\label{def:dissipativity}
			For $\omega \geq 0$, an operator $(\cL,\cD(\cL))$ on $\cX$ is $\omega$-{quasi dissipative} if for all $x \in \cD(\cL)$ and $\lambda > 0$
			\begin{equation*}
				\norm{(\lambda-(\cL-\omega))x}_{\cX} \geq \lambda \norm{x}_{\cX}.
			\end{equation*}
			If $\omega = 0$, we call the operator dissipative. 
		\end{defi}
		In what follows, the notation $\rg$ stands for the range of an operator $(\cL, \cD(\cL))$, defined as $\rg(\cL) = \{\cL(x):x \in \cD(\cL)\}$. Then we can state the theorem in the following way:
		
		\begin{thm}[Lumer-Phillips]\label{thm:lumer-phillips}
			Let $(\cL,\cD(\cL))$ be a densely defined, $\omega$-dissipative operator on $\cX$. Then, the closure $\overline{\cL}$ generates a $\omega$-quasi contraction semigroup iff there exists a $\lambda>0$ such that $\rg(\lambda-(\cL-\omega))$ is dense in $\cX$.
		\end{thm}
		
		\begin{rmk}
			When $\cL$ is dissipative and there is a $\lambda>0$ such that $\rg(\lambda-\cL)$ is dense in $\cX$, then this holds for all $\lambda > 0$ \cite[Prop.~II.3.14]{Engel.2000}.
		\end{rmk}
		
		In the case of a $C_0$-semigroup on a Hilbert space $\cH$, the following result proves useful:
		\begin{prop}\label{cor:lumer-phillips}
			Let $(G,\cD(G))$ be a densely defined linear operator on $\cH$ and assume that $G$ and $G^\dagger$ are $\omega$-quasi dissipative. Then, $\overline{G}$ generates a $\omega$-quasi contraction $C_0$-semigroup on $\cH$. Moreover, if $(G,\cD(G))$ generates a $\omega$-quasi contraction semigroup, $G$ and $G^\dagger$ are $\omega$-quasi dissipative.
		\end{prop}
		
		In this paper, we are also concerned with extensions of the above theory to the setting of time-dependent $C_0$-semigroups. In this case, we refer to the family as a \textit{$C_0$-evolution system}: A two-parameter family of bounded operators $( P_{t,s})_{0\le s\le t}$ is called an \textit{evolution system} if
		\begin{itemize}
			\item[$-$]$P_{t,t}=\1$,
			\item[$-$]$P_{t,r}P_{r,s}=P_{t,s}$ for all $0\le s\le r\le t$, and
			\item[$-$]$(t,s)\mapsto P_{t,s}$ is strongly continuous.
		\end{itemize}
		
		A well-known class of evolution systems is given by $C_0$-semigroup after imposing $P_{t,s}=P_{t-s}$. A subtle difference between semigroups and evolution systems is that the latter are not necessarily differentiable for any $x\neq 0$ \cites[p.~478]{Engel.2000}. Here, we recall sufficient conditions under which the following master equation admits a unique \textit{solution operator}:
		\begin{equation}\label{eq:time-dep-master-equation}
			\frac{\partial}{\partial t}x(t)=\cL_t(x(t))\quad\text{and}\quad x(s)=x_s\quad\text{ for }0\leq s\leq t\,.
		\end{equation}
		We start with a set of assumptions often referred to as being of \textit{hyperbolic type} \cites[Chap.~5]{Pazy.1983}[pp.~127~ff.]{Prato.2011} and which allow for the generation of an evolution system starting from a $C_0$-semigroup. Here, the existence of a so-called admissible subspace plays an important role:
		\begin{defi}[Admissible subspaces]\label{defi:admissible-spaces}
			For a $C_0$-semigroup $(P_t)_{t\geq0}\subset\cB(\cX)$ with generator $(\cL,\cD(\cL))$, a subspace $(\cY\subset\cX,\|\cdot\|_{\cY})$ is called admissible for $(P_t)_{t\ge 0}$, or simply $\cL$-admissible if $\cY$ is an invariant closed subspace of the semigroup, i.e.~$e^{t\cL}\cY\subset\cY$, and $e^{t\cL}|_{\cY}$ defines a $C_0$-semigroup on $(\cY,\|\cdot\|_\cY)$. Similarly, $(\cY,\|\cdot\|_{\cY})$ is an admissible subspace of an evolution system $P_{0\leq s\leq t}$ if it is an invariant closed subspace of the evolution system and $(P_{0\leq s\leq t}|_{\cY})_{0\leq s \leq t}$ defines an evolution system on $(\cY,\|\cdot\|_\cY)$.  
		\end{defi}
		Our first basic assumption is that for every fixed $s$ the operator $(\cL_s, \cD(\cL_s))$ generates a $C_0$-semigroup. Moreover,
		\begin{itemize}
			\item[$(1)$\hspace{1ex}] $(\cL_s)_{ s\ge 0}$ is a \textit{stable} family, i.e.~there is $c\ge 0$ and $\omega\in\R$ such that $\|e^{t\cL_s}\|_{\cX\rightarrow\cX}\leq c \,e^{\omega t}$ for all $ s\ge 0$;
			\item[$(2)$\hspace{1ex}] There exists a subspace $\cY\subset \cX$ and a norm $\|\cdot\|_\cY$ on $\cY$ endowing $\cY$ with a Banach space structure, such that for all $ s\ge 0$, $(\cY,\|\cdot\|_\cY)$ is \textit{$\cL_s$-admissible} and $(\cL_s)_{s\ge 0}$ is stable on $(\cY,\|\cdot\|_\cY)$\,;
			\item[$(3)$\hspace{1ex}] Finally, the map $s\mapsto\cL_s\in\cB((\cY,\|\cdot\|_\cY),(\cX,\|.\|_\cX))$ is uniformly continuous.
		\end{itemize}
		Under these assumptions, \cite[Thm.~3.1]{Pazy.1983} shows the existence of a unique evolution system $(P_{t,s})_{0\le s\le t}$. If one further requests this evolution system to have the following properties
		\begin{itemize}
			\item[$(4)$] $P_{t,s} \cY \subseteq \cY$ for $0\le s\le t$; 
			\item[$(5)$] $(s, t) \mapsto P_{s,t}$ is strongly continuous on $(\cY, \norm{\cdot}_{\cY})$;
		\end{itemize}
		one obtains the following theorem:
		\begin{thm}[Time-dependent semigroups \texorpdfstring{\cite[Thm.~3.1,~4.3]{Pazy.1983}}{}]\label{thm:time-dependent-semigroups}
			Let $\{(\cL_s,\cD(\cL_s)\}_{s\ge 0}$ be a family of generators of $C_0$-semigroups, which satisfy assumption $(1-3)$. Then, there exists a unique evolution system which satisfies
			\begin{itemize}
				\item[$-$] $\|P_{t,s}\|_{\cX\rightarrow\cX}\leq c\,e^{(t-s)\omega}$ for all $0\le s\le t$;
				\item[$-$] $\lim_{t\downarrow s}\frac{1}{t-s}(P_{t,s}x-x)=\cL_sx$ for all $x\in \cY$; and
				\item[$-$]$\frac{\partial}{\partial s}P_{t,s}x=-P_{t,s}\cL_sx$ for all $x\in\cY$ and $0\le s\le t$.
			\end{itemize}
			The two limits above are both with respect to the topology induced by $\|.\|_\cX$. If further (4) and (5) hold then for every $v\in \cY$, $P_{t,s}v$ is a unique solution for the initial value problem in \Cref{eq:time-dep-master-equation} in $(\cY,\|.\|_\cX)$.
		\end{thm}
		
		\begin{rmk*}[Kato's \texorpdfstring{$C^1$}{C1}-condition]
			For a time-independent domain $\cD$, the above conditions follow if $\{(\cL_s,\cD)\}_{ s\ge 0}$ is a stable family of generators and if $s\mapsto\cL_s$ is strongly continuously differentiable w.r.t.~$\|.\|_\cX$ \cite[Chap.~5, Thm.~4.8]{Pazy.1983}.
		\end{rmk*}
		
		Finally, we discuss the general idea of how to control perturbed semigroups through certain admissible subsets associated with the domain of an invertible operator. More explicit variants are postponed to Section \ref{sec:example-perturbation-bounds}. 
		\begin{thm}\label{thm:semigroup-perturbation}
			Let $(\cL,\cD(\cL))$ and $(\cL+ \cK,\cD(\cL+\cK))$ be two generators of $C_0$-semigroups on $\cX$, for an operator $(\cK,\cD(\cK))$. Moreover, let $(\cW,\cD(\cW))$ be an invertible operator on $\cX$ with bounded inverse, such that $\cD(\cW)$ is an $\cL+\cK$-admissible subspace {(see Definition \ref{defi:admissible-spaces})} and such that $\cK\cW^{-1}$ is bounded. Then, for all $t\ge 0$,
			\begin{equation*}
				\|e^{t\cL}-e^{t(\cL+\cK)}:\cW\to\cX\|\leq {t} \,\|\cK\cW^{-1}\|_{\cX\to \cX}\; \int_{0}^1\| e^{(1-s)t\cL}\|_{\cX\to \cX}\;\|e^{st(\cL+\cK)}\|_{\cW\to \cW}\;ds\,.
			\end{equation*}
			In particular, for all $t\ge 0$  and $x\in \cD(\cW)$, the following equation holds in the Bochner sense:
			\begin{equation*}
				(e^{t\cL}-e^{t(\cL+\cK)})x= {t} \int_{0}^1 e^{(1-s)t\cL} \cK e^{st(\cL+\cK)}x\;ds.
			\end{equation*}
		\end{thm}
		\begin{proof}
			See Theorem \ref{thm-appx:semigroup-perturbation}.
		\end{proof}
		\begin{rmk}
			In words, Theorem \ref{thm:semigroup-perturbation} shows that the integral equation for semigroups is well-defined by generalizing the standard method that requires the following relative boundedness condition (see e.g.~\cite[Chapter 2]{Kato.1995}): for all $x\in \cD(\cL+\cK)$, $x\in\cD(\cK)$ and
			\begin{equation*}
				\|\cK x\|_{\cX}\leq \|(\cL+\cK) x\|_{\cX}.
			\end{equation*}
			Indeed, the choice $\cW\coloneqq \1-(\cL+\cK)$ with the resolvent $R(1,\cL+\cK)$ as its bounded inverse shows that our scheme is a generalization of the above. Clearly, $\cW$ generates an admissible subspace and $\|\cK \cW^{-1}x\|_{\cX}\leq \|x\|_{\cX}$ shows the implication. Note also that the above bound can be extended to evolution systems,
		\end{rmk}

	\subsection{Continuous variable quantum systems}\label{subsec:bosonic-systems}

		 An important feature in quantum physics is that of the indistinguishability of particles which results in the introduction of Bosonic and Fermionic particles \cite[Chap.~5.2]{Bratteli.1981}. In the second quantization formalism, a Bosonic or continuous variable quantum system can be described by the Fock space $\cH=L^2(\mathbb{R})$ endowed with an orthonormal (Fock) basis $\{\ket{n}\}_{n=0}^\infty$, where $n$ labels the number of photons present in a given mode. The space of vectors with finite support is denoted by $\cH_f = \{\ket{\psi}\in\cH: \exists M\in\N\,:\,\ket{\psi}= \sum_{n=0}^M \braket{n\,,\psi}\ket{n}\}$, where $\langle \varphi,\psi\rangle$ denotes the standard inner product on $L^2(\mathbb{R})$. Next, we define the \textit{annihilation} and \textit{creation} operators through the following relations
		 \begin{equation*}
		 	a\ket{n}=\sqrt{n}\ket{n-1},\quad a\ket{0}=0,\quad \ad\ket{n}=\sqrt{n+1}\ket{n+1}\,.
		 \end{equation*}
		 The operators $a$ and $a^\dagger$ satisfy the \textit{canonical commutation relation} (CCR), i.e.~$[a,\ad]=\1$ on $\cH_f$. We can construct the number operator of the latter two as,
		 \begin{equation}\label{eq:number-operator}
		 	\Nind = \ad a = \sum_{n=0}^\infty n\ketbra{n}{n} \,.
		 \end{equation}
		 It counts the number of photons in a mode. All of $a$, $\ad$, and $N$, although linear, are unbounded operators, hence are only defined on a (dense) subset of $\cH$, namely
		 \begin{equation}
		 	\cD(a^\dagger)=\cD(a)=\{\ket{\phi}\in\cH:\|a\ket{\phi}\|<\infty\}=\{\ket{\phi}=\sum_{n=0}^\infty\lambda_n\ket{n}:\sum_{n=0}^\infty n|\lambda_n|^2<\infty\}=\cD(N^{\frac{1}{2}})\,.
		 \end{equation}
		 Note that the above domains are {maximal}, i.e.~$\cD(a)=\{\ket{\phi}\in\cH:a\ket{\phi}\in \cH\}$.
		 In most parts of the paper, we consider operators constructed by polynomials $p\in\C[X,Y]$ of $a$, $\ad$ where the variables $X$ and $Y$ are considered non-commuting, i.e.~$XY$ is a different polynomial then $YX$. Note the $\C[X,Y]$ is the polynomial ring in $X,Y$ over the complex field. Using the CCR, we can always assume without loss of generality that the polynomial has the following normal form: there exist complex coefficients $\lambda_{ij}$ and $\mu_{kl}$ such that
		 \begin{equation}\label{eq:ccr-polynomial-representation}
		 	p(a\,,\ad)=\sum_{i + 2j \le \deg(p)}\lambda_{ij}(\ad)^iN^j+\,\sum_{k + 2l\le \deg(p)}\mu_{kl}N^la^k\,.
		 \end{equation}
		 One possible domain of these operators can be described by the degree $d$ of $p$ (see Section \ref{sec-appx:annihilation-creation}):
		 \begin{equation}\label{eq:domain-ccr-polynomial}
		 	\cD(p(a,\ad))=\cD(N^{d/2}).
		 \end{equation}
		 Next, we add the number operator to the power of twice the leading order to the polynomial, i.e.
		 \begin{equation}\label{eq:polynomial+number-op}
		 	\Tilde{p}(a,\ad)\coloneqq (N+\1)^{2d}+p(a,\ad)
		 \end{equation}
		 which allows us to show that the domain is maximal in the sense that the operator is closed. Note that he choice of $(N+\1)^{2d}$ is adapted to the proof of Lemma \ref{lem:semigroup-of-G-for-positivity}. Moreover, we prove that $\cH_f$ is a core for the considered polynomial (cf.~\cite[Sec.~7.1]{Simon.2015}).
		 
		 \begin{lem}[Adjoint and core of polynomials of $a, \ad$]\label{lem:formal-polynomial-ccr-adjoint-core}
		 	Let $p\in\C[X,Y]$ be a polynomial on $\C$ and $(p(a,\ad),\cD(N^{d/2}))$ the unbounded operator in normal form \eqref{eq:domain-ccr-polynomial}. Then, $p(a,\ad)$ is closable and there is a $c\geq0$ such that for all $\phi\in\cD(N^{d/2})$
		 	\begin{equation*}
		 		\|p(a,\ad)\ket{\phi}\|\leq c\|(\1+N)^{d/2}\ket{\phi}\|\,.
		 	\end{equation*}
		 	The modification $\tilde{p}(a,\ad)=(N+\1)^{2d}+p(a,\ad)$ is a closed operator with domain $\cD(\Tilde{p}(a,\ad))=\cD(\Tilde{p}(a,\ad)^\dagger)=\cD(N^{2d})$ and core $\cH_f$.
		 \end{lem}
		 \begin{proof}
		 	See Lemma \ref{lem-appx:formal-polynomial-ccr-adjoint-core}.
		 \end{proof}
		 
		 \begin{rmk*}
		 	The above lemma will allow us to reduce the analysis of the unbounded operator $p(a,\ad)$ in the strong topology to that on finite sums.  
		 \end{rmk*}
		 
		 We end this preliminary section by introducing a family of weighted Banach spaces which we coin as \textit{Bosonic Sobolev spaces} in analogy with classical harmonic analysis. The Bosonic Sobolev space of order $k\in {\mathbb{R}_+}$ is defined on
		 \begin{equation*}
		 	\cD(\cW^k)=\{(\cW^k)^{-1}(x) \in\cT_{1,\,\operatorname{sa}} \;:\; x \in \cT_{1,\,\operatorname{sa}}\}
		 \end{equation*}
		 via
		 \begin{equation}\label{eq:bosonic-symmetric-weight}
		 	\cW^k(x)\coloneqq(\1+N)^{k/4} x (\1+N)^{k/4}\,.
		 \end{equation}
		 
		 Since the inverse $(\cW^k)^{-1}(x)=(\1+N)^{-k/4} x (\1+N)^{-k/4}$ is a bounded operator, $(\cD(\cW^k), \|\cdot\|_{\cW^k})$ is a Banach space by Theorem \ref{thm:weighted-banach-space}. For the sake of notation, we define
		 \begin{equation}\label{eq:sobolev-space}
		 	W^{k,1}\coloneqq\cD(\cW^k)\quad\text{and}\quad\|\cdot\|_{W^{k,1}}\coloneqq\|\cdot\|_{\cW^k}.
		 \end{equation}
		 For $k=0$, $\cW^k=\1$ and $(\cD(\cW^0),\|\cdot\|_{W^{0,1}})=(\cT_{1,\,\operatorname{sa}},\|\cdot\|_1)$. 
		 
		 \begin{lem}\label{lem:sobolev-embedding}
		 	Let ${k < k' \in \mathbb{R}_+ := [0, \infty)}$. Then, 
		 	\begin{equation*}
		 		W^{k',1}\Subset W^{k,1}.
		 	\end{equation*}
		 \end{lem}
		 \begin{proof}
		 	The proof relies on the abstract Theorem \ref{thm:compact-embedding-weighted-spaces}, which is applied for different values $k\in{\R_+}$ of 
		 	\begin{equation*}
		 		\cW^k(x)\coloneqq(\1+N)^{k/4} x (\1+N)^{k/4}.
		 	\end{equation*}
		 	For $k'>k$, the operator $\cW^{k}\cW^{-k'}=\cW^{k-k'}$ is bounded. Next we show compactness by proving that $\cW^{-l}$ with $-l=k-k'$ is approximated by a sequence of finite rank operators:
		 	\begin{equation*}
		 		\cW^{-l}_{f,M}(x)\coloneqq\sum_n^M(1+n)^{-l/4}\ketbra{n}{n} x \sum_m^M(1+m)^{-l/4}\ketbra{m}{m},
		 	\end{equation*}
		 	which can be deduced through Hölder's inequality
		 	\begin{align*}
		 		\|\cW^{-l}(x)-\cW^{-l}_{f,M}(x)\|_1&=\|\sum_{m,n>M}(1+n)^{-l/4}\ketbra{n}{n} x (1+m)^{-l/4}\ketbra{m}{m}\|_1\leq M^{-l/2}\|x\|_1.
		 	\end{align*}
		 	Since finite rank operators are compact by the Bolzano-Weierstrass theorem and the limit is a compact operator again \cite[Thm.~2.13.4]{Hille.2012}, the operator $\cW^{-l}$ is a compact operator on $\cT_{1,\,\operatorname{sa}}$. Applying Theorem \ref{thm:compact-embedding-weighted-spaces} shows that 
		 	\begin{equation*}
		 		W^{k',1}\Subset W^{k,1}\,.\vspace{-2ex}
		 	\end{equation*}
		 \end{proof}
		 
		 The following theorem will become helpful later and has an analogue in the theory of commutative $L_p$ spaces, which inspired its name. Although it can be proved by interpolation theory, we provide a more rudimentary approach using only Hadamard's three-line theorem.
		 
		 \begin{thm}[Stein-Weiss theorem for Bosonic Sobolev spaces]\label{thm:stein-weiss}
		 	Let $k_0 < k_1 \in \R_+$ and $T: W^{k_j, 1} \to W^{k_j, 1}$, be a linear map with $\norm{T}_{W^{k_j, 1} \to W^{k_j, 1}} \le M_j$ for some $M_j \ge 0$, $j = 1, 2$. Then for $\theta \in [0, 1]$, $T: W^{k_\theta, 1} \to W^{k_\theta, 1}$ with $k_\theta = (1-\theta) k_0 + \theta  k_1$ obtained by restriction of the input $T:W^{k_0, 1} \to W^{k_0, 1}$ to $W^{k_\theta, 1} \cap W^{k_0, 1}$ is a well defined bounded linear map with
		 	\begin{equation}
		 		\norm{T}_{W^{k_\theta, 1} \to W^{k_\theta, 1}} \le M_0^{1-\theta} M_1^{ \theta} \, . 
		 	\end{equation}
		 \end{thm}
		 \begin{proof}
		 	We have that $k_0 < k_1$ and hence $W^{k_1, 1} \Subset W^{k_0, 1}$. Let $k_\theta = (1-\theta) k_0 + \theta k_1$ with $\theta \in (0, 1)$ and $x \in \cT_f \subseteq W^{k_1, 1} \Subset W^{k_\theta, 1} \Subset W^{k_0, 1}$. We will show
		 	\begin{equation}
		 		\norm{T(x)}_{W^{k_\theta, 1}} \le M_0^{1-\theta} M_1^{\theta} \norm{x}_{W^{k_\theta, 1}}
		 	\end{equation}
		 	which proves that $T$ can be uniquely extended to a bounded linear map on $W^{k_\theta, 1}$ that agrees with the restriction of $T:W^{k_0, 1} \to W^{k_0, 1}$, to $W^{k_\theta, 1} \cap W^{k_0, 1}$. This agreement on intersections is due to the compact embeddings of the Bosonic Sobolev spaces into one another. For $x \in \cT_f$ and $z \in S \coloneqq \{z \in \C \;:\; 0 \le \Re(z) \le 1\}$, we define $k(z) = (1-z) k_0 + z k_1$ so that for a fixed $Z \in \cB(\cH)$ with $\norm{Z}_\infty \le 1$,
		 	\begin{equation}
		 		\begin{aligned}
		 			g: S &\to \C,\,g(z) = \tr[(N + \1)^{\frac{k(z)}{4}} T\Big((N + \1)^{\frac{k_\theta - k(z)}{4}} x (N + \1)^{\frac{k_\theta - k(z)}{4}}\Big) (N + \1)^{\frac{k(z)}{4}} Z]
		 		\end{aligned}
		 	\end{equation}
		 	is well-defined, uniformly bounded, and continuous on $S$ (q.v.~Lemma \ref{lem:continuity-G}) and further holomorphic on $\mathring{S} \coloneqq \{z \in \C \;:\; 0 < \Re(z) < 1\}$ (q.v.~Lemma \ref{lem:differentiability-G}). Note that for $\theta\in(0,1)$
		 	\begin{equation}
		 		|g(\theta)| = \left|\tr[(N + \1)^{\frac{k_\theta}{4}} T(x) (N + \1)^{\frac{k_\theta}{4}} Z]\right| \, ,
		 	\end{equation}
		 	for $t\in\R$
		 	\begin{equation}\label{eq:g(it)}
		 		\begin{aligned}
		 			|g(it)| &= \left|\tr\left[(N + \1)^{\frac{k_0 + it(k_1 - k_0)}{4}}T\Big((N + \1)^{\frac{k_\theta - k(it)}{4}} x (N + \1)^{\frac{k_\theta - k(it)}{4}}\Big) (N + \1)^{\frac{k_0 + it(k_1 - k_0)}{4}}Z\right]\right|\\
		 			&\le \norm{T\Big((N + \1)^{\frac{k_\theta - k_0 + it(k_0 - k_1)}{4}} x (N + \1)^{\frac{k_\theta - k_0 + it(k_0 - k_1)}{4}}\Big)}_{W^{k_0, 1}} \norm{Z}_\infty\\
		 			&\le \norm{T}_{W^{{k_0}, 1} \to W^{{k_0}, 1}} \norm{(N + \1)^{\frac{k_\theta - k_0}{4}} x (N + \1)^{\frac{k_\theta - k_0}{4}}}_{W^{k_0, 1}}\\
		 			&= M_0 \norm{x}_{W^{k_\theta, 1}}\,,
		 		\end{aligned}
		 	\end{equation}
		 	and similarly
		 	\begin{equation}\label{eq:g(it + 1)}
		 		\begin{aligned}
		 			|g(1 + it)| & \le \norm{T\Big((N + \1)^{\frac{k_\theta - k_1 + it(k_0 - k_1)}{4}} x (N + \1)^{\frac{k_\theta - k_1 + it(k_0 - k_1)}{4}}\Big)}_{W^{k_1, 1}}\norm{Z}_\infty\\
		 			&\le \norm{T}_{W^{k_1, 1} \to W^{k_1, 1}} \norm{(N + \1)^{\frac{k_\theta - k_1 + it(k_0 - k_1)}{4}} x (N + \1)^{\frac{k_\theta - k_1 + it(k_0 - k_1)}{4}}}_{W^{k_1, 1}}\\
		 			&=M_1 \norm{x}_{W^{k_\theta, 1}} \, . 
		 		\end{aligned}
		 	\end{equation}
		 	An application of Hadamard's three-lines theorem now gives us that for $G_0 \coloneqq \sup\limits_{t \in \R} |g(it)|$ and $G_1 \coloneqq \sup\limits_{t \in \R} |g(it + 1)|$
		 	\begin{equation}
		 		|g(\theta)| \le G_0^{1-\theta} G_1^{ \theta} \le M_0^{1-\theta} M_1^{ \theta} \norm{x}_{W^{k_\theta, 1}} \, ,
		 	\end{equation}
		 	where the last inequality follows from the bounds in \Cref{eq:g(it)} and \Cref{eq:g(it + 1)}. Since $Z$ was arbitrary, we can deduce that 
		 	\begin{equation}
		 		\begin{aligned}
		 			\norm{T(x)}_{W^{k_\theta, 1}} &= \sup\left\{ \left|\tr[(N + \1)^{\frac{k_\theta}{4}} T(x) (N + \1)^{\frac{k_\theta}{4}} Z ]\right| \;:\; \norm{Z}_\infty \le 1 \right\} \\
		 			&\le M_0^{1-\theta} M_1^{ \theta} \norm{x}_{W^{k_\theta, 1}}
		 		\end{aligned}
		 	\end{equation}
		 	where we used the dual characterisation of $\norm{\cdot}_1$. This concludes the claim.
		 \end{proof}
		 
		 With the Sobolev embedding at hand, we introduce the notion of a \textit{Sobolev preserving semigroup} as a semigroup defined on a sequence of Bosonic Sobolev spaces.
		 
		 \begin{defi}[Sobolev preserving evolution system]\label{defi:sobolev-preserving-semigroups}
		 	Let $(\cP_t)_{t\ge 0}$ be a $C_0$-semigroup on $\cT_{1,\operatorname{sa}}$. We then call $(\cP_t)_{t\ge 0}$ \textit{Sobolev preserving} if there exists a divergent sequence $\{k_r\}_{r \in \N} \to \infty$, such that for all $r \in \N$, $W^{k_r, 1}$ is an admissible subspace for $(\cP_t)_{t\ge 0}$. Similarly for $(\cP_{t,s})_{0\leq s\leq t}$ an evolution system on $\cT_{1,\operatorname{sa}}$, we call $(\cP_{t,s})_{0\leq s\leq t}$ \textit{Sobolev preserving} if for all $r \in \N$, $W^{k_r, 1}$ is admissible for $(\cP_{t,s})_{0\leq s\leq t}$ to $W^{k_r, 1}$.
		 \end{defi}
		 Note that with the Stein-Weiss theorem for Bosonic Sobolev spaces, Theorem \ref{thm:stein-weiss}, one can immediately interpolate a semigroup and evolution system defined on $W^{k_0, 1}$ and $W^{k_1, 1}$ to $W^{k_\theta, 1}$ with $k_\theta = (1 - \theta)k_0  + \theta k_1$, $\theta \in [0, 1]$ (q.v.~Lemma \ref{lem:interpolation-lemma}). This means the above definition is equivalent to the definition, requiring that for all $k \in \R_+$, $W^{k, 1}$ are admissible subspaces of the semigroup or evolution system respectively.
		 
		 The following example shows that not every semigroup is Sobolev preserving:
		 \begin{ex}
		 	An example of a $C_0$-semigroup that is not Sobolev preserving is the depolarizing semigroup, i.e.~$\cP_t(\rho) = e^{-t} \rho + (1 - e^{-t})\tr[\rho] \sigma$ where $\sigma$ is a quantum state with $\tr[N^{\frac{1}{2}}\sigma N^{\frac{1}{2}}] = \infty$. Then for a quantum state $\rho \in W^{2, 1}$ we find that 
		 	\begin{equation*}
		 		\norm{\cP_t(\rho)}_{W^{2, 1}} = \infty \quad \qquad \forall t > 0 \, .
		 	\end{equation*}
		 \end{ex}
		 
		 \section{Sobolev preserving quantum Markov semigroups}\label{sec:polynomial-generators}
		 A quantum evolution in bosonic systems is described by a master equation 
		 \begin{equation}\label{eq:mastereq}
		 	\frac{d}{dt} x(t) = \cL(x(t)) \quad x(0) \in \cD(\cL) \quad\text{and}\quad t \ge 0 \, . 
		 \end{equation}
		 where $\cL$ is potentially unbounded. In the following, we state two sufficient assumptions for the existence and uniqueness of an operator-valued solution to \eqref{eq:mastereq} in terms of a semigroup. In other words, we prove a generation theorem for bosonic quantum Markov semigroups. This is generalized in Section \ref{sec:timedependentgeneration} to the case of time-dependent generators.

	\subsection{Strongly continuous bosonic semigroups}\label{subsec:time-indep-generation}

		We start with the time-independent setting, for which we need two working assumptions. The first assumption is motivated by the so-called GKSL \cite{Lindblad.1976,Gorini.1976} form that generators of quantum dynamical semigroups over finite-dimensional quantum systems take, as well as our natural choice to consider jump and Hamiltonian operators described by polynomials in the annihilation and creation operators:
		
		\begin{assum}\label{assum:finite-degree}
			The operator $(\cL,\cT_f)$ has $\operatorname{GKSL}$ form, i.e.~for $x\in\cT_f$ 
			\begin{equation}\label{eq:lindblad}
				\begin{aligned}
					\cL:\cT_f \to \cT_f \quad x \mapsto\cL(x) &= - i [H, x] + \sum\limits_{j = 1}^K L_j x L_j^\dagger  - \frac{1}{2}\{L_j^\dagger L_j, x\} \\
					&\coloneqq Gx + x G^\dagger + \sum\limits_{j = 1}^K L_j x L_j^\dagger\, , 
				\end{aligned}
			\end{equation}
			for some $K\in\mathbb{N}$ and with $G=-iH-\frac{1}{2}\sum_{j=1}^K L_j^\dagger L_j$, where $\{A,B\}=AB+BA$ denotes the anticommutator of two operators $A,B$ on a suitable domain. For the above equation to make sense, the operators $H$ and $L_j$ are assumed to be polynomials of the creation and annihilation operators, i.e.~$H\coloneqq p_H(a,\ad)$ and $L_j\coloneqq p_j(a,\ad)$, and $H$ is assumed to be symmetric. This ensures that $\cT_f$ is invariant under $\cL$.
			We denote the degree of $p_H$ by $d_H\coloneqq\deg(p_H)$, those of $p_j$ by $d_j\coloneqq\deg(p_j)$, and $d\coloneqq\max\{d_1,...,d_K,d_H\}$.
		\end{assum}
		
		The second assumption will lead to the semigroup being Sobolev preserving, which allows us not only to prove the existence and uniqueness of the evolution generated by \eqref{eq:mastereq} but further to conduct a perturbation analysis as well as to extend our results to the case of a time-dependent Master equation:
		
		\begin{assum}\label{assum:sobolev-stability}
			There exists a non-negative sequence $\{k_r\}_{r \in \N}\rightarrow\infty$ s.t.~for all $r \in \N$ there exist $\omega_{k_r} \ge 0$ such that for all positive semi-definite $x \in \cT_f$
			\begin{equation}\label{eq:Assumptionsobolevstability}
				\tr[\cL(x) (N + \1)^{k_r/2}] \le \omega_{k_r} \tr[x (N + \1)^{k_r/2}] \, .
			\end{equation}
		\end{assum}
		
		We are now ready to state and prove the main theorem of the section: 
		\begin{thm}[Generation of bosonic semigroups]\label{thm:generation-theorem}
			Let $(\cL, \cD(\cL))$ be an operator defined on $\cT_{1,\operatorname{sa}}$. If $(\cL, \cD(\cL))$ satisfies Assumption \ref{assum:finite-degree} and Assumption \ref{assum:sobolev-stability}, then the closure $\overline{\cL}$ generates a strongly continuous, positivity preserving semigroup $(\cP_t)_{t\ge 0}$ on $W^{k, 1}$ for all $k\geq 0$ with 
			\begin{equation}\label{eq:sobolev-bound}
				\norm{\cP_t}_{W^{k, 1} \to W^{k, 1}} \le e^{\omega_k t} \,\quad \forall t\ge 0\, . 
			\end{equation}
			where $\omega_k = \frac{k_{r_1} - k}{k_{r_1} - k_{r_0}}\omega_{k_{r_0}} + \frac{k - k_{r_0}}{k_{r_1} - k_{r_0}}\omega_{k_{r_1}}$ for an $r$ such that $k_{r_0}\leq k <k_{r_1}$. Finally, for $k = 0$, the semigroup is contractive and trace-preserving.
		\end{thm}
		
		\begin{rmk}
			For the existence and well-posedness of a semigroup between $k$ and $0$, Assumption \ref{assum:sobolev-stability} with $\tilde k \geq k + 4d$ would be sufficient. However, in all of the examples we found, Assumption \ref{assum:sobolev-stability} was either fully satisfied or completely violated under the precondition of Assumption \ref{assum:finite-degree}. Since we intend to later perform perturbation analysis, it is more convenient to adopt the stricter condition, allowing us to compare semigroups with very different degrees. Similarly, one could also weaken the requirements for Theorem \ref{thm:timedep-generation-theorem}, Theorem \ref{thm:multi-mode-generation-theorem}, and Theorem \ref{thm:multi-mode-timedep-generation-theorem}.
		\end{rmk}
		
		Before proving Theorem \ref{thm:generation-theorem}, we provide an example for which Assumption \ref{assum:sobolev-stability} is not satisfied.
		
		\begin{ex}[Pure birth process \texorpdfstring{\cite[Ex.~3.3.]{Davies.1977}}{[9,Ex.~3.3.]}]\label{eq:no-sobolev-stability}
			Let $L=(\ad)^2$ and $G=-\frac{1}{2}a^2(\ad)^2$, i.e.~
			\begin{equation*}
				\cL(x)=Gx+x G^\dagger+Lx L^\dagger\,.
			\end{equation*}
			By construction, this generator satisfies Assumption \ref{assum:finite-degree}. However, one can show that it is not trace-preserving, and therefore it cannot satisfy Assumption \ref{assum:sobolev-stability}.
		\end{ex}
		
		\begin{rmk*}
			Note that Assumption \ref{assum:sobolev-stability} implies the conditions in \cite{Chebotarev.2003,Chebotarev.1998} under the precondition that the closure of $(\cL,\cD(\cL)$ is a generator of a $C_0$-semigroup. Then a similar bound to \Cref{eq:sobolev-bound} can be shown. 
		\end{rmk*}
		
		\paragraph{\emph{Proof strategy}:} Our proof is partly inspired by \cite{Davies.1977}, however, Assumption \ref{assum:sobolev-stability} allows us to go beyond minimal semigroups and obtain a trace-preserving evolution. We believe that a similar result can be obtained by using the tools in \cite{Fagnola.1999}. An important intermediate step is that the considered generators are also generators on $W^{k,1}$, which will allow us to provide simple perturbation analysis on specific examples in Section \ref{sec:example-perturbation-bounds}. Our proof starts with Lemma \ref{lem:semigroup-of-G-for-positivity}, where we show that $G_\varepsilon\coloneqq G-\varepsilon (N+\1)^{4d}$ is a generator on the Fock space and the implemented semigroup $t\mapsto e^{tG_\varepsilon}\cdot e^{tG_\varepsilon^\dagger}$ admits $\cT_f$ as a core. Then Lemma \ref{lem:semigroup-of-G-for-sobolev-stability} extends the result to semigroups on $W^{k, 1}$ for all $k \in \R_+$. By the compact embedding lemma \ref{thm:compact-embedding-weighted-spaces} for $W^{k, 1}$ in $\cT_{1,\operatorname{sa}}$, we can transfer these properties to the unperturbed evolution. Next, we more closely follow the method introduced in \cite{Davies.1977}. In particular, we prove that a perturbed version of \Cref{eq:lindblad} generates a Sobolev and positivity preserving $C_0$-semigroup.
		
		\begin{lem}\label{lem:semigroup-of-G-for-positivity}
			For $\varepsilon > 0$, the closure of the operator
			\begin{equation*}
				\cG_\varepsilon: \cT_f \to \cT_f, \qquad x \mapsto G x + x G^\dagger - \varepsilon \{(N + \1)^{4d}, x\}\,,
			\end{equation*}
			where $G$ is defined in \Cref{eq:lindblad}, generates a strongly continuous, contractive, positivity preserving semigroup on $\cT_{1,\operatorname{sa}}$.
		\end{lem}
		\begin{proof}
			The proof is structured in the following two steps: 
			\begin{itemize}
				\item[1)] The closure of $G_\varepsilon:\cH_f \to \cH$, $\ket{\psi} \mapsto G_\varepsilon\ket{\psi} \coloneqq (-\varepsilon(N + \1)^{4d} + G)\ket{\psi}$ generates a strongly continuous contractive semigroup on $\cH$, which we denote by $(P_t^\varepsilon)_{t \ge 0}$.
				\item[2)] The family of maps $(\cP_t^\varepsilon \coloneqq P_t^\varepsilon \cdot (P_t^\varepsilon)^\dagger:\cT_{1,\operatorname{sa}} \to \cT_{1,\operatorname{sa}})_{t \ge 0}$, with $(P_t^\varepsilon)_{t \ge 0}$ from step 1, defines a strongly continuous, contractive, positivity preserving semigroup on $\cT_{1,\operatorname{sa}}$ generated by the closure of $\cG_\varepsilon$.
			\end{itemize}
			
			\textit{Step 1)} By Assumption \ref{assum:finite-degree} there exists $p_\varepsilon \in \C[X,Y]$ such that $G_\varepsilon = p_\varepsilon(a, a^\dagger)$ which shows by Lemma \ref{lem:formal-polynomial-ccr-adjoint-core} that $G_\varepsilon$ is closed with domain $\cD(N^{4d})$. We will now show dissipativity for $G_\varepsilon$ and $G_\varepsilon^\dagger$ to conclude the claim using Corollary \ref{cor:lumer-phillips}. It suffices also to consider $G^\dagger_\varepsilon$ on $\cH_f$, as it is a core by Lemma \ref{lem:formal-polynomial-ccr-adjoint-core} ($\deg(G)=2d$) and therefore dissipativity of $G_\varepsilon^\dagger$ on $\cH_f$ directly implies dissipativity of $G^\dagger_\varepsilon$ on all of its domain. We only show the dissipativity of $G_\varepsilon$, since the proof for $G_\varepsilon^\dagger$ is completely analogous. Let $\ket{\psi}\in \cH_f$, then for any $\lambda>0$
			\begin{align*}
				\|(\lambda - G_{\varepsilon})& \ket{\psi}\|^2 \\
                &= \lambda^2 \braket{\psi\,,\psi} + \braket{G_\varepsilon\psi, G_\varepsilon\psi} - \lambda(\braket{G_\varepsilon\psi,\psi} + \braket{\psi, G_\varepsilon\psi})\\
				&\ge\lambda^2 \braket{\psi\,,\psi} - \lambda(\braket{G_\varepsilon\psi,\psi} + \braket{\psi, G_\varepsilon\psi})\\
				&\ge \lambda^2 \braket{\psi\,,\psi} + \varepsilon \lambda (\braket{(N + \1)^{4d}\psi,\psi} + \braket{(N + \1)^{4d}\psi,\psi}) - \lambda(\braket{G\psi,\psi} + \braket{\psi, G\psi})\\
				&\ge \lambda^2 \braket{\psi\,,\psi} - \lambda(\braket{G\psi,\psi} + \braket{\psi, G\psi}) \, . 
			\end{align*}
			By the requirement of Assumption \ref{assum:finite-degree}, it is clear that $\tr[\cL(x)] = 0$ for $x \in \cT_f$, using the cyclicity of the trace. For the explicit case of a pure state $x = \ketbra{\psi}{\psi} \in \cT_f$ where the last inclusion holds due to $\ket{\psi} \in \cH_f$, we get
			\begin{equation*}
				\lambda\braket{\psi, -(G^\dagger + G) {\psi}} = \lambda \sum\limits_{j = 1}^K \braket{L_j\psi, L_j {\psi}} \ge 0 \, .
			\end{equation*}
			Hence, we conclude 
			\begin{equation*}
				\norm{(\lambda - G_\varepsilon)\ket{\psi}}^2 \ge \lambda^2 \braket{\psi,\psi} = \lambda^2 \norm{\ket{\psi}}^2 \, . 
			\end{equation*}
			Taking the square root on both sides proves the claim. Note that for $G_\varepsilon^\dagger$ all steps are similar due to the simple observation that $\braket{G\psi,\psi} + \braket{\psi, G\psi} = \braket{G^\dagger\psi,\psi} + \braket{\psi, G^\dagger\psi}$ for $\ket{\psi} \in\cH_f$.
			
			\textit{Step 2)} That the implemented semigroup\footnote{A discussion on implemented semigroups can be found in \cites{Alber.2001}.} $(\cP_t)_{t \ge 0}$ is a strongly continuous, positivity-preserving contractive semigroup that can be easily checked. We further get from \cite[Prop.~2.1]{Davies.1977} that it is generated by the closure of the operator
			\begin{equation*}
				\widetilde{\cG}_\varepsilon: \cD(\widetilde \cG_\varepsilon) = \{R(1,\overline{G}_\varepsilon) x R(1,\overline{G}_\varepsilon)^\dagger \; : \; x \in \cT_{1,\operatorname{sa}}\}  \to \cT_{1,\operatorname{sa}}, \quad x \mapsto \overline{G}_\varepsilon x + x \overline{G}_\varepsilon^\dagger
			\end{equation*}
			where $\overline{G}_\varepsilon$ is the closure of $G_\varepsilon$, $\overline{G}_\varepsilon^\dagger$ its adjoint, and $R( 1,\overline{G}_\varepsilon)$ its resolvent on $\cH$, respectively. Since $\cH_f$ is a core for the generator $\overline{G}_\varepsilon$, the set $O \coloneqq (\1 - \overline{G}_\varepsilon)\cH_f = (\1 - G_\varepsilon)\cH_f$ is dense in $\cH$, which in turn means $\cO = \text{span}\{\ketbra{\psi}{\varphi} \;:\; \ket{\psi}, \ket{\varphi} \in O\}$ is dense in $\cT_{1,\operatorname{sa}}$. A simple calculation further shows that $R(1,\overline{G}_\varepsilon)\cO R(1,\overline{G}_\varepsilon)^\dagger = \cT_f$. Hence for $y \in \cD(\widetilde \cG_\varepsilon)$, we find $x \in \cT_{1,\operatorname{sa}}$ and a sequence $\{x_n\}_{n \in \N} \subseteq \cO$ with $x_n \to x$ for $n \to \infty$ such that
			\begin{align*}
				\widetilde \cG_\varepsilon(y) &= \widetilde \cG_\varepsilon(R(1,\overline{G}_\varepsilon) x R(1,\overline{G}_\varepsilon)) = \widetilde \cG_\varepsilon(R(1,\overline{G}_\varepsilon) \lim\limits_{n \to \infty}x_n R(1,\overline{G}_\varepsilon)^\dagger)\\
				&= \lim\limits_{n \to \infty} \widetilde\cG_\varepsilon(R(1,\overline{G}_\varepsilon)x_n R(1,\overline{G}_\varepsilon)^\dagger)\\
				&= \lim\limits_{n \to \infty} \cG_\varepsilon(y_n)
			\end{align*}
			where we used the continuity of the map $\widetilde \cG_\varepsilon(R(1,\overline{G}_\varepsilon) \cdot R(1,\overline{G}_\varepsilon)^\dagger)$ and the set defined by $y_n = R(1,\overline{G}_\varepsilon) x_n R(1,\overline{G}_\varepsilon)^\dagger$. Note that $\{y_n\}_{n \in \N}$ is by construction a convergent sequence on $\cT_f$. In the last equality, we used that $\widetilde \cG_\varepsilon$ and $\cG_\varepsilon$ agree on $\cT_f$. This shows not only that $\cG_\varepsilon$ is closable but also that its closure is the closure of $\widetilde \cG_\varepsilon$. Hence the closure of $\cG_\varepsilon$ is the generator of $(\cP_t^\varepsilon)_{t \ge 0}$.
		\end{proof}
		Using the above lemma, we are now able to prove the following.
		
		\begin{lem}\label{lem:semigroup-of-G-for-sobolev-stability}
			For all $k \ge 0$ and $\varepsilon > 0$, the closure of the operator $\cG_\varepsilon$ from Lemma \ref{lem:semigroup-of-G-for-positivity} generates a strongly continuous, positivity preserving semigroup $(\cP^\varepsilon_t)_{t \ge 0}$ on $W^{k, 1}$ such that, for all $t\ge 0$, 
			\begin{equation*}
				\norm{\cP^\varepsilon_t}_{W^{k, 1} \to W^{k, 1}} \le e^{\omega_k t} 
			\end{equation*}
			where $\omega_k = \frac{k_{r_1} - k}{k_{r_1} - k_{r_0}}\omega_{k_{r_0}} + \frac{k - k_{r_0}}{k_{r_1} - k_{r_0}}\omega_{k_{r_1}}$ for an $r$ such that $k_{r_0}\leq k <k_{r_1}$. Finally, for $k = 0$, the semigroup is contractive.
		\end{lem}
		\begin{proof}
			Without loss of generality, we can restrict to $k \in \{k_r\}_{r \in \N}$ as for $k$ inbetween, we can interpolate between the $\{k_r\}_{r \in \N}$ (shown below) and $k = 0$ (shown in Lemma \ref{lem:semigroup-of-G-for-positivity}) using Lemma \ref{lem:interpolation-lemma}. Let $\varepsilon > 0$. In the following proof, the closure, domain and boundedness of an operator are always with respect to the Banach space $W^{k, 1}$ if not stated otherwise. We will show the claim, by first arguing that $\cG_\varepsilon$ is closable, that all $\lambda > \omega_k$ are in the resolvent set of the closure $\overline{\cG}_\varepsilon$ and further that $\norm{R(\lambda,\overline{\cG}_\varepsilon)}_{W^{k, 1} \to W^{k, 1}} \le \tfrac{1}{\lambda - \omega_k}$. By Theorem \ref{thm:hille-yosida}, the above immediately gives the existence of the semigroup on $W^{k, 1}$ and provides us with the claimed bound. The property of positivity preservation traces back to the representation of the semigroup via the Euler approximation and therefore the positivity of the resolvent: for any $x\in \cT_{1,\operatorname{sa}}$:
			\begin{equation}
				\cP_{t}^{\varepsilon}(x)=\lim_{n\to\infty}\,\left(\frac{n}{t}\right)^n\,R(n/t,\overline{\cG}_\varepsilon)^{n}(x)\,.
			\end{equation}
			
			The claims are proven in three steps.
			\begin{itemize}
				\item[\normalfont Step 1.]\label{step:step-1} Show that $\cG_\varepsilon:\cT_f \to \cT_f$ is closable and there exists a $\lambda > \omega_k$ such that $\lambda - \overline{\cG}_\varepsilon:\cD(\overline{\cG}_\varepsilon) \to W^{k, 1}$ is bijective.
				\item[\normalfont Step 2.] Using Assumption \ref{assum:sobolev-stability} and Lemma \ref{lem:semigroup-of-G-for-positivity} we prove that if $\lambda > \omega_k$ is in the resolvent set of $\overline{\cG}_{\varepsilon}$, we not only have that the resolvent is positivity preserving but further $$\norm{R(\lambda,\overline{\cG}_\varepsilon)}_{W^{k, 1} \to W^{k, 1}} \le \frac{1}{\lambda - \omega_k}\,.$$
				\item[\normalfont Step 3.] The surjectivity of $\lambda - \overline{\cG}_\varepsilon$ for a specific $\lambda > \omega_k$ from step 1.~and the bound on the resolvent from step 2.~allow us to successively use the series expansion of the resolvent as it is done in \cite[Prop.~IV.1.3]{Engel.2000} to get that $(\omega_k, \infty)$ is in the resolvent set of $\overline{\cG}_\varepsilon$, and therefore conclude the proof.
			\end{itemize}
			\medskip 
			
			\noindent  Proof of step 1.~We introduce the map 
			\begin{equation*}
				\cI_{d, \varepsilon}: \cT_f \to \cT_f, \qquad x \mapsto \cI_{d, \varepsilon}(x) \coloneqq - \varepsilon\{(N + \1)^{4d}, x\} \, . 
			\end{equation*}
			For $\lambda \ge 0$, $x \in \cT_f$, we can use Lemma \ref{lem:(n+1)-(n+1)-properties} to write
			\begin{equation*}
				(\lambda - \cG_\varepsilon)(x) = (\1 - \cG_0 \circ (\lambda - \cI_{d, \varepsilon})^{-1}) \circ (\lambda - \cI_{d, \varepsilon})(x)
			\end{equation*}
			where $\circ$ is the function composition and $\lambda - \cI_{d, \varepsilon}:\cT_f \to \cT_f$ is a bijection, with bounded inverse (see Lemma \ref{lem:(n+1)-(n+1)-properties}) between dense subspaces of $W^{k, 1}$. This means in particular that it is closable and that its closure has a bounded inverse. We will hence focus on the map $\1 - \cG_0 \circ (\lambda - \cI_{d, \varepsilon})^{-1}:\cT_f \to \cT_f$ and show that it is bounded (on the dense subset $\cT_f$ of $W^{k, 1}$) and hence uniquely extendable to all of $W^{k, 1}$. This then immediately gives us that $\lambda - \cG_0:\cT_f \to \cT_f$ is closable as it is the composition of a map with a dense range succeeded by a bounded map. Note first that we can apply Lemma \ref{lem:infinitesimal-boundedness-W-k-1} to $\cG_0$ to get that there exists $C_k \ge 0$ such that for all $\kappa > 0$ and $x \in \cT_f$
			\begin{equation*}
				\norm{\cG_0(x)}_{W^{k, 1}} \le \kappa \norm{\cI_{d, \varepsilon}(x)}_{W^{k, 1}} + \frac{C_k}{\kappa\varepsilon} \norm{x}_{W^{k, 1}}\,.
			\end{equation*}
			Using the bijectivity of $\lambda - \cI_{d, \varepsilon}:\cT_f \to \cT_f$ we get for $x \in \cT_f$
			\begin{align*}
				\norm{\cG_0 \circ (\lambda - \cI_{d, \varepsilon})^{-1}x}_{W^{k, 1}} \le \kappa\norm{\cI_{d, \varepsilon} \circ (\lambda - \cI_{d, \varepsilon})^{-1}(x)}_{W^{k, 1}} + \frac{C_k}{\varepsilon \kappa} \norm{(\lambda - \cI_{d, \varepsilon})^{-1}(x)}_{W^{k, 1}}\\
				\le ({2}\kappa + \frac{C_k}{\kappa \varepsilon} \frac{1}{\lambda + 2\varepsilon}) \norm{x}_{W^{k, 1}} =: f_k(\lambda, \kappa) \norm{x}_{W^{k, 1}}
			\end{align*}
			where we used properties of $\cI_{d, \varepsilon}$ derived in Lemma \ref{lem:(n+1)-(n+1)-properties}. This gives us not only that $\cG_0 \circ (\lambda - \cI_{d, \varepsilon})^{-1}:\cT_f \to \cT_f$ is bounded, hence uniquely extendable to a bounded map on $W^{k, 1}$ but for a fixed $\kappa < {\frac{1}{2}}$ and $\lambda > \lambda_\kappa$ where $\lambda_\kappa$ is chosen s.t. $f_k(\lambda_\kappa, \kappa) < 1$, we get that its closure is a strict contraction on $W^{k, 1}$. As a direct consequence, we find that again for $\lambda > \lambda_\kappa$ the closure of $\1 - \cG_0 \circ (\lambda - \cI_{d, \varepsilon}): \cT_f \to \cT_f$ is invertible with bounded inverse, and that its inverse function is just given by the geometric series of the closure of $\cG_0 \circ (\lambda - \cI_{d, \varepsilon})^{-1}$. To conclude, we can set $\lambda = 0$ in the above result and get that $-\cG_\varepsilon$ and hence $\cG_\varepsilon$ is closable and further that for $\kappa < \frac{1}{2}$ all $\lambda$ with $\lambda > \lambda_\kappa$ are in the resolvent set of $\overline{\cG}_{\varepsilon}$.
			\medskip
			
			\noindent Proof of step 2. Let $\lambda > \omega_k$ be in the resolvent set of $\overline{\cG}_\varepsilon$. From the compact embedding of $W^{k, 1}$ in $\cT_{1,\operatorname{sa}}$, we immediately get that $R(\lambda,\overline{\cG}_\varepsilon):W^{k, 1} \to W^{k, 1}$ agrees with the respective restricted resolvent of the closure $\widehat{\cG}_\varepsilon$ of $\cG_\varepsilon$ on $\cT_{1,\operatorname{sa}}$ that we obtained in Lemma \ref{lem:semigroup-of-G-for-positivity}. We know that the latter resolvent is positivity preserving, as the semigroup is. This is due to the following integral representation for strongly continuous semigroups \cite[Thm.~II.1.10 (i)]{Engel.2000}: for all $x\in x(\widehat{\cG}_\varepsilon)$,
			\begin{equation}
				R(\lambda, \widehat{\cG}_\varepsilon)(x)=\int_0^\infty\,e^{-\lambda s}\,e^{s\widehat{\cG}_\varepsilon}(x)\,ds\,.
			\end{equation}
			Hence $R(\lambda,\overline{\cG}_\varepsilon)$ is positivity preserving as well. Using Assumption \ref{assum:sobolev-stability}, we have that for $x \in \cT_f$, $x$ positive semi-definite,
			\begin{equation*}
				\tr[\cL(x) (N + \1)^{k/2}] \le \omega_k \tr[x (N + \1)^{k/2}] \, .
			\end{equation*}
			Adding non-negative terms, using the cyclicity of the trace and splitting up $\cL$ gives us
			\begin{align*}
				&\sum\limits_{j = 1}^K\tr[(N+\1)^{k/4}L_j x L_j^\dagger(N+\1)^{k/4}] + (\lambda - \omega_k)\tr[(N + \1)^{k/4} x (N + \1)^{k/4}] \\
				&\qquad\qquad\qquad\qquad\qquad\qquad\qquad\qquad\qquad \le \tr[(N + \1)^{k/4}(\lambda - \cG_\varepsilon)(x)(N + \1)^{k/4}]\,,
			\end{align*}
			and therefore
			\begin{align*}
				(\lambda - \omega_k) \norm{x}_{W^{k, 1}} \le \norm{(\lambda - \cG_\varepsilon)(x)}_{W^{k, 1}} \, 
			\end{align*}
			where we have just dropped non-negative terms and used $\tr[\cdot] \le \norm{\cdot}_1$ with equality if the argument is positive semi-definite. Since $\overline{\cG}_\varepsilon$ is the closure of $\cG_\varepsilon$, the above inequality extends to $x \in \cD( \overline{\cG}_\varepsilon)$, $x$ positive semi-definite and $\overline{\cG}_\varepsilon$ instead of $\cG_\varepsilon$. Together with the positivity preserving property of the resolvent, this gives us that for all $x \in W^{k, 1}$, $x$ positive semi-definite
			\begin{equation}\label{eq:resolvent-bound-G-positive-semidefinite}
				\norm{R(\lambda,\overline{\cG}_\varepsilon)x}_{W^{k, 1}} \le \frac{1}{\lambda - \omega_k} \norm{x}_{W^{k, 1}} \, . 
			\end{equation}
			For a general $x \in W^{k, 1}$, we set $x_\pm = \frac{1}{(N + \1)^{k/4}}\,[(N + 1)^{k/4} x (N + \1)^{k/4}]_\pm \,\frac{1}{(N + \1)^{k/4}} \in W^{k, 1}$, where $[\cdot]_\pm$ denotes the positive, resp. the negative part of a self-adjoint trace-class operator. We have that $x = x_+ - x_-$ and further that $x_+, x_-$ are positive semi-definite by construction. Hence
			\begin{align*}
				\norm{R(\lambda,\overline{\cG}_\varepsilon)x}_{W^{k, 1}} &\le \norm{R(\lambda,\overline{\cG}_\varepsilon )x_+}_{W^{k, 1}} + \norm{R(\lambda,\overline{\cG}_\varepsilon)x_-}_{W^{k, 1}}\\
				&\le \frac{1}{\lambda - \omega_k}(\norm{x_+}_{W^{k, 1}} + \norm{x_-}_{W^{k, 1}}) = \frac{1}{\lambda - \omega_k} \norm{x_+ - x_-}_{W^{k, 1}}\\
				&= \frac{1}{\lambda - \omega_k}\, \norm{x}_{W^{k, 1}}
			\end{align*}
			where we used \Cref{eq:resolvent-bound-G-positive-semidefinite} and the construction of $x_+$ and $x_-$, which concludes step 2.
			\medskip
			
			\noindent Proof of step 3. From step 1.~we get that there exists a $\lambda > \omega_k$ in the resolvent set of $\overline{\cG}_\varepsilon$ whereas step 2.~tells us that, for this $\lambda$, the resolvent is bounded by $\frac{1}{\lambda - \omega_k}$. We can use the same proof strategy as in \cite[Prop.~II.3.14 (ii)]{Engel.2000} where the authors employ the series expansion of the resolvent and its explicit bound to make conclusions about the resolvent set. Following their steps we first get that $(\omega_k, 2\lambda - \omega_k)$ is part of the resolvent set, and then using step 2.~again, we obtain the positivity preservation property as well as the explicit bound for all of those resolvents. This allows us to successively use these arguments and conclude that the resolvent set contains $(\omega_k, \infty)$.
		\end{proof}
		
		Putting together the results from Lemma \ref{lem:semigroup-of-G-for-positivity} and Lemma \ref{lem:semigroup-of-G-for-sobolev-stability}, we are now able to get rid of the perturbation $\cI_{d, \varepsilon}$.
		
		\begin{lem}\label{lem:eliminating-perturbation-G}
			The closure of 
			\begin{equation*}
				\cG:\cT_f \to \cT_f, \quad x \mapsto \cG(x) = Gx + x G^\dagger\,,
			\end{equation*}
			where $G$ is defined in \Cref{eq:lindblad}, generates a strongly continuous, positivity preserving semigroup $(\cP_t)_{t\ge 0}$ on $W^{k, 1}$ for all $k \in \N$ with 
			\begin{equation*}
				\norm{\cP_t}_{W^{k, 1} \to W^{k, 1}} \le e^{\omega_k t} \, . 
			\end{equation*}
			where $\omega_k = \frac{k_{r_1} - k}{k_{r_1} - k_{r_0}}\omega_{k_{r_0}} + \frac{k - k_{r_0}}{k_{r_1} - k_{r_0}}\omega_{k_{r_1}}$ for an $r$ such that $k_{r_0}\leq k <k_{r_1}$. Finally, for $k = 0$, the semigroup is contractive.
		\end{lem}
		\begin{proof}
			The proof is a direct application of Lemma \ref{lem:approximation-lemma} to the semigroups we obtained in Lemma \ref{lem:semigroup-of-G-for-sobolev-stability} taking $\varepsilon \to 0$. Since the semigroups in Lemma \ref{lem:semigroup-of-G-for-sobolev-stability} were positivity preserving, so is the obtained semigroup in the limit $\varepsilon\to 0$ (c.f.~Lemma \ref{lem:approximation-lemma}).
		\end{proof}
		
		We are now ready to prove the main Theorem of the section.
		
		\begin{proof}[Proof of Theorem \ref{thm:generation-theorem}]
			The proof strategy is inspired by \cite[Thm.~2.5]{Davies.1977}. It however makes use of Lemma \ref{lem:approximation-lemma} to avoid the issues discussed in \cite[§3]{Davies.1977}. Let $k \in \{k_r\}_{r \in \N}$ or $k = 0$ for the moment. Note that from Assumption \ref{assum:finite-degree}, we can conclude $\omega_k = 0$ for $k = 0$ in \Cref{eq:Assumptionsobolevstability}. We first define for $\delta \in (0, 1)$ the map
			\begin{align*}
				\cL_\delta:\cT_f \to \cT_f, \quad x \mapsto \cL_\delta(x) &= G x + x G^\dagger + \delta\sum\limits_{j = 1}^K L_j x L_j^\dagger =: \cG(x) + \delta\Sigma(x) , 
			\end{align*}
			and show that its closure defines a strongly continuous, positivity preserving semigroup $(\cP_t^\delta)_{t \ge 0}$ on $W^{k, 1}$ which further satisfies
			\begin{equation*}
				\norm{\cP_t^\delta}_{W^{k, 1} \to W^{k, 1}} \le e^{\omega_k t} \, . 
			\end{equation*}
			We first note that for $\widetilde \lambda > 0$, a rearrangement of \Cref{eq:Assumptionsobolevstability} using cyclicity of the trace and that $\tr[\cdot] \le \norm{\cdot}_1$ with equality if the argument is positive semi-definite gives
			\begin{equation*}
				\norm{\Sigma(x)}_{W^{k, 1}} \le \norm{(\widetilde \lambda + \omega_k - \cG)(x)}_{W^{k, 1}} 
			\end{equation*}
			for $x \in \cT_f$ and $x\ge 0$. Now using that $\cG$ is closable (Lemma \ref{lem:eliminating-perturbation-G}) and its resolvent positivity preserving we can conclude  for $\lambda \coloneqq \widetilde \lambda + \omega_k > \omega_k$, $x \in (\lambda - \cG)\cT_f$ and $x \ge 0$, 
			\begin{equation*}
				\norm{\Sigma \circ R(\lambda,\overline{\cG}) (x)}_{W^{k, 1}} \le \norm{x}_{W^{k, 1}}\,.
			\end{equation*}
			Applying similar methods as in step 2. of the proof of Lemma \ref{lem:semigroup-of-G-for-sobolev-stability}, we can extend the above inequality to general $x \in (\lambda - \cG) \cT_f$. Hence, $\Sigma \circ R(\lambda,\overline{\cG})$ is contractive on the dense set $(\lambda - \cG)\cT_f$ and positivity preserving, since both $\Sigma$ and $R(\lambda,\overline{\cG})$ are. It can therefore be uniquely extended to a positivity preserving contractive map on all of $W^{k, 1}$ which we will call $\cA_\lambda$ in the following. As a consequence $(\cL_\delta, \cD(\cL_\delta))$ is closable and $\lambda > \omega_k$ in the resolvent set of the closure. Both facts follow from the representation
			\begin{equation*}
				(\lambda - \cL_\delta) = (\1 - \delta\Sigma \circ R(\lambda,\overline{\cG})) \circ (\lambda - \cG)
			\end{equation*}
			which decomposes $\lambda - \cL_\delta$ into a composition of a closable map with a dense range and a map that is bounded on that range. We further get for the resolvent of the closure
			\begin{equation*}
				R(\lambda,\overline{\cL}_\delta) = R(\lambda,\overline{\cG}) \sum\limits_{n = 0}^\infty \delta^n \cA_\lambda^n \,,  
			\end{equation*}
			which immediately lets us conclude that the resolvent is positivity preserving as $\cA_\lambda$ and $R(\lambda,\overline{\cG})$ are. Lastly, we will show that for $\lambda > \omega_k$
			\begin{equation}\label{eq:bound-resolvent-L_r}
				\norm{R(\lambda,\overline{\cL}_\delta)}_{W^{k, 1} \to W^{k, 1}} \le \frac{1}{\lambda - \omega_k} \, . 
			\end{equation}
			To obtain this inequality we again rearrange Assumption \ref{assum:sobolev-stability}, add non-negative terms, use cyclicity of the trace and that $\tr[\cdot] \le \norm{\cdot}_1$ with equality if the argument is positive semi-definite, to conclude that for $x \in (\lambda - \cL_r)\cT_f$, $x$ positive semi-definite,
			\begin{equation*}
				\norm{R(\lambda,\overline{\cL}_\delta) x}_{W^{k, 1}} \le \frac{1}{\lambda - \omega_k}\norm{x}_{W^{k, 1}}\,.
			\end{equation*}
			We again extend the above bound to all $x \in (\lambda - \cL_\delta)\cT_f$ analogously to step 2 in the proof of Lemma \ref{lem:eliminating-perturbation-G}. Using that $(\lambda - \cL_\delta)\cT_f$ is dense then gives \Cref{eq:bound-resolvent-L_r}. Employing Theorem \ref{thm:lumer-phillips}, we get that indeed for all $\delta \in (0, 1)$ the closure of $(\cL_\delta, \cD(\cL_\delta))$ generates a strongly continuous semigroup which is positivity preserving since the resolvent is and satisfies the claimed bound. To now fill the gap between $0$ and the $\{k_r\}_{r \in \N}$ respectively, we interpolate between the semigroups (q.v.~Lemma \ref{lem:interpolation-lemma}), obtaining $e^{t\omega_k}$ where $\omega_k = \frac{k_{r_1} - k}{k_{r_1} - k_{r_0}}\omega_{k_{r_0}} + \frac{k - k_{r_0}}{k_{r_1} - k_{r_0}}\omega_{k_{r_1}}$ for an $r$ such that $k_{r_0}\leq k <k_{r_1}$, as the bound of the interpolated semigroups. Now that we have the result for all $k \ge 0$ we can employ Lemma \ref{lem:approximation-lemma} and take the limit $\delta \to 1$ to obtain the assertion. The contractivity and trace-preserving property of the semigroup in the case $k = 0$  just follows from the GKSL form of $(\cL, \cD(\cL))$, i.e.~$\tr[\cL(x)] = 0$ for $x \in \cT_f$, or put differently Assumption \ref{assum:finite-degree}.
		\end{proof}

	\subsection{Bosonic evolution systems}\label{sec:timedependentgeneration}

		Next, we consider time-dependent generators in GKSL form. For this, we modify Assumptions \ref{assum:finite-degree} and \ref{assum:sobolev-stability} in the following way: 
		
		\begin{assum}\label{assum:finite-degree-time-dep}
			The operator $(\cL_s,\cT_f)$ has $\operatorname{GKSL}$ form, i.e.~for $x\in\cT_f$ and $s\in[0,\infty)$ 
			\begin{equation}\label{eq:lindblad-time-dep}
				\begin{aligned}
					\cL_s:\cT_f \to \cT_f \quad x \mapsto\cL_s(x) &= - i [H(s), x] + \sum\limits_{j = 1}^K L_j(s) x L_j^\dagger(s)  - \frac{1}{2}\{L_j^\dagger(s) L_j(s), x\} \\
					&\coloneqq G(s)x + x G^\dagger(s) + \sum\limits_{j = 1}^K L_j(s) x L_j^\dagger(s)\, , 
				\end{aligned}
			\end{equation}
			where $K\in\mathbb{N}$, $G(s)=-iH(s)-\frac{1}{2}\sum_{j=1}^K L_j^\dagger(s) L_j(s)$, and $H(s)\coloneqq p_{H(s)}(a,a^\dagger), L_j(s)\coloneqq p_{j,s}(a,a^\dagger)$ are polynomials of the creation and annihilation operators with time-dependent, continuous coefficients. Again, $d_H\coloneqq\sup_{s\ge 0}\,\deg(p_{H(s)})<\infty$, $d_j\coloneqq\sup_{s\ge 0}\deg(p_{j,s})<\infty$, and $d\coloneqq\max\{d_1,...,d_K,d_H\}$.
		\end{assum}
		
		The next assumption will lead to the evolution system being Sobolev preserving, which allows us not only to prove the existence and uniqueness of the evolution generated by \eqref{eq:mastereq} but further to conduct a perturbation analysis as well as to extend our results to the case of a time-dependent Master equation:
		
		\begin{assum}\label{assum:sobolev-stability-time-dep}
			There exists a non-negative sequence $\{k_r\}_{r \in \N}\rightarrow\infty$ s.t.~for all $r \in \N$ there exist $\omega_{k_r} \ge 0$ such that for all $s \in \R_+$ and $x \in \cT_f$ positive semi-definite,
			\begin{equation}
				\tr[\cL_s(x) (N + \1)^{k/2}] \le \omega_{k_r} \tr[x (N + \1)^{k/2}] \, .
			\end{equation}
			Note that the coefficients $\omega_{k_r}$ are independent of $s$.
		\end{assum}
		
		Under the above assumptions we can state the generation theorem for evolution systems as follows:
		
		\begin{thm}[Generation of bosonic evolution systems]\label{thm:timedep-generation-theorem}
			Let $(\cL_s, \cD(\cL_s))_{s\in[ 0,\infty)}$ be a family of operators that fulfill Assumption \ref{assum:finite-degree-time-dep} and Assumption \ref{assum:sobolev-stability-time-dep}. Then $(\overline{\cL}_s, \cD(\overline{\cL}_s))_{s \in \R_+}$ gives rise to a unique evolution system $(\cP_{t,s})_{0\le s\le t}$ on $W^{k, 1}$ for all $k\ge 0$ with the following properties
			\begin{enumerate}
				\item $\cP_{t,s} (W^{k + 4d, 1}) \subseteq W^{k + 4d, 1}$ for all $0\le s\le t$
				\item For any $x \in W^{k + 4d, 1}$, the family $( \cP_{t,s}(x))_{0\le s\le t}$ is the unique solution to the initial value problem
				\begin{equation}\label{eq:time-dep-init-value-problem}
					\frac{d}{dt} x(t) = \overline{\cL}_t(x(t)) \qquad t \in [s, \infty), \; x(s) = x\,.
				\end{equation}
			\end{enumerate}
			For $k=0$, the evolution system is contractive and trace-preserving.
		\end{thm}
		\begin{proof}
			We assume w.l.o.g.~that $s \in [0,1]$ is fixed since the same argument works for all compact intervals. Theorem \ref{thm:generation-theorem} shows that $(\overline{\cL}_s,\cD(\overline{\cL}_s))$ generates an $\omega_k$-quasi-contractive semigroup $(\cP_t^s)_{t \ge 0}$ on $W^{k, 1}$. Next, we realize that $W^{k + 4d, 1}$ are $\cL_s$-admissible subspaces, where we recall that $d$ denotes the degree of $\cL_s$. This already proves assumptions (1) and (2) in Theorem \ref{thm:time-dependent-semigroups}. Since the coefficients of the polynomials $p_H$ and $p_j$ are continuous and operators of the form
			\begin{equation*}
				(N+\1)^{k/4} a^j(\ad)^l (N+\1)^{-(k/4+d)}\,,
			\end{equation*}
			for $j+l\leq d$, are bounded (see Lemma \ref{lem:boundedness-polynomials}) w.r.t.~the operator norm, we have by Hölder inequality that
			\begin{equation*}
				s \mapsto (N + \1)^{k/4} \cL_s((N + \1)^{-k/4 + d} (\cdot) (N + \1)^{-k/4 + d})(N + \1)^{k/4} =: \cA(s)
			\end{equation*}
			is a bounded and uniformly continuous family of operators. Therefore,
			\begin{equation*}
				s\mapsto\cL_s\in\cB(W^{k,1},W^{k+4d,1})
			\end{equation*}
			is uniformly continuous, which proves condition (3) in Theorem \ref{thm:time-dependent-semigroups}.
			Hence Theorem \ref{thm:time-dependent-semigroups} provides the existence and uniqueness of an evolution system on $W^{k, 1}$. By repeating the above arguments on $\cY\coloneqq W^{k+4d}$, i.e.~by choosing our $\cL_s$-admissible subspace as $W^{k + 8d, 1}$, Theorem \ref{thm:time-dependent-semigroups} provides existence and uniqueness of a solution on $\cY = W^{k + 4d, 1}$ which agrees with the former one on $W^{k, 1}$ by the compact embedding of $W^{k + 4d, 1}$ into $W^{k, 1}$. Therefore, conditions (4) and (5) are satisfied for the evolution system on $W^{k, 1}$ and the admissible subspace $\cY=W^{k + 4d, 1}$, which through Theorem \ref{thm:time-dependent-semigroups} proves the claim. Moreover, the evolution system is positivity preserving because it can be constructed by a concatenation of time-independent positivity preserving semigroups (see Theorem \ref{thm:generation-theorem} and \cites[Eq.~5.3.5]{Pazy.1983}). Contractivity and the property of trace preservation are a consequence of the fact that $\omega_0$ can be chosen to be $0$.
		\end{proof}

	\subsection{Multi-mode extension}\label{sec:multi-mode-extension}
		This section discusses the extension of Section \ref{sec:polynomial-generators} to the multi-mode setting. Since the details are almost completely analogous to the single-mode situation, we choose to elaborate only at places where some ambiguities might remain. Let us first fix the notations for this setting. We consider the Hilbert space of an $m$-mode system, $m\in\N$, whose Hilbert space we conveniently denote by $\cH_m=L^2(\mathbb{R}^m)$. We further use $\cB(\cH_m)$ for the bounded, $\cT_1$ for the trace class, and $\cT_{1, \operatorname{sa}}$ for the self-adjoint trace class operators. Now we define $\cT_f$ to be
		\begin{equation*}
			\cT_f \coloneqq \{x = \sum\limits_{\text{finite}} f_{\n, \p}  \ketbra{n_{1}}{p_{1}} \otimes \hdots \otimes \ketbra{n_m}{p_m} \;:\; f_{\n, \p} \in \C, \; x = x^\dagger\} \, ,
		\end{equation*}
		where $\n = (n_1, \hdots, n_m) \in \N^m$ and $\p$ analogously function as an index in $f_{\n, \p}$. For $\k = (k_1, \hdots, k_m) \in \R_+^m$, we define $\k \prec \k'$ if $k_j < k_j'$ for all $j = 1, \hdots, m$. Analogously, we define $\preceq$. We set for $\k \in \R_+^m$
		\begin{equation*}
			(N + \1)^\k \coloneqq (N_1 + \1)^{k_1} \otimes \hdots \otimes (N_m + \1)^{k_m},
		\end{equation*} 
		and with this define $W^{\k, 1}$, $\norm{\,\cdot\,}_{W^{\k, 1}}$. Remark that the latter spaces are Banach spaces and in correspondence to Lemma \ref{lem:sobolev-embedding} we find:
		\begin{lem}\label{lem:multi-mode-sobolev-embedding}
			Let $\k, \k' \in \N^m$ with $\k \prec \k'$, then 
			\begin{equation}
				W^{\k', 1} \Subset W^{\k, 1} \, . 
			\end{equation}
		\end{lem}
		The strategy to prove the above claims is analogous to the single-mode case. Next, we slightly generalize the single-mode results Theorem \ref{thm:stein-weiss} and Definition \ref{defi:sobolev-preserving-semigroups} to the multi-mode setting. 
		
		\begin{thm}[Stein-Weiss theorem for multi-mode Bosonic Sobolev spaces]
			Let $\k_0, \k_1 \in \R_+^m$, $\k_0 \prec \k_1$ and $T: W^{\k_j, 1} \to W^{\k_j, 1}$ a linear map with $\norm{T}_{W^{\k_j, 1} \to W^{\k_j, 1}} \le M_j$, bounded by $M_j \ge 0$ for $j = 0, 1$ respectively. Then for $\theta \in [0, 1]$, $T: W^{\k_\theta, 1} \to W^{\k_\theta, 1}$ with $\k_\theta = (1 - \theta) \k_0 + \theta \k_1$ obtained by restriction of the input of  $T:W^{\k_0, 1} \to W^{\k_0, 1}$ to $W^{\k_\theta, 1} \cap W^{\k_0, 1}$, is a well defined bounded linear map with
			\begin{equation}
				\norm{T}_{W^{\k_\theta, 1} \to W^{\k_\theta, 1}} \le M_0^{(1 - \theta)} M_1^\theta \, . 
			\end{equation}
		\end{thm}
		
		In the multi-mode setting, we cannot interpolate between the elements of the divergent sequence to obtain $W^{\k, 1}$ as admissible subspace for all $0 \prec \k$ but only for elements in the convex hull of the divergent sequence. The property of being Sobolev preserving is again defined for a sequence $\{\k_r\}_{r \in \N}$ such that $\lim\limits_{r \to \infty} \min\limits_{j = 1, \hdots, m} k_{j, r} = \infty$.
		
		\begin{defi}[Sobolev preserving semigroup/evolution system in multi-mode systems]
			Let $(\cP_t)_{t\ge 0}$ be a $C_0$-semigroup on $\cT_{1,\operatorname{sa}}$. We then call $(\cP_t)_{t \ge 0}$ \textit{Sobolev preserving} if there exists a divergent sequence $\{\k_r\}_{r \in \N} \subset \R_+^m$, in the sense that $\lim\limits_{r \to \infty} \min\limits_{j = 1, \hdots, m} k_{j, r} = \infty$, s.t. for all $r \in \N$, $W^{\k_r, 1}$ is an admissible subspace for $(\cP_t)_{t\ge 0}$. Similarly for an evolution system $(\cP_{t,s})_{0\leq s\leq t}$ on $\cT_{1, \operatorname{sa}}$, we call it $(\cP_{t,s})_{0\leq s\leq t}$ \textit{Sobolev preserving} if for all $r \in \N$, $W^{\k_r, 1}$ is admissible for $(\cP_{t,s})_{0\leq s\leq t}$ to $W^{\k_r, 1}$.
		\end{defi}
		
		With these preliminaries in place we can now lift Assumption \ref{assum:finite-degree}, Assumption \ref{assum:sobolev-stability}, Assumption \ref{assum:finite-degree-time-dep}, and Assumption \ref{assum:sobolev-stability-time-dep}.
		
		\begin{assum}\label{assum:multi-mode-finite-degree}
			The operator $(\cL, \cT_f)$ has $\operatorname{GKSL}$ form, i.e.~for $x \in \cT_f$,
			\begin{equation}
				\begin{aligned}
					\cL:\cT_f \to \cT_f \quad x \mapsto\cL(x) &= - i [H, x] + \sum\limits_{j = 1}^K L_j x L_j^\dagger  - \frac{1}{2}\{L_j^\dagger L_j, x\} \\
					&\coloneqq Gx + x G^\dagger + \sum\limits_{j = 1}^K L_j x L_j^\dagger\, , 
				\end{aligned}
			\end{equation}
			for some $K \in \N$ and with $G = - iH - \frac{1}{2}\sum\limits_{j = 1}^K L_j^\dagger L_j$, where $\{A, B\} = AB + BA$ denotes the anti-commutator of two operators $A, B$ on a suitable domain. Further $H$ and $L_j$ are assumed to be polynomials of the creation and annihilation operators, i.e.~$H \coloneqq p_H(a_1, a_1^\dagger, \hdots, a_m, a_m^\dagger)$, $L_j \coloneqq p_j(a_1, a_1^\dagger, \hdots, a_m, a_m^\dagger)$ and $H$ symmetric. This ensures that $\cT_f$ is invariant under $\cL$. We denote the degree of $p_H$ by $d_H \coloneqq \deg p_H$, those of $p_j$ by $d_j \coloneqq \deg p_j$, and $d \coloneqq \max\{d_1, \hdots, d_K, d_H\}$.
		\end{assum}
		
		The second assumption becomes:
		
		\begin{assum}\label{assum:multi-mode-sobolev-preserving}
			There exists a non-negative sequence $\{\k_r\}_{r \in \N} \subset \R_+^m$, in the sense that \\ $\lim\limits_{r \to \infty} \min\limits_{j = 1, \hdots, m} k_{j, r} = \infty$, s.t.~for every $r \in \N$, there exist $\omega_{\k_r} \ge 0$ such that for all positive semi-definite $x \in \cT_f$
			\begin{equation}\label{eq:multi-modeassumptionsobolevstability}
				\tr[\cL(x) (N + \1)^{\k_r/2}] \le \omega_{\k_r} \tr[x (N + \1)^{\k_r/2}] \, .
			\end{equation}
		\end{assum}
		
		Then employing the single-mode strategy, we obtain the following theorem.
		
		\begin{thm}[Generation of multi-mode bosonic semigroups]\label{thm:multi-mode-generation-theorem}
			Let $(\cL, \cD(\cL))$ be an operator defined on $\cT_{1,\operatorname{sa}}$. If $(\cL, \cD(\cL))$ satisfies Assumption \ref{assum:multi-mode-finite-degree} and Assumption \ref{assum:multi-mode-sobolev-preserving}, then the closure $\overline{\cL}$ generates a strongly continuous, positivity preserving semigroup $(\cP_t)_{t\ge 0}$ on $W^{\k_r, 1}$ for all $\{\k_r\}_{r \in \N}$ from Assumption \ref{assum:multi-mode-sobolev-preserving}. We further find that the semigroup satisfies the bound
			\begin{equation}
				\norm{\cP_t}_{W^{\k_r, 1} \to W^{\k_r, 1}} \le e^{\omega_{\k_r} t} \,\quad \forall t\ge 0\, .
			\end{equation}
			In the special case $\k = 0$, the semigroup is contractive and trace-preserving.
		\end{thm}
		
		Note that we can extend the above semigroups to the convex hull of $\{\k_r\}_{r \in \N} \cup \{0\}$ using a generalisation of the interpolation lemma for single mode semigroups (q.v. Lemma \ref{lem:interpolation-lemma}). 
		
		Lastly, we can also generalize the generation theorem for evolution systems modifying the assumptions accordingly.
		
		\begin{assum}\label{assum:multi-mode-finite-degree-time-dep}
			The operator $(\cL_s,\cT_f)$ has $\operatorname{GKSL}$ form, i.e.~for $x\in\cT_f$ and $s\in[0,\infty)$ 
			\begin{equation}\label{eq:multi-mode-lindblad-time-dep}
				\begin{aligned}
					\cL_s:\cT_f \to \cT_f \quad x \mapsto\cL_s(x) &= - i [H(s), x] + \sum\limits_{j = 1}^K L_j(s) x L_j^\dagger(s)  - \frac{1}{2}\{L_j^\dagger(s) L_j(s), x\} \\
					&\coloneqq G(s)x + x G^\dagger(s) + \sum\limits_{j = 1}^K L_j(s) x L_j^\dagger(s)\, , 
				\end{aligned}
			\end{equation}
			where $K\in\mathbb{N}$, $G(s)=-iH(s)-\frac{1}{2}\sum_{j=1}^K L_j^\dagger(s) L_j(s)$, and $H(s)\coloneqq p_{H(s)}(a_1,a^\dagger_1, \hdots, a_m, a^\dagger_m)$, $L_j(s)\coloneqq p_{j,s}(a_1,a^\dagger_1, \hdots, a_m, a^\dagger_m)$ are polynomials of the creation and annihilation operators with time-dependent, continuous coefficients. Again, $d_H\coloneqq\sup_{s\ge 0}\,\deg(p_{H(s)}) < \infty$, $d_j\coloneqq\sup_{s\ge 0}\deg(p_{j,s})<\infty$, and $d\coloneqq\max\{d_1,...,d_K,d_H\}$.
		\end{assum}
		
		The second assumption in the time-dependent case generalizes to the following:
		
		\begin{assum}\label{assum:multi-mode-sobolev-stability-time-dep}
			There is a divergent sequence $\{\k_r\}_{r \in \N} \subset \R_+^m$, meaning $\lim\limits_{r \to \infty} \min\limits_{j = 1, \hdots, m} k_{j, r} = \infty$, s.t. for every $r \in \N$, there exist $\omega_{\k_r} \ge 0$ such that for all $s \in \R_+$ and $x \in \cT_f$ positive semi-definite,
			\begin{equation}
				\tr[\cL_s(x) (N + \1)^{\k_r/2}] \le \omega_{\k_r} \tr[x (N + \1)^{\k_r/2}] \, .
			\end{equation}
			Note that the coefficients $\omega_{\k_r}$ are independent of $s$.
		\end{assum}
		
		Under the above assumptions we can state the generation theorem for multi-mode evolution systems as follows:
		
		\begin{thm}[Generation of multi-mode bosonic evolution systems]\label{thm:multi-mode-timedep-generation-theorem}
			Let $(\cL_s, \cD(\cL_s))_{s\in[ 0,\infty)}$ be a family of operators that fulfills Assumption \ref{assum:multi-mode-finite-degree-time-dep} and Assumption \ref{assum:multi-mode-sobolev-stability-time-dep}. Then 
			$(\overline{\cL}_s, \cD(\overline{\cL}_s))_{s \in \R_+}$ gives rise to a unique evolution system $(\cP_{t,s})_{0\le s\le t}$ on $W^{\k_r, 1}$ for all $r \in \N$ with the following properties: for $\k_{r'}$ with $\min\limits_{j = 1, \hdots, m}|k_{j, r} - k_{j, r'}| \ge d$
			\begin{enumerate}
				\item $\cP_{t,s} (W^{\k_{r'}, 1}) \subseteq W^{\k_{r'}, 1}$ for all $0\le s\le t$;
				\item For any $x \in W^{\k_{r'}, 1}$, the family $(\cP_{t,s}(x))_{0\le s\le t}$ is the unique solution to the initial value problem
				\begin{equation}\label{eq:multi-mode-time-dep-init-value-problem}
					\frac{d}{dt} x(t) = \overline{\cL}_t(x(t)) \qquad t \in [s, \infty), \; x(s) = x\,.
				\end{equation}
			\end{enumerate}
			For $\k = 0$ as a special case, we get that the evolution system is contractive and trace-preserving.
		\end{thm}

\section{Examples of Sobolev preserving semigroups}\label{sec:examples-sobolev-preserving-semigroup}
	In this section, we consider two classes of examples of practical relevance in quantum information processing for which Assumption \ref{assum:finite-degree} (or \ref{assum:finite-degree-time-dep},\ref{assum:multi-mode-finite-degree},\ref{assum:multi-mode-finite-degree-time-dep}) trivially holds and derive Assumption \ref{assum:sobolev-stability} (or \ref{assum:sobolev-stability-time-dep}, \ref{assum:multi-mode-sobolev-preserving}, \ref{assum:multi-mode-sobolev-stability-time-dep}). Particular care will be given to finding time-independent upper bounds on the $W^{k,1}\to W^{k,1}$ norm of the semigroup. For this, the overall strategy is as follows: given the generator $(\cL,\cT_f)$, we prove that there are coefficients $\mu_{k_r}\ge 0, c_{k_r}>0$ for a divergent sequence $\{k_r\}_{r \in \N}$ such that for all state $\rho\in\cT_f$
	\begin{equation}\label{eq:examples-assum2-step1}
		\begin{aligned}
			\tr[\cL(\rho) (\Nind+\1)^{k_r/2}] &\leq - c_{k_r}\tr[\rho (\Nind+\1)^{k_r/2}] + \mu_{k_r} \\
			& \leq (\mu_{k_r}- c_{k_r})\tr[\rho (\Nind+\1)^{k_r/2}]\,,
		\end{aligned}
	\end{equation}
	where we have used $\tr[\rho (\Nind+\1)^{k_r/2}]\geq \tr[\rho] = 1$ in the second inequality.
	Then, Theorem \ref{thm:generation-theorem} can be applied, which shows that for all $k \in \R_+$, the closure of $(\cL,\cT_f)$ generates a positivity preserving $C_0$-semigroup $(\cP_t)_{t\ge 0}$ on $W^{k,1}$. In the case $k \in \{k_r\}_{r \in \N}$:
	\begin{equation}\label{eq:sobolev-preserving-stab-constant}
		\|\cP_t(x)\|_{W^{k,1}}\leq e^{|\mu_k- c_k|\,t}\|x\|_{W^{k,1}}\,.
	\end{equation}
	for all $x\in W^{k,1}$. The bounds for the intermediate values of $k$ can be obtained using Lemma \ref{lem:interpolation-lemma}. One can strengthen the above bounds using \Cref{eq:examples-assum2-step1} as follows:
	
	\begin{prop}\label{prop-ex:uniformly-bounded-semigroup}
		Let $(\cL,\cT_f)$ be an operator satisfying Assumption \ref{assum:finite-degree} and \Cref{eq:examples-assum2-step1}. Then, for all $k\in\N$, the closure of $(\cL,\cT_f)$ generates a positivity preserving $C_0$-semigroup $(\cP_t)_{t\ge 0}$ on $W^{k,1}$. For all $r \in \N$ and all states $\rho\in W^{k,1}$,
		\begin{equation*}
			\|\cP_t(\rho)\|_{W^{k_r,1}}\leq \max\left\{\|\rho\|_{W^{k_r,1}},\,\frac{\mu_{k_r}}{c_{k_r}}\right\} \, .
		\end{equation*}
		For a general $k \in \R_+$ and $x \in W^{k, 1}$ one obtains
		\begin{equation}\label{eq:improved-semigroup-bound}
			\|\cP_t(x)\|_{W^{k_r,1}}\leq \gamma_k \|x\|_{W^{k_r,1}} \, ,
		\end{equation}
		where $\gamma_k = \max\{1, \frac{\mu_k}{c_k}\}$ for $k \in \{k_r\}_{r \in \N}$ and an interpolated time-independent constant in all other cases. Note that for $k > 0$ and $\rho \in W^{k, 1}$ there exists a sequence $\{t_n\}_{n \in \N}$, such that
		\begin{equation*}
			\lim_{t_n\rightarrow\infty}\cP_{t_n}(\rho)=\overline{\rho}
		\end{equation*}
		for $\overline{\rho}\in W^{k,1}$. Similar conclusions hold in multi-mode as well as time-dependent settings.
	\end{prop}
	\begin{proof}
		By assumption, Theorem \ref{thm:generation-theorem} shows that the closure of $(\cL,\cT_f)$ defines a positivity preserving, quasi-contractive semigroup $(\cP_t)_{t\ge 0}$. Moreover, for $k \in \{k_r\}_{r \in \N}$, $\rho(t)\coloneqq \cP_t(\rho)$
		\begin{equation*}
			\begin{aligned}
				\frac{d}{dt}\|\rho(t)\|_{W^{k,1}}&=\tr[\cL(\rho(t))(N+\1)^{k/2}]\\
				&\leq-c_{k}\tr[\rho(t) (\Nind+\1)^{k/2}] + \mu_{k}\\
				&=-c_{k}\|\rho\|_{W^{k,1}} + \mu_{k}\,.
			\end{aligned}
		\end{equation*}
		Thus, for $\|\rho(t)\|_{W^{k,1}}\geq\frac{\mu_k}{c_k}$, we have $\frac{d}{dt}\|\rho(t)\|_{W^{k,1}}\leq0$, which concludes the bound. Using the positivity preserving property of the semigroup and that $\norm{\cdot}_1 \le \norm{\cdot}_{W^{k, 1}}$ one can lift the bound to \Cref{eq:improved-semigroup-bound} for general $x \in W^{k, 1}$ and Theorem \ref{thm:stein-weiss} allows us to conclude 
		\begin{equation*}
			\|\cP_t(x)\|_{W^{k_r,1}}\leq \gamma_k \|x\|_{W^{k_r,1}}
		\end{equation*}
		extend to all $k \in \R_+$. Finally, for every $k > 0$, every sequence $n\rightarrow \cP_{t_n}(\rho)$ is uniformly bounded in $W^{k,1}$ so that the compact embedding shows that there exists a converging subsequence in $W^{k-\varepsilon,1}$ for $\varepsilon$ suitably chosen, which is also converging in $W^{k,1}$. This finishes the proof.
	\end{proof}
	To achieve the inequality stated in \Cref{eq:examples-assum2-step1}, we will make heavy use of the following simple commutation relations: given a real-valued function $f:\mathbb{N}\to\mathbb{R}$,  
	\begin{equation}\label{eq:symmetry-function}
		\begin{aligned}
			af(\Nind+j\1)=f(\Nind + (j+1)\1)a,&\quad\quad a^\dagger\,1_{>j}f(\Nind-j\1)=f(\Nind - (j+1)\1)a^\dagger\,1_{>j}\,,\\
			f(\Nind-j\1)a\,1_{>j}=af(\Nind - (j+1)\1)1_{>j},&\quad\quad f(\Nind+j\1)a^\dagger=a^\dagger f(\Nind + (j+1)\1)\,,
		\end{aligned}
	\end{equation}
	where the operators above are defined e.g.~on $\cH_f$. We also use the canonical commutation relation to write $(\ad)^la^l$ as a function of $N$ (see Lemma \ref{lem:l-ccr}): 
	\begin{align*}
		&(\ad)^la^l=(N-(l-1)\1)(N-(l-2)\1)\cdots(N-\1)N\\
		& a^l(\ad)^l=(N+\1)(N+2\1)\cdots(N+(l-1)\1)(N+l\1)\,.
	\end{align*}
	In the following, we adopt the notations:
	\begin{align*}
		\cL[L]\coloneqq L(\cdot)L^\dagger -\frac{1}{2}\,\{L^\dagger L,\,\cdot\}\,\qquad \text{ and }\qquad \cH[H]\coloneqq -i[H,\cdot]\,.
	\end{align*}
	Although this notation collides with the one for the Hilbert space, the meaning can always be deduced from context.

	\subsection{Quantum Ornstein Uhlenbeck semigroup}

		We start with the generator of the quantum Ornstein Uhlenbeck semigroup \cite{Cipriani.2000,Carbone.2007} defined by
		\begin{equation}
			\cL_{\operatorname{qOU}} = \lambda^2 \cL[a] + \mu^2 \cL[\ad]
		\end{equation}
		for $\mu, \lambda \ge0$. Given an suitably domain $\cD(\cL_{\operatorname{qOU}})$, the operator $(\overline{\cL}_{\operatorname{qOU}},\cD(\cL_{\operatorname{qOU}}))$
		is known to generate a quantum dynamical semigroup $(\cP_t^{\operatorname{qOU}})_{t\ge 0}$. Here, we further show that the quantum Ornstein Uhlenbeck semigroup defines a semigroup on all $W^{k,1}$. This is the topic of the following lemma: 
		
		\begin{lem}\label{lem-ex:qOU-differential-stability}
			Let $(\cL_{\operatorname{qOU}},\cT_f)$ be the generator of the quantum Ornstein Uhlenbeck semigroup and $k \in \N$. Then, there exist constants $\mu_k$ explicated in \eqref{eq-ex:qou-lambda>mu} such that, for all states $\rho\in\cT_f$,
			\begin{equation*}
				\begin{aligned}
					\tr[\cL_{\operatorname{qOU}}(\rho)(N+\1)^{\frac{k}{2}}]\le\begin{cases}
						\frac{k}{4}(\mu^2-\lambda^2)\tr\big[\rho(N+\1)^{k/2}\big]+\mu_k&\lambda>\mu\\
						\frac{k}{2}(2\mu^2+k)\tr\big[\rho(N+\1)^{k/2}\big]&\lambda\leq\mu
					\end{cases}
				\end{aligned}\,.
			\end{equation*}
			Therefore, the semigroup $e^{t\mathcal{L}_{\operatorname{qOU}}}$ is a Sobolev and positivity preserving quantum Markov semigroup satisfying for all states $\rho\in W^{k,1}$
			\begin{equation*}
				\begin{aligned}
					\|e^{t\mathcal{L}_{\operatorname{qOU}}}(\rho)\|_{W^{k,1}}\leq\begin{cases}
						\max\left\{\|\rho\|_{W^{k,1}},\,\frac{4\mu_{k}}{k(\mu^2-\lambda^2)}\right\}&\lambda>\mu\\
						e^{t\frac{k}{2}(2\mu^2+k)}\|\rho\|_{W^{k,1}}&\lambda\leq\mu
					\end{cases}\,.
				\end{aligned}
			\end{equation*}
		\end{lem}
		\begin{proof}
			We consider $\cL_{\operatorname{qOU}}^\dagger(f(N))$ where $f(x)=(x+1)^{k/2} 1_{x\ge -1}$. By \Cref{eq:symmetry-function},
			\begin{equation*}
				\begin{aligned}
					\cL_{\operatorname{qOU}}^\dagger(f(N))&=\lambda^2N(f(N-\1)-f(N))+\mu^2(N+\1)(f(N+\1)-f(N))\,.
				\end{aligned}
			\end{equation*}
			Note that the case $k=0$ follows from the GKLS form and $k=2$ is by definition of $f$ trivially given by $(\mu^2-\lambda^2)N+\1$. 
			Next, we define an auxiliary function which will also prove useful in the following proofs: 
			\begin{equation}\label{eq:f-g-l-function}
				g_l(x) = \begin{cases}
					f(x) - f(x - l) & x \ge l;\\
					f(x) & l > x \ge 0;\\
					0 & 0 > x\,.
				\end{cases}
			\end{equation}
			It allows us to redefine $\cL_{\operatorname{qOU}}^\dagger(f(N))$ by 
			\begin{equation*}
				\begin{aligned}
					\cL_{\operatorname{qOU}}^\dagger(f(N))&=-\lambda^2Ng_1(N)+\mu^2(N+\1)g_1(N+\1)\,.
				\end{aligned}
			\end{equation*}
			Then, applying Lemma \ref{lem:upper-lower-bound-gl} to the spectral decomposition of the polynomial in the number operator above, we get 
			\begin{equation*}
				\begin{aligned}
					&\cL_{\operatorname{qOU}}^\dagger(f(N))\\
					&\quad\leq \frac{k}{2}(\mu^2-\lambda^2)(N+\1)^{k/2}+\lambda^2\frac{k}{2}(N+\1)^{k/2-1}+1_{k\geq3}(N+\1)^{k/2-2}\frac{k^2}{8}+\mu^2\frac{2-k}{2}\ketbra{0}{0}\\
					&\quad\leq \frac{k}{2}(\mu^2-\lambda^2)(N+\1)^{k/2}+\frac{k}{2}\left(\lambda^2+\mu^2+{k}\right)(N+\1)^{k/2-1}
				\end{aligned}
			\end{equation*}
			where we separated the vacuum state from the rest of the decomposition. Note that this bound can also be used when $k=1$ since $(N+\1)^{-1/2}$ is then bounded by $1$. Therefore, we assume $k\geq3$ in the following and start with the case $\lambda>\mu$ so that the leading order is negative. Then, we use half of the latter to bound the other terms by a constant. This is done by the following classical optimization
			\begin{equation*}
				\sup_{x\geq0}\left(-x^\nu+cx^{\nu-1}\right)=c^\nu\left(\frac{(\nu-1)^{\nu-1}}{\nu^\nu}\right)
			\end{equation*}
			for $\nu\geq1$ and $c\geq0$ defined as 
			\begin{equation*}
				c=2\frac{\lambda^2+\mu^2+k}{\lambda^2-\mu^2}\qquad\text{and}\qquad\nu=\frac{k}{2}\,.
			\end{equation*}
			Then,
			\begin{equation}\label{eq-ex:qou-lambda>mu}
				\begin{aligned}
					\cL_{\operatorname{qOU}}^\dagger(f(N))&\leq \frac{k}{4}(\mu^2-\lambda^2)(N+\1)^{k/2}+c^\nu\left(\frac{(\nu-1)^{\nu-1}}{\nu^\nu}\right)\\
                    &=:\frac{k}{4}(\mu^2-\lambda^2)(N+\1)^{k/2}+\mu_k^{\lambda>\mu}
				\end{aligned}
			\end{equation}
			The second case is $\lambda\leq\mu$, which can be easily upper bounded by
			\begin{equation*}
				\begin{aligned}
					&\cL_{\operatorname{qOU}}^\dagger(f(N))\leq \frac{k}{2}(2\mu^2+k)(N+\1)^{k/2}
				\end{aligned}\,.
			\end{equation*}
			This completes the proof of the statement by Theorem \ref{thm:generation-theorem} and Proposition \ref{prop-ex:uniformly-bounded-semigroup}.
		\end{proof}

	\subsection{Photon-dissipation and CAT qubits}\label{sec:cat-qubits}

		Next, we consider a family of Lindbladians that has been recently studied in the setting of error correction with continuous variable quantum systems. For an introduction to the field, we refer the interested reader to the following lecture notes \cites{Preskill.2021}{Guillaud.2023}. The abstract idea here is that the code-space is continuously protected by a dissipative evolution, i.e.~an evolution which is exponentially converging for $t\rightarrow\infty$ to an invariant subspace --- the code-space. This behavior is achieved through the so-called $l$-photon dissipation generated for $\kappa>0$ and $\alpha\in\C$ by
		\begin{equation}\label{eq:l-photon-dissipation}
			\kappa\cL[a^l-\alpha^l]\,,
		\end{equation}
		where we sometimes omit the identity so that $\alpha^l\coloneqq \alpha^l\1$ in what follows. The invariant subspace (code-space) to which the evolution is exponentially converging \cite{Azouit.2016} is defined by 
		\begin{equation}\label{eq:codespace}
			\cC_l\coloneqq\spa\left\{\ketbra{\alpha_1}{\alpha_2}\,:\,\alpha_1,\alpha_2\in\left\{\alpha e^{\frac{i2\pi j}{l}}\,|\,j\in\{0,...,l-1\}\right\}\right\}\,,
		\end{equation}
		where $\ket{\alpha}$ denotes the coherent state
		\begin{equation*}
			\ket{\alpha}=e^{-\frac{|\alpha|^2}{2}}\sum_{n=0}^{\infty}\frac{\alpha^n}{\sqrt{n!}}\ket{n}\,.
		\end{equation*}
		and satisfies $a\ket{\alpha}=\alpha\ket{\alpha}$ by definition.
		
		Besides the $l$-photon dissipation, we consider the CAT qubit error correction protocol introduced in \cite{Guillaud.2019} associated with the $2$-photon dissipation and code-space $\cC_2$ and with corresponding universal gate-set generated by the following generators: for some parameters $T,\kappa,\varepsilon>0$,
		
		\medskip
		\medskip 
		
		\textit{Identity-gate:}
		\begin{equation}\label{eq:cat-identity}
			\kappa\cL[a^2-\alpha^2]
		\end{equation}
		
		\textit{$Z(\theta)$-gate:}
		\begin{equation}\label{eq:cat-z}
			\kappa\cL[a^2-\alpha^2]+\varepsilon\cH[a+\ad]
		\end{equation}
		
		\textit{$X$-gate:}
		\begin{equation}\label{eq:cat-X}
			\kappa\cL[a^2-e^{2i\pi t/T}\alpha^2]
		\end{equation}
		
		\textit{$\operatorname{CNOT}$-gate:}
		\begin{equation}\label{eq:cat-cnot}
			\kappa\cL[a^2-\alpha^2]+\varepsilon\cL[b^2-\alpha^2-\frac{\alpha}{2}(1-e^{2i\pi t/T})(a-\alpha)]
		\end{equation}
		
		\textit{Toffoli-gate:}
		\begin{equation}\label{eq:cat-toffoli}
			\kappa\cL[a^2-\alpha^2]+\kappa\cL[b^2-\alpha^2]+\varepsilon\cL[c^2-\alpha^2+\frac{1}{4}(1-e^{2i\pi t/T})(ab-\alpha(a+b)+\alpha^2)]\,.
		\end{equation}
		
		Note that the CNOT gate operates on two modes, while the Toffoli gate acts on three modes, with the annihilation and creation operators on the second mode denoted by $b$ and $b^\dagger$, and on the third mode by $c$ and $c^\dagger$. Furthermore, from a mathematical standpoint, the implementation remains somewhat unclear, as the sum gates are constructed using an adiabatic limit, for which, to our knowledge, rigorous error bounds have yet to be established.
		
		In the following, we prove that the above operators generate Sobolev-preserving quantum dynamical semigroups, with the exception of the Toffoli gate. Due to its more complicated structure, we leave the analysis of the latter to future work. We start by proving that the $l$-photon dissipation satisfies \Cref{eq:examples-assum2-step1}, and therefore that it generates a Sobolev preserving semigroup by Proposition \ref{prop-ex:uniformly-bounded-semigroup}.
		
		\begin{lem}[$l$-photon dissipation]\label{lem:l-diss}
			For any $k\ge 1$, $l\ge 2$, $\alpha \in\mathbb{C}$ and any state $\rho\in\cT_f$, 
			\begin{align*}
				\tr\big[\cL[a^l-\alpha^l ](\rho)(N+\1)^{k/2}\big]&\le -\frac{l}{2} \tr\big[\rho\,(N+\1)^{\nu}\big]+\frac{l}{2}\mu_k^{(l)}\le -\frac{l}{2} \tr\big[\rho\,(N+\1)^{k/2}\big]+\frac{l}{2}\mu_k^{(l)}\,,
			\end{align*}
			where $\mu_k^{(l)}=\Delta_l^\nu\left(\frac{(\nu-1)^{\nu-1}}{\nu^\nu}\right)$ with $\nu=l+\frac{k}{2}-1$ and $\Delta_l=(l+1)l+2|\alpha|^lkl^{k/2 - 1}\sqrt{l!}$. Therefore, $\cL_l\coloneqq \cL[a^l-\alpha^l]$ generates a Sobolev and positivity preserving quantum Markov semigroup satisfying for all states $\rho\in W^{k,1}$
			\begin{equation*}
				\|e^{t\cL_l}(\rho)\|_{W^{k,1}}\leq\max\Big\{\|\rho\|_{W^{k,1}}, \mu_k^{(l)}\Big\}\,.
			\end{equation*}
		\end{lem}
		\begin{proof}
			By \Cref{eq:symmetry-function}, we have for $f(x)=(x+1)^{k/2} 1_{x\ge -1}$:
			\begin{equation*}
				\begin{aligned}
					\cL[a_1^l-\alpha^l]^\dagger(f(N))&=(\ad)^lf(N)a^l-\frac{1}{2}\Big((\ad)^la^lf(N)+f(N)(\ad)^la^l\Big)\\
					&\quad+\frac{1}{2}(\overline{\alpha}^la^lf(N)-\overline{\alpha}^lf(N)a^l+\alpha^l f(N)(\ad)^l-\alpha^l(\ad)^lf(N))\\
					&=(\ad)^la^l\Big(f(N-l\1 )-f(N)\Big)\\
					&\quad+\frac{1}{2}\left[\overline{\alpha}^la^l\Big(f(N)-f(N-l\1)\Big)+\alpha^l\Big(f(N)-f(N-l\1)\Big)(\ad)^l\right]\,.
				\end{aligned}
			\end{equation*}
			In what follows, we use the function defined in \Cref{eq-appx:f-g-l-function}
			\begin{equation*}
				g_l(x) = \begin{cases}
					f(x) - f(x - l) & x \ge l;\\
					f(x) & l > x \ge 0;\\
					0 & 0 > x\,.
				\end{cases}
			\end{equation*}  
			Using the canonical commutation relation to write $(a^\dagger)^la^l$ as a function of $N$ (cf.~Lemma \ref{lem:l-ccr}) and with help of the notation
			\begin{equation}\label{eq:notation-product}
				N_k[r:j]\coloneqq (N_k+r\1)\cdots (N_k+j\1)
			\end{equation}
			with the convention $N_k[r:j]=\1$ whenever $r>j$, we thus have that
			\begin{align*}
				\tr\big[\rho\,\cL[a^l-\alpha^l]^\dagger(f(N))\big]&=-\tr\big[\rho\,N[-l+1:0]g_l(N)\big]+\frac{1}{2}\tr\big[\rho\,(\overline{\alpha}^la^lg_l(N)+\alpha^lg_l(N)(\ad)^l)\big]\,,
			\end{align*}
			Since $g_l$ is positive and increasing, the last term above can be upper bounded by Lemma \ref{lem:two-point-hamiltonian-bound},
			\begin{equation*}
				\begin{aligned}
					\frac{1}{2}\tr\big[\rho(\overline{\alpha}^la^lg_l(N)+{\alpha}^lg_l(N)(\ad)^l)\big]&{\leq}\,{|\alpha|^l}\tr[\rho\, g_l(N+l\1)\sqrt{N[1:l]}] \\
					&\overset{(1)}{\leq}|\alpha|^l \frac{kl^{k/2}}{2}\,\tr[\rho\, (N+\1)^{k/2-1}\sqrt{N[1:l]}]\\
					&\leq|\alpha|^lkl^{k/2}\sqrt{l!}\,\tr[\rho\, (N+\1 )^{k/2-1+\frac{l}{2}}]\,.
				\end{aligned}
			\end{equation*}
			In $(1)$ above, we used Lemma \ref{lem:upper-lower-bound-gl} for the bound
			\begin{align*}
				g_l(N+l\1) \le  \frac{kl}{2}\,(N+l\1)^{k/2-1}\leq\frac{kl^{k/2}}{2}\,(N+\1)^{k/2-1}\,.
			\end{align*}
			Therefore, we have proven that 
			\begin{equation*}
				\begin{aligned}
					\tr\big[\rho\cL[a^l-\alpha^l]^\dagger(f(N))\big]&\le -\tr\big[\rho\,N[-l+1:0]\,g_l(N)\big]\\
					&\qquad+|\alpha|^lkl^{k/2}\sqrt{l!}\,\tr[\rho\, (N+\1 )^{k/2-1+\frac{l}{2}}]\,.
				\end{aligned}
			\end{equation*}
			Next, we upper bound the first term above 
			\begin{equation*}
				\begin{aligned}
					\tr\big[\rho\,N[-l+1:0]\,g_l(N)\big]&\overset{(3)}{\geq}l\tr\big[\rho\,N[-l+1:0](N+\1)^{k/2-1}\big]\\
					&\overset{(4)}{\geq}l\,\tr\big[\rho\,(N+\1)^{l+k/2-1}\,\big]-\frac{(l+1)l^2}{2}\tr\big[\rho\,(N+\1 )^{l+k/2-2}\big]\,.
				\end{aligned}
			\end{equation*}
			In $(3)$, we used  Lemma \ref{lem:upper-lower-bound-gl} below with the fact that $N[-l+1:0]$ is supported on the Fock states $|n\rangle$ with $n\ge l-1$;
			in $(4)$ we used that 
			\begin{align*}
				N[-l+1:0]&=\sum_{n\ge 0}\,(n-l+1)\dots n \,\ketbra{n}{n}\\
				&=\sum_{n\ge l}\,(n-l+1)\dots n\,\ketbra{n}{n}\\
				&\overset{(5)}{\ge} \sum_{n\ge l}\left((n+1)^l-\frac{(l+1)l}{2}(n+1)^{l-1}\right)\ketbra{n}{n}\\
				&\ge (N+\1)^l-\frac{(l+1)l}{2}\,(N+\1)^{l-1}\,,
			\end{align*}
			where $(5)$ comes from Lemma \ref{lem:bounds-ccr-l-product} below, whereas the last inequality follows from the fact that $l\ge 2$. To sum up, we showed that
			\begin{equation}\label{eq-ex:l-dissipation-upper-bound}
				\begin{aligned}
					\tr\big[\cL[a^l-\alpha^l](\rho)(f(N))\big]&\le -l\tr\big[\rho\,(N+\1)^{l+k/2-1}\,\big]+\frac{(l+1)l^2}{2}\tr\big[\rho\,(N+\1 )^{l+k/2-2}\big]\\
					&\qquad+|\alpha|^lkl^{k/2}\sqrt{l!}\,\tr[\rho\, (N+\1 )^{k/2-1+\frac{l}{2}}]\\
					&\le -l\tr\big[\rho\,(N+\1)^{l+k/2-1}\,\big]\\
					&\qquad+\frac{l}{2}\biggl(\underbrace{{(l+1)l}+2|\alpha|^lkl^{k/2 - 1}\sqrt{l!}}_{\eqqcolon \Delta_l}\biggr)\,\tr\big[\rho\,(N+\1 )^{l+k/2-2}\big]
				\end{aligned}
			\end{equation}
			where we used again that $l\ge 2$ in the last inequality. Half of the leading order term can be used to control the second term by a constant. For that, we use the spectral decomposition of the operator $N$ so that the above problem can be reduced to the following simple optimization: 
			\begin{equation}\label{eq:optimization}
				\sup_{x\geq0}\left(-x^\nu+\Delta_lx^{\nu-1}\right)=\Delta_l^\nu\left(\frac{(\nu-1)^{\nu-1}}{\nu^\nu}\right)
			\end{equation}
			for $\nu\geq1$ defined as 
			\begin{equation*}
				\nu=l+\frac{k}{2}-1\,.
			\end{equation*}
			The result follows after invoking Proposition \ref{prop-ex:uniformly-bounded-semigroup}. 
		\end{proof}
		\begin{rmk}\label{rmk:l-diss-multi-mode}
			The single-mode bound proved above can be generalized to the multi-mode setting, with generated given for some $\alpha_j\in\mathbb{C}$, $j\in[m]$, by 
			\begin{equation*}
				\cL_l^{(m)}\coloneqq\sum_{j=1}^m\cL[a_j^l-\alpha_j^l]\,.
			\end{equation*}
			Since all the bounds used in the proof of Lemma \ref{lem:l-diss} were derived at the operator level, we directly get for $\k \in \N^m$
			\begin{align*}
				\tr\big[\cL_l^{(m)}(\rho)(N+\1)^{\k/2}\big]&\le \sum_{i=1}^m-\frac{l}{2}\tr\big[\rho\,(N_i+\1)^{l-1}(N + \1)^{\k/2}\big]+\mu^{(l)}_{k_i}\tr\big[\rho\,\prod_{j\neq i}(N_j+\1)^{k_j/2}\big]
			\end{align*}
		\end{rmk}
		For later references, we single out the case $l=2$.
		\begin{cor}[$2$-photon dissipation]\label{lem:cat-identity}
			For any integers $k\ge 1$, $\alpha\in\mathbb{C}$ and any state $\rho\in\cT_f$, 
			\begin{align*}
				\tr\big[\cL[a^2-\alpha^2](\rho)(N+\1)^{k/2}\big]&\le -\tr\big[\rho\,(N+\1)^{k/2}\big]+\mu^{(2)}_k
			\end{align*}
			where $\mu^{(2)}_k=(\Delta^{(2)}_k)^\nu\left(\frac{(\nu-1)^{\nu-1}}{\nu^\nu}\right)$ with $\nu=\frac{k}{2}+1$ and $\Delta_2=6+2\sqrt{2!}|\alpha|^2k2^{k/2 - 1}$. Therefore, $\cL_2\coloneqq \cL[a^2-\alpha^2 ]$ generates a Sobolev and positivity preserving quantum Markov semigroup which satisfies for all states $\rho$ in $W^{k,1}$
			\begin{equation*}
				\|e^{t\cL_2}(\rho)\|_{W^{k,1}}\leq\max\left\{\|\rho\|_{W^{k,1}},\mu^{(2)}_k\right\}\,.
			\end{equation*}
		\end{cor}
		From the above bounds, we directly get the property of Sobolev preservation for the $X$-gate:
		\begin{cor}[$X$-gate]
			For any $T>0$, $\alpha\in\mathbb{C}$, $k\in\N$ and all states $\rho\in\cT_f$,
			\begin{equation*}
				\tr[\cL[a^2-e^{2i\pi t/T}\alpha^2](\rho)(N+\1)^{k/2}]\leq-\tr[\rho(N+\1)^{k/2}]+\mu_k^{(2)}\,,
			\end{equation*}
			where $\mu_k^{(2)}$ is defined in Lemma \ref{lem:cat-identity}. Therefore, $\cL[a^2-e^{2i\pi t/T}\alpha^2]$ generates a Sobolev and positivity preserving quantum evolution system $\cP_{t,t_0}$ which satisfies for all states $\rho\in W^{k,1}$
			\begin{equation*}
				\|\cP_{t,t_0}(\rho)\|_{W^{k,1}}\leq\max\left\{\|\rho\|_{W^{k,1}},\mu_k^{(2)}\right\}\,.
			\end{equation*}
		\end{cor}
		\begin{proof}
			The statement directly follows from Lemma \ref{lem:cat-identity} and $|e^{2i\pi t/T}\alpha^2|=|\alpha^2|$
		\end{proof}
		
		Since \Cref{eq:examples-assum2-step1}, which implies Assumption \ref{assum:sobolev-stability} (or \ref{assum:multi-mode-sobolev-stability-time-dep}), is linear in the supposed generators and the leading order found in \Cref{eq-ex:l-dissipation-upper-bound} has the ability to suppress smaller terms, certain Hamiltonians can be regularized in the sense of \Cref{eq:examples-assum2-step1} by adding an $l$-photon dissipation. Especially, we consider a Hamiltonian of degree $d_H = 2(l-1)$ with the structure: for $\lambda_{i,j}\in\mathbb{C}$ with $\max_{i,j}|\lambda_{i,j}|=\Lambda$,
		\begin{equation}\label{Hpolyrep}
			H=p(a,\ad)=\sum_{\substack{i \le j\\i+j \le d_H}}\lambda_{i,j}a^i(\ad)^j + \overline{\lambda_{i,j}} a^j (\ad)^i \, . 
		\end{equation}
		Note that any monomial in $a,\ad$ of degree at most $d_H$ can be achieved from the representation above thanks to the CCR.
		
		\begin{lem}\label{lem:l-diss-hamiltonian}
			Let $\cL_l\coloneqq \cL[a^l-\alpha^l]$, $\alpha\in\mathbb{C}$, be the $l$-photon dissipation and $H$ as in \eqref{Hpolyrep}. Then, for all states $\rho\in\cT_f$
			\begin{equation}
				\begin{aligned}
					\tr[(\cL_l+\cH[H])(\rho)(N+\1)^{k/2}] &\leq-\frac{l}{2}\,\tr[\rho(N+\1)^{k/2}]+\frac{l}{2}\mu_k\,.\label{eqdiffLH}
				\end{aligned}
			\end{equation}
			for $\mu_k,\nu\geq1$ defined by 
			\begin{equation*}
				\mu_k=c^\nu\left(\frac{(\nu-1)^{\nu-1}}{\nu^\nu}\right)\quad\text{with}\quad c={(l+1)l}+2|\alpha|^lkl^{k/2-1}\sqrt{l!}+\Lambda(2l)^{k/2}\sqrt{(2l)!}\,,\quad\nu=l+\frac{k}{2}-1\,.
			\end{equation*}
			Therefore, $\cL_l+\cH[H]$ generates a Sobolev and positivity preserving quantum Markov semigroup which satisfies for all states $\rho\in W^{k,1}$
			\begin{equation}\label{eqintegratedLH}
				\|e^{t(\cL_l+\cH[H])}(\rho)\|_{W^{k,1}}\leq\max\Big\{\|\rho\|_{W^{k,1}},\mu_k\Big\}\,.
			\end{equation}
		\end{lem}
		\begin{proof}
			We reuse the bound given in \Cref{eq-ex:l-dissipation-upper-bound}:
			\begin{equation*}
				\begin{aligned}
					\tr[\rho\cL_l^\dagger(f(N))]\leq-l\tr\big[\rho\,(N+\1)^{l+k/2-1}\,\big]+\frac{l}{2}\Delta_l\,\tr\big[\rho\,(N+\1 )^{l+k/2-2}g_l(N)\big]\,,
				\end{aligned}
			\end{equation*}
			where $f(x)=(x+1)^{k/2} 1_{x\ge -1}$ and $\Delta_l={(l+1)l}+2|\alpha|^lkl^{k/2 - 1}\sqrt{l!}$\,. To upper bound
			\begin{equation*}
				\begin{aligned}
					\tr[\cH[H](\rho)(N+\1)^{k/2}]&=i\tr[{\rho[(N+\1)^{k/2},H]}]\,,
				\end{aligned}
			\end{equation*}
			we define $g_u$ similarly to \Cref{eq:f-g-l-function} by 
			\begin{equation*}
				g_u(x) = \begin{cases}
					f(x) - f(x - u) & x \ge u-1;\\
					f(x) & u-1 > x \ge 0;\\
					0 & 0 > x\,.
				\end{cases}
			\end{equation*}
			For $d_H=0$ the bound is trivial, so we assume $d_H\geq1$. Then, we compute
			\begin{equation*}
				\begin{aligned}
					i[&f(N),H]\\
					&=i\sum_{\substack{0\leq j< i\\0<i+j\leq d_H}}f(N)(\lambda_{i,j}(\ad)^ia^j+\overline{\lambda_{i,j}}(\ad)^ja^i)-(\lambda_{i,j}(\ad)^ia^j+\overline{\lambda_{i,j}}(\ad)^ja^i)f(N)\\
					&=i\sum_{\substack{0\leq j< i\\0<i+j\leq d_H}}\lambda_{i,j}f(N)N[-i+1:-i+j](\ad)^{i-j}+\overline{\lambda_{i, j}}a^{i-j}f(N-i+j)N[-i+1:-i+j]\\
					&\qquad\qquad-\lambda_{i,j}N[-i+1:-i+j]f(N-i+j)(\ad)^{i-j}-\overline{\lambda_{i, j}}a^{i-j}N[-i+1:-i+j]f(N)\\
					&=i\sum_{\substack{0< r\leq i\\0<2i-r\leq d_H}}-\overline{\lambda_{i,i-r}}a^{r}N[-i+1:-r]g_r(N)+\lambda_{i,i-r}g_r(N)N[-i+1:-r](\ad)^{r}\\
					&\overset{(1)}{\leq} \sum_{\substack{0< r\leq i\\0<2i-r\leq d_H}}2|\lambda_{i,i-r}|\sqrt{(N+\1)\cdots(N+r\1)}g_{r}(N+r\1)N[r-i+1:0]\\
					&\overset{(2)}{\leq} \sum_{\substack{0< r\leq i\\0<2i-r\leq d_H}}2\sqrt{r!}|\lambda_{i,i-r}|g_{r}(N+r\1)(N+\1)^{i-r/2}\\
				\end{aligned}
			\end{equation*}
			\begin{equation*}
				\begin{aligned}
					i\tr[[H,\rho](N+\1)^{k/2}]&\leq 2\Lambda\sum_{i=1}^{d_H}\sum_{r=1}^i\sqrt{r!}\tr[\rho g_{r}(N+r)(N+\1)^{i-r/2}]\\
					&\overset{(3)}{\leq}2\Lambda\sum_{i=1}^{d_H}\sum_{r=1}^i\sqrt{r!}r^{k/2-1}\tr[\rho (N+\1)^{k/2+i-r/2-1}]\\
					&\leq \Lambda (d_H + 1)d_H\sqrt{d_H!}d_H^{k/2-1}\tr[\rho (N+\1)^{k/2+d_H/2-1}]\,,
				\end{aligned}
			\end{equation*}
			where we used Lemma \ref{lem:upper-lower-bound-gl} in $(3)$. As the above function is monotone in $d_H$ we can w.l.o.g assume $d_H=2(l-1)$ and conclude
			\begin{equation*}
				\begin{aligned}
					\tr[(\cL_l+\cH[H])(f(N))]&\leq-l\tr\big[\rho\,(N+\1)^{l+k/2-1}\,\big]\\
					&\qquad\qquad+\frac{l}{2}\left(\Delta_l+\Lambda(2l)^{k/2}\sqrt{(2l)!}\right)\tr[\rho (N)(N+\1)^{l+k/2-2}]\,.
				\end{aligned}
			\end{equation*}
			The same optimization as in \Cref{eq:optimization} provides inequality \eqref{eqdiffLH}. Inequality \eqref{eqintegratedLH} follows after invoking Proposition \ref{prop-ex:uniformly-bounded-semigroup}.
		\end{proof}
		The case $l=2$ deals with the sum of a $2$-photon dissipation with the displacement operator $i[a+\ad,\cdot]$ used in the construction of the $Z(\theta)$-gate \cite{Mirrahimi.2014}. In this specific case, one can improve the error bound in the following way:
		\begin{lem}[$Z(\theta)$-gate]\label{lem:2-photon-diss-displacement}
			For any state $\rho\in\cT_f$, $\alpha\in\mathbb{C}$, $\varepsilon>0$ and $k\in\N$ 
			\begin{equation*}
				\tr[(\varepsilon\cH[a+\ad]+\cL[a^2+\alpha^2])(\rho)(N+\1)^{k/2}]\leq -\,\tr\big[\rho\, (N+\1)^{k/2}\big]+\mu_k\,.
			\end{equation*}
			where $\mu_k\geq0$ is defined by 
			\begin{equation*}
				\mu_k=(\Delta_2+\varepsilon4k)^\nu\left(\frac{(\nu-1)^{\nu-1}}{\nu^\nu}\right)\qquad\text{with}\qquad\nu=\frac{k}{2}+1\,.
			\end{equation*}
			Therefore, $\varepsilon\cH[a+\ad]+\cL[a^2+\alpha^2]$ generates a Sobolev and positivity preserving quantum Markov semigroup which satisfies for all states $\rho\in W^{k,1}$
			\begin{equation}\label{etlboundfff}
				\|e^{t(\varepsilon\cH[a+\ad]+\cL[a^2+\alpha^2])}(\rho)\|_{W^{k,1}}\leq\max\Big\{\|\rho\|_{W^{k,1}},\mu_k\Big\}\,.
			\end{equation}
		\end{lem}
		\begin{proof}
			By \Cref{eq-ex:l-dissipation-upper-bound} in Lemma \ref{lem:l-diss},
			\begin{equation*}
				\begin{aligned}
					\tr\big[\cL[a^2-\alpha^2](\rho)(f(N))\big]&\le -2\tr\big[\rho\,(N+\1)^{k/2+1}\,\big]\\
					&\qquad+\biggl(\underbrace{6+2|\alpha|^2k2^{k/2 - 1}\sqrt{2}}_{\eqqcolon \Delta_2}\biggr)\,\tr\big[\rho\,(N+\1 )^{k/2}\big]
				\end{aligned}
			\end{equation*}
			where $f(x)=(x+1)^{k/2} 1_{x\ge -1}$. Next, by \Cref{eq:symmetry-function}, Lemma \ref{lem:two-point-hamiltonian-bound} and Lemma \ref{lem:upper-lower-bound-gl}, we have that
			\begin{equation}\label{eq:sobolevstab-displacement}
				\begin{aligned}
					\tr[\cH[a+\ad](\rho)f(N)]&=i\tr[\rho\left(f(N)(a+\ad)-(a+\ad)f(N)\right)]\\
					&=\tr[\rho\left(-iag_1(N)+ig_1(N)\ad \right)]\\
					&\leq 2\,\tr[\rho g_1(N+\1 )\sqrt{N+\1}]\\
					&\leq 2k\tr[\rho(N+\1)^{k/2-\frac{1}{2}}]\,,
				\end{aligned}
			\end{equation}
			where we recall that
			\begin{equation}\label{eq-appx:f-g-l-functionlequal1}
				g_1(x) = \begin{cases}
					f(x) - f(x - 1) & x \ge 0;\\
					0 & 0 > x\,.
				\end{cases}
			\end{equation}
			Thus,
			\begin{equation*}
				\begin{aligned}
					\tr\big[(\epsilon\cH[a+a^\dagger]&+\cL[a^2-\alpha^2])(\rho)(f(N))\big]\\
                    &\le -2\tr\big[\rho\,(N+\1)^{k/2+1}\,\big]+\left(\Delta_2+\varepsilon2k\right)\,\tr\big[\rho\,(N+\1 )^{k/2}\big].
				\end{aligned}
			\end{equation*}
			for $\nu\geq1$ defined as 
			\begin{equation*}
				\nu=\frac{k}{2}+1
			\end{equation*}
			ends the proof of the differential upper bound, and \eqref{etlboundfff} follows from Proposition \ref{prop-ex:uniformly-bounded-semigroup}.
		\end{proof}
		Interestingly, \Cref{eq:sobolevstab-displacement} shows directly that the displacement operator satisfies Assumption \ref{assum:sobolev-stability} so that Theorem \ref{thm:generation-theorem} can be applied.
		\begin{cor}\label{lem:z-theta-energetic-stab}
			For any state $\rho\in\cT_f$, $\alpha\in\mathbb{C}$, $\varepsilon>0$ and $k\in\N$ 
			\begin{equation*}
				\tr\Bigl[\cH[a+\ad](N+\1)^{k/2}\Bigr]\leq 2k\tr\big[\rho\, (N+\1)^{k/2}\big]
			\end{equation*}
			Therefore, $\cH[a+\ad]$ generates a Sobolev and positivity-preserving quantum Markov semigroup.
		\end{cor}
		Nevertheless, the above result is not of the form given in \Cref{eq:sobolev-preserving-stab-constant} so the improvement \Cref{eq:improved-semigroup-bound} is not applicable. In the context of bosonic error correction, the projective Zeno effect wr.r.t.~the semigroup above and the projection onto the code space (\ref{eq:codespace}) helps to understand the $Z(\theta)$-gate mathematically \cite{Mirrahimi.2014}. A quantitative convergence rate can be proven via \cite{Moebus.2022}. 
		
		\begin{prop}[CNOT-gate]\label{lem:assum2-CNOT}
			For all $\k\coloneqq (k_1,k_2) \in \N^2$ such that
			\begin{equation*}
				32|\alpha|k_12^{k_1/2-1/2}\leq k_2\,,
			\end{equation*}
			there exists a constant $\mu_{\k}$ such that for all states $\rho \in \cT_f$
			\begin{equation*}
				\begin{aligned}
					&\tr[\left(\cL[a^2-\alpha^2]+\cL[b^2-\alpha^2-\frac{\alpha}{2}(1-e^{2i\pi t/T})(a-\alpha)]\right)(\rho)(N_1 + \1)^{k_1/2}(N_2 + \1)^{k_2/2}]\\
					&\qquad\qquad\qquad\qquad\qquad\qquad\qquad\qquad\leq-\frac{1+k_2}{8}\tr[\rho\Bigl((N_1+\1)^{k_1/2}(N_2+\1)^{k_2/2}\Bigr)]+\mu_{\k}\,.
				\end{aligned}
			\end{equation*}
			Therefore, the CNOT-gate generates a Sobolev and positivity preserving quantum Markov semigroup which satisfies for all states $\rho\in W^{\k,1}$
			\begin{equation}\label{lastclaimsobolevbound}
				\|\cP^{\operatorname{CNOT}}_{t,t_0}(\rho)\|_{W^{\k,1}}\leq\max\left\{\|\rho\|_{W^{\k,1}},\frac{8\mu_{\k}}{1+k_2}\right\}\,.
			\end{equation}
			For a general $\k \in \R_+^2$ and $x \in W^{\k, 1}$ one obtains
			\begin{equation*}
				\|\cP^{\operatorname{CNOT}}_{t,t_0}(x)\|_{W^{\k_r,1}}\leq \gamma_k \|x\|_{W^{\k_r,1}} \, ,
			\end{equation*}
			where $\gamma_{\k} = \max\{1,\frac{8\mu_{\k}}{1+k_2}\}$ for $\k \in \{\k_r\}_{r \in \N}$ and an interpolated constant in all other cases. 
		\end{prop}
		\begin{proof}
			We denote $f(x_1,x_2)=f_1(x_1)f(x_2)$ with $f_{1}(x_{1})=(x_{1}+1)^{k_{1}/2} 1_{x_{1}\ge -1}$, $f_{2}(x_{2})=(x_{2}+1)^{k_{2}/2} 1_{x_{2}\ge -1}$, and rewrite the CNOT-generator as
			\begin{equation*}
				\begin{aligned}
					\cL\coloneqq \cL[a^2-\alpha^2]&+\cL[b^2-\alpha^2-\frac{\alpha}{2}(1-e^{2i\pi t/T})(a-\alpha)]=\cL[a^2-\alpha^2]+\cL[b^2+za+w]
				\end{aligned}
			\end{equation*}
			where $z\coloneqq-\frac{\alpha}{2}(1-e^{2i\pi t/T})$ and $w\coloneqq -\alpha(z+\alpha)$. As in the previous proofs, we investigate the action of the adjoint on $f(N)\coloneqq f(N_1,N_2)$
			\begin{equation*}
				\tr[\cL(\rho)f(N)]=\tr[\rho\cL^\dagger(f(N))]\,.
			\end{equation*}
			We first focus on the second Lindbladian $\cL[b^2+za+w]$: we first consider, for $n\coloneqq (n_1,n_2)\in\N^2$, 
			\begin{equation*}
				\begin{aligned}
					&\cL[b^2+za+w]^\dagger(\ketbra{n}{n})\\
					&\quad =((\bd)^2+\overline{z}\ad+\overline{w})\ketbra{n}{n}(b^2+za+w)-\frac{1}{2}\left\{((\bd)^2+\overline{z}\ad+\overline{w})(b^2+za+w),\ketbra{n}{n}\right\}\\
					&\quad =F_1(n)\ketbra{n}{n}+F_2(n) \ketbra{n_1,n_2+2}{n_1,n_2+2}+F_3(n) \ketbra{n_1+1,n_2}{n_1+1,n_2}\\
					&\quad\qquad + \big(F_4(n) \ketbra{n_1,n_2+2}{n_1+1,n_2}+h.c.\big)+\big(F_5(n) \ketbra{n_1+1,n_2-2}{n} + h.c.\big) \\
					&\quad\qquad +\big(F_6(n)\ketbra{n_1,n_2-2}{n}+h.c.\big) +\big( F_7(n)\ketbra{n_1-1,n_2}{n} + h.c.\big)\\
					&\quad\qquad + \big(F_8(n)\ketbra{n_1,n_2+2}{n}+h.c. \big)+\big(F_9(n)\ketbra{n_1+1,n_2}{n}+h.c. \big)  \\
					&\quad\qquad + \big(F_{10}(n)\ketbra{n_1-1,n_2+2}{n}+h.c.\big)\,,
				\end{aligned}
			\end{equation*}
			where the notation $h.c.$ above stands for Hermitian conjugate, $|n\rangle=0$ whenever $n\notin \mathbb{N}^2$ by convention, and where
			\begin{align*}
				&F_1(n)\coloneqq - n_2(n_2-1)-|z|^2n_1\\
				&F_2(n)\coloneqq (n_2+1)(n_2+2) \\
				&F_3(n)\coloneqq |z|^2(n_1+1) \\
				&F_4(n)\coloneqq  z\sqrt{(n_1+1)(n_2+1)(n_2+2)}\\
				&F_5(n)\coloneqq -\frac{1}{2}\overline{z}\sqrt{(n_1+1)n_2(n_2-1)}\\
				&F_6(n)\coloneqq  -\frac{1}{2}\overline{w}\sqrt{n_2(n_2-1)}\\
				&F_7(n)\coloneqq -\frac{1}{2}\overline{w}z\sqrt{n_1}\\
				&F_8(n)\coloneqq \frac{1}{2} w\sqrt{(n_2+1)(n_2+2)}\\
				&F_9(n)\coloneqq \frac{1}{2}\overline{z}w\sqrt{n_1+1}\\
				&F_{10}(n)\coloneqq -\frac{1}{2}z\sqrt{n_1(n_2+1)(n_2+2)}\,.
			\end{align*}
			In the next step, we regroup the $17$ terms into terms differing only by a shift:
			\medskip
			\noindent\textit{Case 0:} Diagonal terms, involving $F_1$, $F_2$ and $F_3$,
			\begin{equation*}
				\begin{aligned}
					C_0(n)\coloneqq F_1(n) \ketbra{n}{n}+F_2(n)\ketbra{n_1,n_2+2}{n_1,n_2+2} + F_3(n)\ketbra{n_1+1,n_2}{n_1+1,n_2}\,.
				\end{aligned}
			\end{equation*}
			
			\noindent\textit{Case 1:} Terms of the form $\ketbra{n_1+1,n_2-2}{n_1,n_2}$, involving $\overline{F}_4$, $F_5$ and $\overline{F}_{10}$,
			\begin{equation*}
				\begin{aligned}
					C_1(n)\coloneqq \overline{F}_4(n) &\ketbra{n_1+1,n_2}{n_1,n_2+2} \\
                    &+ F_5(n) \ketbra{n_1+1,n_2-2}{n} + \overline{F}_{10}(n) \ketbra{n}{n_1-1,n_2+2} \,.
				\end{aligned}
			\end{equation*}
			
			\noindent\textit{Case 1':} Terms of the form $\ketbra{n_1,n_2}{n_1+1,n_2-2}$, involving $F4$, $\overline{F}_5$ and $F_{10}$,
			\begin{equation*}
				\begin{aligned}
					C_{1'}(n)\coloneqq {F}_4(n) &\ketbra{n_1,n_2+2}{n_1+1,n_2}\\
                    &+ \overline{F}_5(n) \ketbra{n}{n_1+1,n_2-2} + {F}_{10}(n) \ketbra{n_1-1,n_2+2}{n} \,.
				\end{aligned}
			\end{equation*}
			
			\noindent\textit{Case 2:} Terms of the form $\ketbra{n_1,n_2-2}{n_1,n_2}$, involving $F_6$ and $\overline{F}_8$,
			\begin{equation*}
				\begin{aligned}
					C_{2}(n)\coloneqq {F}_6(n) \ketbra{n_1,n_2-2}{n}+\overline{F}_8(n)\ketbra{n}{n_1,n_2+2} \,.\\
				\end{aligned}
			\end{equation*}
			
			\noindent\textit{Case 2':} Terms of the form $\ketbra{n_1,n_2}{n_1,n_2-2}$, involving $\overline{F}_6$ and ${F}_8$,
			\begin{equation*}
				\begin{aligned}
					C_{2'}(n)\coloneqq \overline{F}_6(n) \ketbra{n}{n_1,n_2-2}+{F}_8(n) \ketbra{n_1,n_2+2}{n} \,.\\
				\end{aligned}
			\end{equation*}
			
			\noindent\textit{Case 3:} Terms of the form $\ketbra{n_1-1,n_2}{n_1,n_2}$, involving $F_7$ and $\overline{F}_9$,
			\begin{equation*}
				\begin{aligned}
					C_{3}(n)\coloneqq {F}_7(n) \ketbra{n_1-1,n_2}{n}+\overline{F}_9(n) \ketbra{n}{n_1+1,n_2} \,.\\
				\end{aligned}
			\end{equation*}
			
			\noindent\textit{Case 3':} Terms of the form $\ketbra{n_1,n_2}{n_1-1,n_2}$, involving $\overline{F}_7$ and ${F}_9$,
			\begin{equation*}
				\begin{aligned}
					C_{3'}(n)\coloneqq \overline{F}_7(n) \ketbra{n}{n_1-1,n_2}+{F}_9(n)\ketbra{n_1+1,n_2}{n} \,.\\
				\end{aligned}
			\end{equation*}
			To summarize, we have decomposed $\cL[b^2+za+w]^\dagger(\ketbra{n}{n})$ into the sum 
			\begin{align}\label{eqLnn}
				\cL[b^2+za+w]^\dagger(|n\rangle\langle n|)=C_0(n)+C_1(n)+C_{1'}(n)+C_2(n)+C_{2'}(n)+C_3(n)+C_{3'}(n)\,.
			\end{align}
			Next, we introduce the functions $g_{j,l}:\mathbb{N}\to \mathbb{R}$, $j\in\{1,2\}$, $l\in\mathbb{N}$, as 
			\begin{equation*}
				g_{j, l}(x) = \begin{cases}
					f_j(x) - f_j(x - l) & x \ge l;\\
					f_j(x) & l > x \ge 0;\\
					0 & 0 > x\,.
				\end{cases}
			\end{equation*}
			Multiplying \Cref{eqLnn} by $f(n)$ and summing over $n\in\mathbb{N}^2$, we find that 
			\begin{align*}
				\cL[b^2+za+w]^\dagger(f(N))=\hat{C}_0+\hat{C}_1+\hat{C}_{1'}+\hat{C}_2+\hat{C}_{2'} +\hat{C}_3+\hat{C}_{3'}\,,
			\end{align*}
			with 
			\begin{align*}
				&\hat{C}_0\coloneqq \sum_{n}\,C_0(n)=\sum_n -\,\Big(f_1(n_1)g_{2,2}(n_2)n_2(n_2-1)+f_2(n_2)g_{1,1}(n_1)|z|^2n_1\Big) \ketbra{n}{n}\\
				&\hat{C}_1\coloneqq \sum_{n}\,C_1(n)\\
				&\quad =\sum_{n}\,-\frac{\overline{z}}{2}\sqrt{(n_1+1)(n_2-1)n_2}\Bigl(f_1(n_1)g_{2,2}(n_2)+g_{1,1}(n_1+1)f_2(n_2-2)\Bigr)\ketbra{n_1+1,n_2-2}{n}\\
				&\hat{C}_{1'}\coloneqq \hat{C}_1^\dagger \\
				&\quad =\sum_n\,-\frac{{z}}{2}\sqrt{(n_1+1)(n_2-1)n_2}\Bigl(f_1(n_1)g_{2,2}(n_2)+g_{1,1}(n_1+1)f_2(n_2-2)\Bigr)\ketbra{n}{n_1+1,n_2-2}\\
				&\hat{C}_2\coloneqq \sum_n\,C_2(n) = \sum_n-\frac{\overline{w}}{2}\sqrt{(n_2-1)n_2}f_1(n_1)g_{2,2}(n_2)\ketbra{n_1,n_2-2}{n}\\
				&\hat{C}_{2'}\coloneqq \hat{C}_2^\dagger= \sum_n-\frac{{w}}{2}\sqrt{(n_2-1)n_2}f_1(n_1)g_{2,2}(n_2)\ketbra{n}{n_1,n_2-2}\\
				&\hat{C}_3\coloneqq \sum_n\,C_3(n) = \sum_n-\frac{\overline{w}z}{2}\sqrt{n_1}g_{1,1}(n_1)f_2(n_2)\ketbra{n_1-1,n_2}{n}\\
				&\hat{C}_{3'}\coloneqq \hat{C}_3^\dagger =\sum_n-\frac{{w}\overline{z}}{2}\sqrt{n_1}g_{1,1}(n_1)f_2(n_2)\ketbra{n}{n_1-1,n_2}\,.
			\end{align*}
			We will use an upper bound on $\hat{C}_0(n)$ in what follows:    
			\begin{equation*}
				\begin{aligned}
					&C_0(n) =- \Big(f_1(n_1)g_{2,2}(n_2)(n_2-1)n_2+|z|^2g_{1,1}(n_1)f_2(n_2)n_1\Big)\\
					&\le -\Big( f_1(n_1)g_{2,2}(n_2)1_{n_2\geq2}((n_2+1)^{2}-3(n_2+1))+|z|^2g_{1,1}(n_1)f_2(n_2)n_1\Big)\\
					&\overset{(1)}{\le} -\Bigl\{ k_2f_1(n_1)(n_2+1)^{k_2/2}1_{n_2\geq2}\Big((n_2+1)-1_{k_2\geq3}\frac{k_2}{2}-6\Big) \\
                    &\qquad+ |z|^2 1_{n_1\geq1}(n_1+1)^{k_1/2-1}f_2(n_2)n_1
					\Bigr\}\\
					&\equiv C_{0'}(n)\,,
				\end{aligned}
			\end{equation*}
			where $(1)$ follows from Lemma \ref{lem:upper-lower-bound-gl}. We denote this upper bound by $\hat{C}_{0'}=\sum_n C_{0'}(n)\ketbra{n}{n}$. Recall that, by Lemma \ref{lem:l-diss} and Remark \ref{rmk:l-diss-multi-mode},
			\begin{equation}\label{eq-ex:CNOT-diss-rep}
				\begin{aligned}
					\cL[a^2-\alpha^2]^\dagger(f(N))&\leq - \,(N_1+\1)^{k_1/2+1}(N_2+\1)^{k_2/2}
					+\mu_{k_1}^{(2)}\, (N_2+\1)^{k_2/2} \\
					&=\sum_{n} \Big(-(n_1+1)f(n)+\mu_{k_1}^{(2)}\,f_2(n_2)\Big)\ketbra{n}{n}=:\hat{C}_{4}\,,
				\end{aligned}
			\end{equation}
			where $\mu_{k_1}^{(2)}=\Delta_2^\nu\left(\frac{(\nu-1)^{\nu-1}}{\nu^\nu}\right)$ with $\nu=\frac{k_1}{2}+1$ and $\Delta_2=6+2|\alpha|^2k_12^{k_1/2-1}\sqrt{2}$. Therefore, the diagonal contribution of $\cL^\dagger(f(N))$ can be controlled by 
			\begin{equation*}
				\begin{aligned}
					( {C}_{0'}+{C}_4 )(n) &\coloneqq -\biggl(k_2f_1(n_1)(n_2+1)^{k_2/2+1}1_{n_2\geq2}+|z|^21_{n_1\geq1}f(n)+(n_1+1)^{k_1/2+1}f_2(n_2)\biggr)\\
					&\qquad+k_2f_1(n_1)(n_2+1)^{k_2/2}1_{n_2\geq2}\biggl(1_{k_2\geq3}\frac{k_2}{2}+6\biggr)\\
					&\qquad+|z|^21_{n_1\geq1}(n_1+1)^{k_2/2-1}f_2(n_2)\\
					&\qquad+\mu_{k_1}^{(2)}(n_2+1)^{k_2/2}\\
					&\le -\frac{1}{2}\biggl(k_2f_1(n_1)(n_2+1)^{k_2/2+1}+(n_1+1)^{k_1/2+1}f_2(n_2)\biggr)+\Delta_0\\
					&\eqqcolon -x(n)
				\end{aligned}
			\end{equation*}
			where the constant $\Delta_0$ is achieved by splitting off half of the negative leading order terms to control the lower order positive contributions (see for example the proof of Lemma \ref{lem:l-diss}).
			Next, we consider operators of the form
			\begin{equation}\label{eq:2x2}
				-x_1\ketbra{e_1}{e_1}-x_2\ketbra{e_2}{e_2}+y\ketbra{e_2}{e_1}+\overline{y}\ketbra{e_1}{e_2}\,,
			\end{equation}
			where $\{e_1, e_2\}$ forms an orthonormal basis of a two-dimensional Hilbert space and $x_1, x_2 \in \R, y \in \C$. The operator in \Cref{eq:2x2} has the eigenvalues
			\begin{equation}
				\lambda_+ = \frac{-x_1 - x_2 + \sqrt{(x_1 - x_2)^2 + 4|y|^2}}{2}, \quad \lambda_- = \frac{-x_1 - x_2 - \sqrt{(x_1 - x_2)^2 + 4|y|^2}}{2}
			\end{equation}
			Moreover, 
			\begin{equation*}
				\frac{-x_1-x_2+\sqrt{(x_1-x_2)^2+4|y|^2}}{2}\leq -\min\{x_1,x_2\}+|y|\,.
			\end{equation*}
			Using this bound, we control each of the off-diagonal operators $\hat{C}_i+\hat{C}_{i'}$, $i\in\{1,2,3\}$, in terms of $\frac{1}{4}X$, where $X=\sum_n x(n) \ketbra{n}{n}$. 
			\begin{equation*}
				\begin{aligned}
					-\frac{1}{4}X+\hat{C}_1+\hat{C}_{1'}&\coloneqq \sum_n-\frac{1}{4}x(n)\ketbra{n}{n}+y_1\ketbra{n}{n_1+1,n_2-2}+\overline{y}_1\ketbra{n_1+1,n_2-2}{n}\\
					&\le \sum_{n|n_2\ge 2}-\frac{1}{8}x(n)\ketbra{n}{n}-\frac{1}{8}x(n_1+1,n_2-2)\ketbra{n_1+1,n_2-2}{n_1+1,n_2-2}\\
					&\quad\qquad + y_1\ketbra{n}{n_1+1,n_2-2}+\overline{y}_1\ketbra{n_1+1,n_2-2}{n}
				\end{aligned}
			\end{equation*}
			with
			\begin{equation*}
				y_1 = y_1(n_1,n_2) = -\frac{{z}}{2}\sqrt{(n_1+1)(n_2-1)n_2}\Bigl(f_1(n_1)g_{2,2}(n_2)+g_{1,1}(n_1+1)f_2(n_2-2)\Bigr)
			\end{equation*}
			so that
			\begin{equation}\label{eqcaseC1C1'}
				-  \frac{1}{4}X+\hat{C}_1+\hat{C}_{1'}\leq \sum_{n|n_2\ge 2}\left(-\min\{x_1,x_2\}+|y_1|\right)(\ketbra{n}{n} + \ketbra{n_1 + 1, n_2 - 2}{n_1 + 1, n_2 - 2})
			\end{equation}
			where $x_1=\frac{1}{8}x(n_1,n_2)$ and $x_2=\frac{1}{8}x(n_1+1,n_2-2)$. Moreover, for $n_2\geq2$
			\begin{equation*}\label{eq-ex:cnot-matrix-bound}
				\begin{aligned}
					|y_1|&\overset{(1)}{\leq} \frac{|z|}{2}\sqrt{(n_1+1)(n_2-1)n_2}\Bigl(f_1(n_1)2k_2(n_2+1)^{k_2/2-1}+k_1(n_1+2)^{k_1/2-1}f_2(n_2-2)\Bigr)\\
					&\leq|z|k_2(n_1+1)^{k_1/2+1/2}(n_2+1)^{k_2/2}\\
                    &\qquad\qquad+\frac{|z|k_1}{2}\sqrt{(n_2-1)^2+n_2-1}\,(n_1+2)^{k_1/2-1/2}f_2(n_2-2)\\
					&\leq|z|k_2(n_1+1)^{k_1/2+1/2}(n_2+1)^{k_2/2}+|z|k_1\,(n_1+2)^{k_1/2-1/2}(n_2-1)^{k_2/2+1}\,,
				\end{aligned}
			\end{equation*}
			where $(1)$ follows from  Lemma \ref{lem:upper-lower-bound-gl}. At this stage, we consider two cases: 
			
			\medskip
			
			\noindent Case (i): $x_2\ge x_1$. In that case, $-\min\{x_1,x_2\}+|y_1|=-x_1+|y_1|$, and therefore
			\begin{equation*}
				\begin{aligned}
					-\min\{x_1,x_2\}+|y_1|&\leq-\frac{1}{16}\biggl(k_2f_1(n_1)(n_2+1)^{k_2/2+1}+(n_1+1)^{k_1/2+1}f_2(n_2)\biggr)+\frac{1}{8}\Delta_0\\
					&\qquad\qquad+\underbrace{|z|k_2(n_1+1)^{k_1/2+1/2}(n_2+1)^{k_2/2}}_{=:A_1}\\
                    &\qquad\qquad+\underbrace{|z|k_1\,(n_1+2)^{k_1/2-1/2}(n_2-1)^{k_2/2+1}}_{=:A_2}\,.
				\end{aligned}
			\end{equation*}
			Note that the first positive non-constant term $A_1$ can be controlled with half the negative contribution in the first term by a constant using the same type of polynomial optimization as in the proof of Lemma \ref{lem:l-diss}. For the last term, i.e.~$A_2$, we use the assumption 
			\begin{equation}\label{assumptionequationok}
				|z|k_12^{k_1/2-1/2}\leq |\alpha|k_12^{k_1/2-1/2}\leq\frac{1}{32}k_2,
			\end{equation}
			which allows us to control $A_2$ with the other half of the first term, as we already did with $A_1$. Recall the definition $z=-\frac{\alpha}{2}(1-e^{2i\pi t/T})$. Summarising the above considerations we can conclude the existence of a constant $\tilde{\Delta}'_1$ such that 
			\begin{equation*}
				\begin{aligned}
					-\min\{x_1,x_2\}+|y_1|\leq\tilde{\Delta}'_1\,.
				\end{aligned}
			\end{equation*}
			
			\medskip
			
			\noindent Case (ii): $x_2\leq x_1$. In that case $-\min\{x_1,x_2\}+|y_1|=-x_2+|y_1|$, and therefore
			\begin{equation*}
				\begin{aligned}
					-\min\{x_1,x_2\}+|y_1|&=-\frac{1}{16}\biggl(k_2f_1(n_1+1)(n_2-1)^{k_2/2+1}+(n_1+2)^{k_1/2+1}f_2(n_2-2)\biggr)+\frac{1}{8}\Delta_0\\
					&\qquad\qquad+|z|k_2(n_1+1)^{k_1/2+1/2}(n_2+1)^{k_2/2}\\
                    &\qquad\qquad+|z|k_1\,(n_1+2)^{k_1/2-1/2}(n_2-1)^{k_2/2+1}\,.
				\end{aligned}
			\end{equation*}
			To upper bound the above, we use again the assumption \eqref{assumptionequationok}, which implies the existence of a constant $\tilde{\Delta}'_1$ such that
			\begin{equation*}
				\begin{aligned}
					-\min\{x_1,x_2\}+|y_1|&\leq\tilde{\Delta}'_1\,.
				\end{aligned}
			\end{equation*}
			Combining cases (i) and (ii) above, denoting $\Delta_1\coloneqq \max\{\Tilde{\Delta}_1,\Tilde{\Delta}_1'\}$ and plugging the bounds into \eqref{eqcaseC1C1'}, we arrive at
			\begin{align}\label{lastC1}
				-\frac{1}{4}X+\hat{C}_1+\hat{C}_{1'}\le \Delta_1 \sum\limits_{n|n_2 \ge 2}(\ketbra{n}{n} + \ketbra{n_1 + 1, n_2 - 2}{n_1 + 1, n_2 - 2})
			\end{align}
			Next, we control $-\frac{1}{4}X+\hat{C}_2+\hat{C}_{2'}$. Here, we have
			\begin{equation*}
				y_2=y_2(n_1,n_2)=-\frac{\overline{w}}{2}\sqrt{(n_2-1)n_2}f_1(n_1)g_{2,2}(n_2)\,,
			\end{equation*}
			$x_1=\frac{1}{8}x(n)$\,, and $x_2=\frac{1}{8}x(n_1,n_2-2)$. By Lemma \ref{lem:upper-lower-bound-gl}, we have that
			\begin{align*}
				|y_2|&\le {|{w}|}\sqrt{(n_2-1)n_2}f_1(n_1)k_2(n_2+1)^{k_2/2-1}\,.
			\end{align*}
			Therefore, the negative contribution from $\min\{x_1,x_2\}$ has leading order in both variables $n_1$ and $n_2$, which implies the existence of a constant $\Delta_2$ such that
			\begin{equation*}
				-\min\{x_1,x_2\}+|y_2|\leq\Delta_2\,.
			\end{equation*}
			Hence,
			\begin{align}\label{lastC2}
				-\frac{1}{4}X+\hat{C}_2+\hat{C}_{2'}\le \Delta_2 \sum\limits_{n|n_2 \ge 2}(\ketbra{n}{n} + \ketbra{n_1, n_2 - 2}{n_1, n_2 - 2})
			\end{align}
			Finally, we consider $-\frac{1}{4}X+\hat{C}_3+\hat{C}_{3'}$. In this case, 
			\begin{equation*}
				y_3=y_3(n_1,n_2)=-\frac{\overline{w}z}{2}\sqrt{n_1}g_{1,1}(n_1)f_2(n_2)\,,
			\end{equation*}
			$x_1=\frac{1}{8}x(n)$\,, and $x_2=\frac{1}{8}x(n_1-1,n_2)$. Similarly to the above, we can argue the existence of a constant $\Delta_3$ such that
			\begin{equation*}
				-\max\{x_1,x_2\}+|y_3|\leq\Delta_3\,.
			\end{equation*}
			Hence,
			\begin{align}\label{lastC3}
				-\frac{1}{4}X+\hat{C}_3+\hat{C}_{3'}\le \Delta_3 \sum\limits_{n|n_1 \ge 1} (\ketbra{n}{n} + \ketbra{n_1 - 1, n_2}{n_1 - 1, n_2}
			\end{align}
			Combining \eqref{lastC1}, \eqref{lastC2} and \eqref{lastC3}, we have shown that
			\begin{align*}
				\cL^\dagger(f(N))&\leq -\frac{X}{4}+\Delta_1\sum_{n|n_2\ge 2}(\ketbra{n}{n} + \ketbra{n_1 + 1, n_2 - 2}{n_1 + 1, n_2 - 2})\\
				&+\Delta_2\sum_{n|n_2\ge 2}(\ketbra{n}{n} + \ketbra{n_1, n_2 - 2}{n_1, n_2 - 2}) \\
				&+ \Delta_3 \sum\limits_{n |n_1 \ge 1} (\ketbra{n}{n} + \ketbra{n_1 - 1, n_2}{n_1 - 1, n_2})\\
				&\le -\frac{X}{4}+ 2(\Delta_1+\Delta_2+\Delta_3)\1\\
				&= -\frac{1}{8}\Big(k_2f_1(N_1)(N_2+1)^{k_2/2+1}+(N_1+1)^{k_1/2+1}f_2(N_2)\Big) +\mu_{\k}\,\1\\
				&\le -\frac{1+k_2}{8}f(N)+\mu_{\k}\1\,,
			\end{align*}
			with
			\begin{equation*}
				\mu_{\k}\coloneqq\frac{\Delta_0}{4}+2(\Delta_1+\Delta_2+\Delta_3)\,.
			\end{equation*}
			The claim \eqref{lastclaimsobolevbound} finally follows from Proposition \ref{prop-ex:uniformly-bounded-semigroup}. 
		\end{proof}

\section{Perturbation bounds}\label{sec:example-perturbation-bounds}
	In this section, we establish a perturbative analysis at any time scale for the semigroups considered in Section \ref{sec:examples-sobolev-preserving-semigroup}. In finite dimensions, \cite[Theorem 6]{Szehr.2013} gives a quantitative bound which controls the perturbation of a quantum dynamical semigroup under the condition that the latter converges exponentially fast to a unique invariant state $\tau$: for two generators $\cL$ and $\cL+\cK$, if $\cL$ satisfies $\|e^{t\cL} - \tr(.)\,\tau\|_{1\rightarrow 1} \le c e^{- \omega t}$ for all $t\geq0$ and some $c, \omega > 0$, then
	\begin{equation*}
		\forall\rho,\sigma \text{ states},\quad  \norm{e^{t\cL}(\rho) - e^{t(\cL+\cK)}(\sigma)}_1 \le 
		\begin{cases}
			\norm{\rho - \sigma}_1 + t \norm{\cK}_{1\rightarrow 1}\,, &t < \hat{t}\\
			c e^{- \omega t} \norm{\rho - \sigma}_1 + \frac{\log(c) + 1 - c e^{-\omega t}}{\omega} \norm{ \cK}_{1\rightarrow 1}\,, &  t \ge \hat{t}
		\end{cases}
	\end{equation*}
	where $\hat{t} \coloneqq \frac{\log(c)}{\omega}$. The result can be easily extended to the case of bounded generators in infinite dimensions, although proving the exponential decay for the semigroup generated by $\cL$ is not easy. The situation becomes even trickier in the case of unbounded generators since the use of a Duhamel integral as in the proof in finite dimensions requires a proper justification. It is precisely these issues that we are interested in and want to address here.

	\subsection{Gaussian perturbations of the quantum Ornstein Uhlenbeck semigroup}

		The quantum Ornstein Uhlenbeck semigroup is well-known to correspond to a so-called beam-splitter channel of exponentially decreasing transmissivity $e^{-(\lambda^2-\mu^2)t}$ with unique Gaussian invariant state (see \cite{DePalma.2018}):
		\begin{equation*}
			\sigma\coloneqq \frac{\lambda^2-\mu^2}{\mu^2}\sum_{k=0}^\infty \left(\frac{\mu^2}{\lambda^2}\right)^k\,\ketbra{k}{k}\,.
		\end{equation*}
		While quantitative statements about the convergence of this semigroup towards $\sigma$ are known \cite{Cipriani.2000,Carbone.2007,Carlen.2017,DePalma.2018}, 
		they do not necessarily imply convergence in trace distance in contrast to their finite-dimensional analogues. In contrast, the semigroup is known to contract a certain kind of quantum Wasserstein distance, which we introduce now. First, given a bounded, self-adjoint operator $X\in\cB(\cH)$, we call $X$ a Lipschitz observable if $aX$, $a^\dagger X$ are bounded, and if $Xa$ and $Xa^\dagger$ are closable operators with bounded closures $\overline{Xa}$ and $\overline{Xa^\dagger}$. In this case, we denote by $\partial_a(X)\coloneqq aX-\overline{Xa}$ and $\partial_{a^\dagger}(X)=a^\dagger X-\overline{Xa^\dagger}$. The Lipschitz constant of $X$ is then defined as
		\begin{align*}
			\|X\|_{\operatorname{Lip}}\coloneqq \max\big\{\|\partial_a(X)\|_\infty,\,\|\partial_{a^\dagger}(X)\|_\infty\big\}\,.
		\end{align*}
		We denote the set of Lipschitz observables by $\operatorname{Lip}$. Next,  any $T\in\cT_{1,\operatorname{sa}}$, we denote
		\begin{align*}
			\|T\|_{W_1}\coloneqq \sup\,\left\{ \tr\big[X\,T\big]:\,X\in\operatorname{Lip},\,\|X\|_{\operatorname{Lip}}\le 1\right\}\,.
		\end{align*}
		In \cite[Proposition 6.4]{Gao.2021}, the authors showed that, for any $T\in\cT_{1,\operatorname{sa}}$ and $t>0$,
		\begin{align}\label{regularization}
			\|e^{t\cL_{\operatorname{qOU}}}(T)\|_1\le \sqrt{\frac{e^{-(\lambda^2-\mu^2)t}}{1-e^{-(\lambda^2-\mu^2)t}}}\,\Big(\|a\sigma-\sigma a\|_1+\|a^\dagger \sigma-\sigma a^\dagger \|_1\Big)\,\|T\|_{W_1} \,.
		\end{align}
		Moreover, using the canonical commutation relations, one can also prove the following identities (see e.g.~\cite{Carlen.2017}, or \cite[Proposition 6.2]{Gao.2021}): for any two states $\rho_1,\rho_2\in \cT_{1, \operatorname{sa}}$,  
		\begin{equation}\label{eq:qou-exponential-dampening}
			\|e^{t\cL_{\operatorname{qOU}}}(\rho_1-\rho_2)\|_{W_1}\le e^{-\frac{(\lambda^2-\mu^2)t}{2}}\,\|\rho_1-\rho_2\|_{W_1}\,.
		\end{equation}
		In the next proposition, we use these conditions to find a perturbation bound for any Gaussian perturbation of the quantum Ornstein Uhlenbeck semigroup. 
		
		\begin{prop}\label{propqOUperturb}
			Let $(\cL_{\operatorname{qOU}},\cT_f)$ be the generator of the quantum Ornstein Uhlenbeck semigroup with $\lambda>\mu\geq0$ and $(\varepsilon\cL_G,\cT_f)\coloneqq (\varepsilon\cL[{\gamma a+\eta\ad}],\cT_f)$ a Gaussian perturbation with $\gamma,\eta\in\mathbb{R}$, $\varepsilon>0$. Then, assuming $\lambda^2-\mu^2+|\gamma|^2-|\eta|^2> 0$, $\cL_{\operatorname{qOU}}+\varepsilon\cL_G$ generates a positivity and Sobolev preserving semigroup on $W^{k,1}$ for $k\geq1$, and there exist uniformly bounded functions $C(\varepsilon),D(\varepsilon)$ depending on $\lambda,\mu,|\eta|,|\gamma|$ such that, for all $t\ge 0$ and states $\rho\in W^{2,1}$
			\begin{equation}\label{eq-ex:perturbation-bound-qOU}
				\Big\|\left(e^{t \cL_{\operatorname{qOU}}}-e^{t(\cL_{\operatorname{qOU}}+\varepsilon\cL_G)}\right)(\rho)\Big\|_{1}\leq \varepsilon\, C(\varepsilon)\, \max\Big\{\norm{\rho}_{W^{2, 1}} ,D(\varepsilon)\Big\}\,.
			\end{equation}
		\end{prop}
		\begin{proof}
			The generation of a Sobolev preserving semigroup was already stated in Lemma \ref{lem-ex:qOU-differential-stability} for $\cL_{\operatorname{qOU}}$ and its proof can easily be extended to $\cL_{\operatorname{qOU}}+\varepsilon \cL_G$. For instance, given a state $\rho\in\cT_f$, one can show that
			\begin{equation*}
				\tr[\cL_{G}(\rho)(N+\1)]\leq -(|\gamma|^2-|\eta|^2)\,\tr[\rho N]+|\eta|^2\,,
			\end{equation*}    
			We have also seen in the proof of Lemma \ref{lem-ex:qOU-differential-stability} that $\tr[\rho \cL_{\operatorname{qOU}}]\le -(\lambda^2-\mu^2)\tr[\rho N]+\mu^2$, so that
			\begin{align*}
				\tr[(\cL_{\operatorname{qOU}}+\varepsilon \cL_G)(\rho)(N+\1)]\le -(\lambda^2-\mu^2+\varepsilon|\gamma|^2-\varepsilon|\eta|^2)\tr[\rho (N+\1)]+\lambda^2+\varepsilon|\gamma|^2\,.
			\end{align*} 
			Therefore, by Proposition \ref{prop-ex:uniformly-bounded-semigroup} we have that, as long as $\lambda^2-\mu^2+\varepsilon|\gamma|^2-\varepsilon|\eta|^2> 0$, for all states $\rho\in W^{2,1}$, $\rho \ge 0$ and all $t\ge 0$,
			\begin{align*}
				\|e^{t(\cL_{\operatorname{qOU}}+\varepsilon \cL_G)}(\rho)\|_{W^{2,1}}\le \max\left\{\norm{\rho}_{W^{2, 1}}, \frac{\lambda^2+\varepsilon|\gamma|^2}{\lambda^2-\mu^2+\varepsilon|\gamma|^2-\varepsilon|\eta|^2}\right\} \,.
			\end{align*}
			Next, for $\rho\in\cT_f$ and $0<u<t $, and denoting $\cL\equiv \cL_{\operatorname{qOU}}$ and $\widetilde{\cL}\equiv \cL_{\operatorname{qOU}}+\varepsilon \cL_G$,
			\begin{align*}
				\Big\|\Big(e^{t\cL}-e^{t\widetilde{\cL}}\Big)(\rho)\Big\|_1\le    \Big\|e^{u\cL}\left(e^{(t-u)\cL}-e^{(t-u)\widetilde{\cL}}\right)(\rho)\Big\|_1 +\Big\| \Big(e^{u\cL}-e^{u\widetilde{\cL}}\Big)e^{(t-u)\widetilde{\cL}}(\rho)\Big\|_1\equiv A+B\,.
			\end{align*}
			We use \Cref{regularization}, so that
			\begin{align*}
				A&\le c_u \Big\|\Big(e^{(t-u)\cL}-e^{(t-u)\widetilde{\cL}}\Big)(\rho)\Big\|_{W_1}\\
				&\le c_u\,\varepsilon\, \int_0^{t-u}\,\Big\|e^{s\cL}\cL_G\,e^{(t-u-s)\widetilde{\cL}}(\rho)\Big\|_{W_1}\,ds\\
				&=  c_u\,\varepsilon\, \int_0^{t-u}\,e^{-\frac{(\lambda^2-\mu^2)s}{2}}\Big\|\cL_G\,e^{(t-u-s)\widetilde{\cL}}(\rho)\Big\|_{W_1}\,ds\,,
			\end{align*}
			where $c_u\coloneqq \sqrt{\frac{e^{-(\lambda^2-\mu^2)u}}{1-e^{-(\lambda^2-\mu^2)u}}}\,\,\Big(\|\partial_a(\sigma)\|_1+\|\partial_{a^\dagger}(\sigma)\|_1\Big)$. Moreover, denoting $\widetilde{\rho}_v\coloneqq e^{v\widetilde{\cL}}(\rho)$ and $b=\gamma a+\eta a^\dagger$, since $\widetilde{\rho}_v\in W^{2,1}$ for all $v\ge 0$,
			\begin{align*}
				\Big\|\cL_G\,\widetilde{\rho}_v\Big\|_{W_1}\,&=\sup_{\|X\|_{\operatorname{Lip}}\le 1}\,\tr[X \cL_G \widetilde{\rho}_v]\\
				&=\frac{1}{2}\,\sup_{\|X\|_{\operatorname{Lip}}\le 1}\, \tr[\partial_{b^\dagger}(X)b \widetilde{\rho}_v-\partial_b(X)\,\widetilde{\rho}_vb^\dagger]\\
				&\le (|\eta|+|\gamma|)\, \Big(\|b\widetilde{\rho}_v\|_1+\|\widetilde{\rho}_vb^\dagger\|_1\Big)\\
				&\le (|\eta|+|\gamma|)\, \Big(\|b\,(N+\1)^{-\frac{1}{2}}\|+\|(N+\1)^{-\frac{1}{2}}\,b^\dagger\|\Big)\,\|\widetilde{\rho}_v\|_{W^{2,1}}\\
				&\le (|\eta|+|\gamma|) \Big(\|b\,(N+\1)^{-\frac{1}{2}}\|+\|(N+\1)^{-\frac{1}{2}}\,b^\dagger\|\Big)\\
				&\qquad \qquad \qquad \cdot \max\left\{\norm{\rho}_{W^{2, 1}}, 
				\frac{\lambda^2+\varepsilon|\gamma|^2}{\lambda^2-\mu^2+\varepsilon|\gamma|^2-\varepsilon|\eta|^2} \right\}  \, .
			\end{align*}
			On the other hand, 
			\begin{align*}
				B&\le \varepsilon\,\int_0^u\,\|e^{s\cL} \cL_Ge^{(t-s)\widetilde{\cL}}(\rho)\|_1\,ds\\
				&\le\,u\varepsilon\,\max_{s\in[0,u]}\, \|\cL_Ge^{(t-s)\widetilde{\cL}}(\rho)\|_1 \\
				&\le u\varepsilon \|\cL_G \circ \mathcal{W}^{-2}\|_{\cT_1\to \cT_1}\,\max_{s\in[0,u]}\,\|e^{(t-s)\widetilde{\cL}}(\rho)\|_{W^{2,1}}\\
				&\le u\varepsilon \|\cL_G \circ \mathcal{W}^{-2}\|_{\cT_1\to \cT_1}\, \max\left\{ \|\rho\|_{W^{2,1}}, 
				\frac{\lambda^2+\varepsilon|\gamma|^2}{\lambda^2-\mu^2+\varepsilon|\gamma|^2-\varepsilon|\eta|^2} \right\}\,.
			\end{align*}
			The result follows from a simple bound on $\|\cL_G\circ \cW^{-2}\|_{\cT_1\to \cT_1}$ with the help of H\"{o}lder's inequality and a standard density argument, and by choosing $u$ appropriately.
		\end{proof}

	\subsection{Photon-dissipation and CAT qubits}\label{subsec:cat-perturbation}

		As mentioned before, one crucial property of the underlying evolution in continuous error correction is that it is exponentially converging to the code-space. In the spirit of \cite{Szehr.2013}, we prove a large time perturbation result for the $l$-photon dissipation perturbed by a Hamiltonian evolution. This can be generalized to dissipative perturbations. First, we recall the exponential convergence of the $l$-photon dissipation defined by $L=a^l-\alpha^l$ (\cite[Theorem 2]{Azouit.2016}):
		\begin{equation}\label{eq:exponential-convergence}
			\tr[Le^{t\cL_l}(\rho) L^\dagger]\leq e^{-l!t}\tr[L\rho L^\dagger]\,.
		\end{equation}
		Additionally, it is shown that there is a unique limit $\overline{\rho}$ of $e^{t\cL_l}(\rho)$ for $t\rightarrow\infty$. We show large-time perturbation bounds by combining this bound with our established generation theory for Sobolev and positivity-preserving Markov semigroups. We start with the $l$-photon dissipation perturbed by the Hamiltonian introduced in Lemma \ref{lem:l-diss}, i.e.~$H=p_H(a,\ad)$ with $d_H\leq2(l-1)$.
		
		\begin{thm}\label{thm:l-diss-hamiltonian-perturbation}
			Let $\cL_l$ be the generator of the $l$-photon dissipation and $p_H\in\C[X,Y]$ with $\deg(p_H)=d_H\leq2(l-1)$ such that $H=p_H(a,\ad)$ is a symmetric operator. Then, there exist explicit constants $c,\gamma>0$ such that for $\varepsilon \ge 0$ and all states $\rho\in W^{2(l + d_H + 2),1}$
			\begin{equation*}
				\begin{aligned}
					\Big|\tr[L\left(e^{t\cL_l}(\rho)-e^{t(\cL_l+\varepsilon\cH[H])}(\rho)\right)L^\dagger]\Big|\leq\varepsilon c\left(1-e^{-l!t}\right)\max\{\gamma,\|\rho\|_{W^{2(l + d_H + 2),1}}\}\,.
				\end{aligned}
			\end{equation*}
		\end{thm}
		\begin{proof}
			The proof consists in applying Lemma \ref{lem:l-diss-hamiltonian} in combination with \Cref{eq:exponential-convergence}. Let $\rho \in \cT_f$ and $\cW^k = (N + \1)^{k/4} \cdot (N + \1)^{k/4}$, then
			\begin{equation*}
				\begin{aligned}
					&\tr[L\biggl(e^{t\cL_l}(\rho)-e^{t(\cL_l+\varepsilon\cH[H])}(\rho)\biggr)L^\dagger]\\
					&\qquad= \varepsilon \int_{0}^t \tr[L e^{s\cL_l}\cW^{-2(l + 2)}\cW^{2(l + 2)}\cH{[H]}e^{(t-s)(\cL_l+\varepsilon\cH[H])}(\rho)L^\dagger] ds\\
					&\qquad\overset{(1)}{\leq}\varepsilon \int_{0}^t\tr[L e^{s\cL_l}\cW^{-2(l + 2)}(\1)L^\dagger] ds \, 2\Lambda d_H^{l + 2} \sqrt{d_H!}\max\left\{\gamma_\varepsilon,\,\|\rho\|_{W^{2(l + d_H + 2),1}}\right\}\\
					&\qquad \overset{(2)}{\leq}\varepsilon \frac{\pi^2}{3}\Lambda d_H^2\sqrt{d_H!}(1+|\alpha|^l\,(l+1)\,\sqrt{l!}+|\alpha|^{2l})\frac{1}{l!}(1-e^{-l!t})\max\left\{\gamma_\varepsilon,\,\|\rho\|_{W^{2(l + d_H + 2),1}}\right\}\\
					&\qquad\eqqcolon\varepsilon c(1-e^{-l!t})\max\left\{\gamma_\varepsilon,\,\|\rho\|_{W^{2(l + d_H + 2),1}}\right\}
				\end{aligned}
			\end{equation*}
			where $\Lambda$ denotes the larges coefficient of $\cH[H]$ in absolute value. Note that the Bochner integral in the calculation is well-defined by the boundedness of the integrand w.r.t.~$s$ and the same argumentation as Theorem \ref{thm:semigroup-perturbation}. Besides the boundedness above, the Sobolev preserving property of the involved semigroups implies by construction that the integral is Sobolev preserving. Therefore, the integral commutes with the map $x\mapsto\tr[LxL^\dagger]$ for $x\in W^{2(l+d_H+2),1}$. In $(1)$ we used that $L e^{s\cL_l}(\cdot) L^\dagger$ preserves positivity and
			\begin{align*}
				\cW^{2(l + 1)} \cH[H] e^{t - s(\cL_l + \varepsilon \cH[H]}(\rho) &\le \norm{\cW^{2(l + 2)} \cH[H] e^{(t - s)(\cL_l + \varepsilon \cH[H]}(\rho)}_\infty \1\\
				&\le \norm{\cW^{2(l + 2)} \cH[H] e^{(t - s)(\cL_l + \varepsilon \cH[H]}(\rho)}_1 \1\\
				&\le 2 \Lambda d_H^{l + 2} \sqrt{d_H!} \norm{e^{(t - s)(\cL_l + \varepsilon \cH[H]}(\rho)}_{W^{2(l + d_H + 2)}} \1\\
				&\le 2 \Lambda d_H^{l + 2} \sqrt{d_H!} \max\{\gamma_\varepsilon, \norm{\rho}_{W^{2(l + d_H + 2)}}\} \1 \, . 
			\end{align*}
			In the above estimation, we used $(N + \1)^{-d_H}$ to control $\cH[H]$. We further used the Sobolev preserving property of the semigroup from Lemma \ref{lem:l-diss-hamiltonian}. Finally, applying the decay (q.v.~\Cref{eq:exponential-convergence}) to $\cW^{-2(l + 1)}(\1) = (N + \1)^{-(l + 2)}$ we estimated in $(2)$:
			\begin{align*}
				\int_{0}^t\tr[L e^{s\cL_l}\cW^{-2(l + 2)}(\1)L^\dagger] ds &\le \int\limits_{0}^t e^{-l! t} \tr[L(N + \1)^{-(l + 2)}L^\dagger]\\
				&\le (1 - e^{-l!t}) (1+|\alpha|^l\,(l+1)\,\sqrt{l!}+|\alpha|^{2l}) \frac{\pi^2}{6}
			\end{align*}
			with the bound 
			\begin{equation*}
				\|L(N+1)^{-l - 2}L^\dagger\|_1\leq (1+|\alpha|^l\,(l+1)\,\sqrt{l!}+|\alpha|^{2l}) \frac{\pi^2}{6}
			\end{equation*}
			that follows from
			\begin{equation}
				\begin{aligned}
					L^\dagger L&=(\ad)^l a^l-\overline{\alpha}^la^l-\alpha^l(\ad)^l+|\alpha|^{2l}\\
					&\leq(N-(l-1)\1)\cdots(N-\1)N+|\alpha|^l\,(l+1)\,\sqrt{(N+l\1)\cdots (N+\1)}+|\alpha|^{2l}\\
					&\leq (N+\1)^l+|\alpha|^l\,(l+1)\,\sqrt{l!}\,(N+\1)^{l/2}+|\alpha|^{2l}\,.
				\end{aligned}
			\end{equation}
			and the $\norm{(N + \1)^{-2}}_1 = \frac{\pi^2}{6}$. This concludes that claim.
		\end{proof}
		
		Special cases of the above result include the $X$ or $Z(\theta)$ gate.
		\begin{cor}[$Z(\theta)$-gate]
			Let $\cL_2$ be the $2$-photon dissipation. Then, for all $\varepsilon\in[0,1]$ and all states $\rho\in W^{10,1}$
			\begin{equation*}
				\begin{aligned}
					&\Big|\tr[L\left(e^{t\cL_2}(\rho)-e^{t(\cL_2+\varepsilon\cH[a+\ad])}(\rho)\right)L^\dagger]\Big|\\
                    &\qquad\qquad\qquad\leq \varepsilon 2(1+6|\alpha|^2+|\alpha|^{4})(1-e^{-2t})\max\left\{\gamma_\varepsilon,\|\rho\|_{W^{10,1}}\right\}
				\end{aligned}
			\end{equation*}
			where $\gamma_\varepsilon=\frac{1}{25}\left(6+\sqrt{2}\,2^6\,5\,|\alpha|^2+\varepsilon4^5\sqrt{24}\right)^6$\,.
		\end{cor}

	\subsection{Application: entropic and capacity continuity bounds}

		Here, we provide one basic application to the perturbation bounds found in this section. We recall the definition of the energy-constrained diamond norm:
		
		\begin{defi}[see \cite{Shirokov.2018,Winter.2017}]
			Given $E\ge 0$ and any two completely positive, trace-preserving maps $\cN,\cM:\cT_1(\cH_m)\to\cT_1(\cH_m)$, their energy constrained diamond norm distance is defined as
			\begin{align}\label{ECnorm}
				\|\cN-\cM\|_\diamond^E\coloneqq \sup_{\cH_r}\,\sup_{\substack{\rho\in \cD(\cH_m\otimes \cH_r)\\\tr\big[\rho_{m}N_m\big]\le E}}\,\big\|(\cN-\cM)\otimes\id_R(\rho)\big\|_1\,,
			\end{align}
			where $N_m\coloneqq \sum_{i=1}^m\,a_i^\dagger a_i$ denotes the total photon number operator, and where the supremum is over all bipartite states $\rho_{mr}$ on $\cH_\otimes \cH_r$ with reduced state $\rho_m$ on $\cH_m$ of average total photon number at most $E$, for some arbitrary separable Hilbert space $\cH_r$. 
		\end{defi}
		
		Not that we denote the set of quantum states over a separable Hilbert space as $\cD(\cH) := \{\rho \in \cT_{1, \operatorname{sa}}(\cH) \;:\; \rho \ge 0, \tr[\rho] = 1\}$ in this section.
		In other words, the energy-constrained diamond norm is a measure of distinguishability between quantum channels with entanglement assistance, and where the input states used for this task are restricted to a physically relevant set of energy-limited states. By the same reasoning as for the usual diamond norm, the supremum in \eqref{ECnorm} can be restricted to $\cH_r\cong \cH_m$, and the optimization can be restricted to pure states on $\cH_m\otimes \cH_r$. Moreover, it turns out that the above definition is equivalent to one where the input state is energy limited on both $m$ and $r$, as introduced in \cite{Pirandola.2017}:
		\begin{align*}
			\vertiii{\cN-\cM}_\diamond^E\coloneqq \sup_{\substack{\rho\in \cD(\cH_m\otimes \cH_{m'})\\ \tr\big[\rho(N_m+N_{m'})\big]\le E}}\,\big\|(\cN-\cM)\otimes\id_{m'}(\rho)\big\|_1\,,
		\end{align*}
		for $\cH_{m'}\cong \cH_m$, and where the difference with \eqref{ECnorm} lies in the input states in the above optimization are energy constrained in both their inputs. Clearly, from the definitions we have that 
		\begin{align}\label{tripledoubleECequiv}
			\vertiii{\cN-\cM}_\diamond^E\le \|\cN-\cM\|_\diamond^E\le \vertiii{\cN-\cM}_\diamond^{2E}\,.
		\end{align}
		The second bound above simply results from taking a state $\rho_{mm'}$ for which $\tr\big[\rho_mN_m\big]\le E$ and unitarily rotating the second register so that $\rho_m'=\rho_m$, and therefore $\tr\big[\rho_{m'}N_{m'}\big]\le E$, too.
		
		Similarly, in the next lemma, we establish a straightforward connection between the energy-constrained diamond norm distance between two channels and certain norms between Sobolev spaces:
		\begin{lem}\label{lemmaequivECW}
			For any two completely positive, trace preserving maps $\cN,\cM:\cT_1(\cH_1)\to \cT_1(\cH_1)$ and $E\ge 0$,
			\begin{align*}
				\|\cN-\cM\|_{\diamond}^E= (1+E)\,\sup_{\rho\in\cD}\,\frac{\|(\cN-\cM)\otimes \operatorname{id}_{M'}(\rho)\|_1}{\|\rho\|_{W^{(2,0),1}}}\,,
			\end{align*}
			where $\cD\coloneqq \cD(\cH_m\otimes \cH_{m'})$ above. Similar identities hold in multi-mode settings.
		\end{lem}
		\begin{proof}
			This is direct from
			\begin{align*}
				\|{\cN-\cM}\|_\diamond^{E}&=\sup_{\substack{\rho \in \cD\\\tr[\rho N]\le E}}\,\|(\cN-\cM)\otimes\id(\rho)\|_1\\
				&= \sup_{\substack{\rho \in \cD\\\|\rho\|_{W^{(2,0),1}}\le 1+E}}\,\|(\cN-\cM)\otimes\id(\rho)\|_1\,.
			\end{align*}
		\end{proof}
		
		Lemma \ref{lemmaequivECW} can be used in combination with perturbation bounds and entropic continuity bounds like those derived e.g.~in \cite{Winter.2016,Winter.2017,Shirokov.2017,Shirokov.2018} to control the deviation of energy-constrained channel capacities in presence of a Lindbladian perturbation. Since these considerations go beyond the scope of the present paper, we do not pursue them further. Instead, we provide a basic illustration of the method we propose in the case of Gaussian perturbations of Gaussian semigroups by combining Proposition \ref{propqOUperturb} with Lemma \ref{lemmaequivECW}.
		
		\begin{cor}\label{corECDN}
			Let $(\cL_{\operatorname{qOU}},\cT_f)$ be the generator of the quantum Ornstein Uhlenbeck semigroup with $\lambda>\mu\geq0$ and $(\varepsilon\cL_G,\cT_f)\coloneqq (\varepsilon\cL[{\gamma a+\eta\ad}],\cT_f)$ a Gaussian perturbation with $\gamma,\eta\in\mathbb{R}$, $\varepsilon>0$. Then, assuming $\lambda^2-\mu^2+|\gamma|^2-|\eta|^2> 0$ as in Proposition \ref{propqOUperturb}, there exist uniformly bounded functions $C(\varepsilon)$, $D(\varepsilon)$ such that, for all $t\ge 0$:
			\begin{equation}
				\Big\|e^{t \cL_{\operatorname{qOU}}}-e^{t(\cL_{\operatorname{qOU}}+\varepsilon\cL_G)}\Big\|_{\diamond}^E\leq\,(1+E) \varepsilon\, C(\varepsilon)\, \max\Big\{1,D(\varepsilon)\Big\}\,.
			\end{equation}
		\end{cor}

\section{Discussion and open questions}\label{sec:discussion}
	In this paper, we proved a generation theorem for bosonic Sobolev preserving quantum dynamical semigroups, which we also extended to time-dependent evolutions on bosonic multi-mode systems. The property of Sobolev preservation allowed all-time perturbation bounds, which we later applied to bosonic Gaussian semigroups, the quantum Ornstein Uhlenbeck semigroup, and dissipative CAT-qubit dynamics. As a consequence, we also established convergence results for the adherent points in the limit of large times. For this, we only required two physical assumptions on the generator (i) that it is in GKSL form with a Hamiltonian and jump-operators defined by polynomials in the annihilation and creation operators; and (ii) a bound on the output moments of the generator in terms of its input-moments. An interesting question for further work would be to further generalize the result by allowing different Sobolev spaces on which the semigroups are defined, and by this enlarge the class of generators covered by the generation theorem. In this line of research, a more detailed consideration of the Toffoli gate would be interesting.

\medskip

\section{Declarations}
	\emph{Competing interests:}
    
    T.M.~and C.R.~received funding from the Munich Center for Quantum Sciences and Technology, C.R.~from the Humboldt Foundation, and all authors from the Deutsche Forschungsgemeinschaft (DFG, German Research Foundation) – Project-ID 470903074 – TRR 352. The authors have no further relevant financial or non-financial interests to disclose.\\[2ex]
	\emph{Data availability statement:}
    
    No datasets were generated or analyzed during the current study.

% References
\sloppy
\setlength{\bibitemsep}{0.5ex}
\printbibliography[heading=bibnumbered]
\vspace{2ex}

\newpage

% Appendix
\appendix
\addtocontents{toc}{\protect\setcounter{tocdepth}{1}}
\addcontentsline{toc}{section}{Appendices} 
\addtocontents{toc}{\protect\setcounter{tocdepth}{0}}

\section{Semigroup perturbation theory}\label{appy:semigroup-perturbation-theory}
Here, we prove Theorem \ref{thm:semigroup-perturbation} stated in the preliminary section. 
\begin{thm}\label{thm-appx:semigroup-perturbation}
	Let $(\cL,\cD(\cL))$ and $(\cL+ \cK,\cD(\cL+\cK))$ be two generators of $C_0$-semigroups on $\cX$, for an operator $(\cK,\cD(\cK))$. Moreover, let $(\cW,\cD(\cW))$ be an invertible operator on $\cX$ with bounded inverse, such that $\cD(\cW)$ is an $\cL+\cK$-admissible subspace {(see Definition \ref{defi:admissible-spaces})} and such that $\cK\cW^{-1}$ is bounded. Then, for all $x\in\cD(\cW)$,
	\begin{equation*}
		\|(e^{t\cL}-e^{t(\cL+\cK)})x\|_{\cX}\leq t \|\cK\cW^{-1}\|_{\cX\to \cX}\; \int_{0}^1\| e^{(1-s)t\cL}\|_{\cX\to \cX}\;\|e^{st(\cL-\cK)}\|_{\cW\to \cW}\;ds\|x\|_{\cW}.
	\end{equation*}
	and especially
	\begin{equation*}
		(e^{t\cL}-e^{t(\cL+\cK)})x= t \int_{0}^1 e^{(1-s)t\cL} \cK e^{st(\cL-\cK)}x\;ds.
	\end{equation*}
\end{thm}
\begin{proof}
	For simplicity, we assume that $t=1$. To be able to use the integral form, we start by considering the following vector-valued map: 
	\begin{equation*}
		[0,1]\ni s\mapsto e^{(1-s)\cL}\cK e^{s(\cL+\varepsilon\cK)}x.
	\end{equation*}
	It is clear that $s\mapsto  e^{(1-s)\cL}$ is a strongly continuous map of bounded operators. Moreover,
	\begin{align*}
		\|\cK e^{s(\cL+\cK)}x-\cK e^{(s+s')(\cL+\cK)}x\|_{\cX}&\leq\|\cK\cW^{-1}\|_{\cX}\, \|(e^{s(\cL+\cK)}x-e^{(s+s')(\cL+\cK)}x)\|_{\cW}
	\end{align*}
	shows that $s\mapsto\cK e^{s(\cL+\cK)}$ is strongly continuous because $s\mapsto e^{s(\cL+\cK)}$ defines a $C_0$-semigroup on $(\cD(\cW),\|\cdot\|_\cW)$. Then, for every converging sequence $(s_n)_{n\in\N}\rightarrow s$ for $n\rightarrow\infty$ the set
	\begin{equation*}
		\left\{\cK e^{s_n(\cL+\cK)}x\,|\,n\in\N\right\}
	\end{equation*}
	is relatively compact in $\cX$. Therefore, the strong continuity of $s\mapsto  e^{(1-s)\cL}$ is equivalent to uniform continuity by \cite[Prop.~A3]{Engel.2000}, the considered vector-valued map is continuous and we can use the fundamental theorem of calculus for the generalized Riemann integral so that
	\begin{equation*}
		\| (e^{\cL}-e^{\cL+\cK})x\|_{\cX}\leq\varepsilon\int_0^1\|  e^{(1-s)\cL}\cK e^{s(\cL+\varepsilon\cK)}x\|_{\cX}ds
	\end{equation*}
	proves the theorem.
\end{proof}

\section{Bosonic single mode system}\label{sec-appx:annihilation-creation}
In this section, we prove some basic properties of polynomials of annihilation and creation operators. We shortly repeat the normal form of our polynomials in $a$ and $\ad$ given in \Cref{eq:ccr-polynomial-representation}: 
\begin{equation*}
	p(a\,,\ad)=\sum_{i + 2j \le \deg(p)}\lambda_{ij}(\ad)^iN^j+\,\sum_{k + 2l\le \deg(p)}\mu_{kl}N^la^k
\end{equation*}
with coefficients $\lambda_{ij}, \mu_{kl}\in\C$. The modification considered in our proof (see Lemma \ref{lem:semigroup-of-G-for-positivity}) is given in \Cref{eq:polynomial+number-op} by
\begin{equation*}
	\Tilde{p}(a,\ad)\coloneqq N^{2d}+p(a,\ad)\,.
\end{equation*}
Then, we start proving a simple representation of a domain of $p(a,\ad)$ and $\Tilde{p}(a,\ad)$ which extends the domain $\cH_f$: For $n\in\N$
\begin{equation}\label{eq:domain-poly}
	\begin{aligned}
		p(a\,,\ad)\ket{n}&=\sum_{i+2j\leq d}\lambda_{ij}n^{j}\sqrt{i!\binom{n+i}{i}}\ket{n+i}+\sum_{k+2l\leq d}\mu_{kl}(n-k)^l\sqrt{k!\binom{n}{k}}\ket{n-k}
	\end{aligned}
\end{equation}
where $d\coloneqq\deg(p)$. This directly implies for $|\phi\rangle=\sum_n\phi_n|n\rangle\in\cD(N^{d/2})$
\begin{equation*}
	\begin{aligned}
		p(a\,,&\ad)\ket{\phi}\\
        &=\sum_{n=0}^\infty\sum_{i+2j\leq d}\phi_n\lambda_{ij}n^{j}\sqrt{i!\binom{n+i}{i}}\ket{n+i}+\sum_{n=0}^\infty\sum_{k+2l\leq d}\phi_n\mu_{kl}(n-k)^l\sqrt{k!\binom{n}{k}}\ket{n-k}\\
		&=\sum_{n=0}^\infty\sum_{i+2j\leq d}\left(\phi_{n-i}\lambda_{ij}(n-i)^{j}\sqrt{i!\binom{n}{i}}+\phi_{n+i}\mu_{ij}n^j\sqrt{i!\binom{n+i}{i}}\right)\ket{n}\,.\\
	\end{aligned}
\end{equation*}
Then, the leading order of the summands in $n$ is maximal of order $d/2$ so that $\cD(N^{d/2})$ is a domain of $p(a,\ad)$. For the modified polynomial $\Tilde{p}$, there is sequence of functions $R_n:\cH \rightarrow\R$ with asymptotics strictly smaller than $n^{4d}$ such that
\begin{equation}\label{eq:leading-order-poly}
	\|\tilde{p}(a,\ad)\ket{\phi}\|^{2}=\sum_{n=0}^\infty |\phi_n|^2n^{4d}+R(\phi)\,.
\end{equation}
Having the domain above in mind, we can prove that $p$ is closable and $\tilde{p}$ is a closed operator with core $\cH_f$:

\begin{lem}[Adjoint and core of polynomials of $a, \ad$]\label{lem-appx:formal-polynomial-ccr-adjoint-core}
	Let $p\in\C[X,Y]$ be a polynomial on $\C$ and $(p(a,\ad),\cD(N^{d/2}))$ the unbounded operator in normal form \eqref{eq:domain-ccr-polynomial}. Then, $p(a,\ad)$ is closable and there is a $c\geq0$ such that for all $\ket{\phi}\in\cD(N^{d/2})$
	\begin{equation*}
		\|p(a,\ad)\ket{\phi}\|\leq c\|(\1+N)^{d/2}\ket{\phi}\|\,.
	\end{equation*}
	The modification $\tilde{p}(a,\ad)=N^{2d}+p(a,\ad)$ is closed with $\cD(\Tilde{p}(a,\ad))=\cD(\Tilde{p}(a,\ad)^\dagger)=\cD(N^{2d})$ and core $\cH_f$.
\end{lem}
\begin{proof}
	First, note that the relative boundedness w.r.t.~the number operator is a direct consequence of \Cref{eq:domain-poly}.
	To prove that $p(a,\ad)$ is closable, we show that $\cD(N^{d/2})\subset\cD(p(a,\ad)^\dagger)$: we recall that the adjoint is defined via boundedness of the functional 
	\begin{equation*}
		\cD(p(a,\ad))\ni\ket{\phi}\mapsto\braket{p(a,\ad)\phi\,,\varphi}
	\end{equation*}
	for $\phi\in\cD(p(a,\ad)^\dagger)$. Since for all $n,m\in\N$
	\begin{equation*}
		\braket{a\,n\,,m}=\delta_{n,0}\delta_{n-1,m}\sqrt{n}=\delta_{n,m+1}\sqrt{m+1}=\braket{n\,,\ad\,m}
	\end{equation*}
	and by the definition of the domains, we know that for all $\ket{\phi},\ket{\psi}\in\cD(N^{k+\frac{l}{2}})$ 
	\begin{equation*}
		\braket{N^k(\ad)^l \phi\,,\psi}=\braket{\phi\,,a^lN^k\psi}\quad\text{and}\quad\braket{a^lN^k\phi\,,\psi}=\braket{\phi\,,N^k(\ad)^l\psi}\,.
	\end{equation*}
	By sesquilinearity of the scalar product and \Cref{eq:domain-ccr-polynomial}, the above equations hold for all polynomials $p\in\C[X,Y]$ which proves $\cD(N^{d/2})\subset\cD(p(a,\ad)^\dagger)$. Since $\cD(N^{d/2})$ is a dense subspace of $\cH$, Theorem 7.1.1 in \cite{Simon.2015} shows that $p(a,\ad)$ is closable. Next, we show that the modified polynomial 
	\begin{equation*}
		(N+\1)^{2d}+p(a,\ad)
	\end{equation*}
	is already closed. Actually, we show $\cD(N^{2d})=\cD(\Tilde{p}(a,\ad)^\dagger)$, i.e.~the domain is maximal, by contradiction. Assume that there exists a $\varphi\in\cD(\Tilde{p}(a,\ad)^\dagger)\backslash\cD(\Tilde{p}(a,\ad))$ and define $P_M$ to be the projection on the first $M$ Fock basis elements. Then, we use the representation in \Cref{eq:ccr-polynomial-representation} so that, denoting by $\Tilde{p}^\dagger$ the polynomial where we took the complex conjugate of the coefficients and swapped the coordinates,
	\begin{equation*}
		P_M\Tilde{p}^\dagger(a,\ad)=P_M(N+\1)^{2d}P_M+P_{M}\left(\sum_{i + 2j \le d}\overline{\lambda}_{ij}N^ja^iP_{M-i}+\,\overline{\mu}_{ij}(\ad)^iN^jP_{M+i}\right)P_{M+d}\,.
	\end{equation*}
	We then can define the state
	\begin{equation*}
		\ket{\phi_M}\coloneqq\frac{P_M\Tilde{p}^\dagger(a,\ad)\ket{\varphi}}{\|P_M\Tilde{p}^\dagger(a,\ad)\ket{\varphi}\|}\in\cH_f
	\end{equation*}
	Next, we use $\phi_M$ to get a lower bound on the operator norm of 
	\begin{equation*}
		\cD(p(a,\ad))\ni\ket{\phi}\mapsto\braket{\Tilde{p}(a,\ad)\phi\,,\varphi}
	\end{equation*}
	by
	\begin{equation*}
		\begin{aligned}
			\sup_{\|\phi\|=1}|\braket{\Tilde{p}(a,\ad)\phi\,,\varphi}|&\geq|\braket{\Tilde{p}(a,\ad)\phi_M\,,\varphi}|=|\braket{\phi_M\,,P_M\Tilde{p}^\dagger(a,\ad)\varphi}|=\|P_M\Tilde{p}^\dagger(a,\ad)\varphi\|.
		\end{aligned}
	\end{equation*}
	Now, by definition of $\Tilde{p}$ and denoting by $\varphi_n$ the coefficients of $|\varphi\rangle$ in the Fock basis,
	\begin{equation*}
		\|P_M\Tilde{p}^\dagger(a,\ad)\varphi\|^2 =\Big\|\sum_{n=0}^{M+d}\varphi_nP_M\left((N+\1)^{2d}+p^\dagger(a,\ad)\right)\ket{n}\Big\|^2\,,
	\end{equation*}
	where $p^\dagger$ is defined similarly to $\tilde{p}^\dagger$. By assumption $\phi\notin\cD(N^{2d})$, so that the above sequence is diverging for $M\rightarrow\infty$ to infinity (see \Cref{eq:leading-order-poly}). This contradicts the assumption and shows $\cD(\tilde{p}(a,\ad)^\dagger)=\cD(N^{2d})$ as well as $\tilde{p}(a,\ad)^\dagger=\tilde{p}^\dagger(a,\ad)$. Moreover, $p(a,\ad)$ is by Theorem 7.1.1 in \cite{Simon.2015} a closed operator. Since $\{\ket{n}\}_{n \in \N}$ is an orthonormal basis and $N$ a multiplication operator on that basis, we can immediately conclude that $\cH_f$ is a core for $N$ and further for all $(\1 + N)^k$, $k \ge 0$. Since $\Tilde{p}(a,\ad)$ is closed w.r.t.~$\cD(N^{p/2+1})$, $\cH_f$ is also a core of $\Tilde{p}(a,\ad)$.
\end{proof}

Having in mind that polynomials of the annihilation and creation are closed operators on certain domains, we use the canonical commutation relation $[a,\ad]=\1$ in the following lemma:
\begin{lem}\label{lem:l-ccr}
	Let $l\in\N$, then the following hold on $\cH_f$ and can be extended to maximal domains by taking limits
	\begin{equation}
		\begin{aligned}
			(\ad)^la^l&=(N-(l-1)\1)(N-(l-2)\1)\cdots(N-\1)N\,,\\
			a^l(\ad)^l&=(N+\1)(N+2\1)\cdots(N+(l-1)\1)(N+l\1)\,.
		\end{aligned}
	\end{equation}
\end{lem}
\begin{proof}
	The above equalities can be proven by induction over $l\in\N$. The cases $l\in\{0,1\}$ are trivial by definition. Next assume that the equation holds for $l\in\N$, then by \Cref{eq:symmetry-function}
	\begin{equation*}
		\begin{aligned}
			(\ad)^{l+1}a^{l+1}&=(\ad)^{l}Na^{l}=(N-l)(\ad)^{l}a^{l}=(N-l)\cdots N
		\end{aligned}
	\end{equation*}
	which finishes the proof by induction. The second expression can be proven by induction as well, and the induction start is again clear by the CCR. Next, we assume the equation for $l\in\N$. Then, \Cref{eq:symmetry-function} shows
	\begin{equation*}
		a^{l+1}(\ad)^{l+1}=a^{l}(N+\1)(\ad)^{l}=a^{l}(\ad)^{l}(N+(l+1)\1)=(N+\1)\cdots(N+(l+1)\1)
	\end{equation*}
	which completes the induction.
\end{proof}

\begin{lem}\label{lem:two-point-hamiltonian-bound}
	Let $\ell_1,\ell_2,k_1,k_2\in\N$ with $\min\{\ell_1,k_1\}=\min\{\ell_2,k_2\}=0$, $z\in\C$, and $h:\mathbb{N}^2\rightarrow\mathbb{R}$ a positive function that is increasing in each of its variables. Then,
	\begin{equation*}
		\begin{aligned}
			&za_1^{\ell_1}a_2^{\ell_2}h(N_1,N_2)(\ad_1)^{k_1}(\ad_2)^{k_2}+\overline{z}a_1^{k_1}a_2^{k_2}h(N_1,N_2)(\ad_1)^{\ell_1}(\ad_2)^{\ell_2}\\
			&\qquad\qquad\qquad\qquad \leq 2|z|\widetilde{h}_{m_1,m_2}(N_1+m_1I,N_2+m_2I)\,,
		\end{aligned}
	\end{equation*}
	where $m_1\coloneqq\max\{\ell_{1},k_{1}\}$, $m_2\coloneqq \max\{\ell_2,k_2\}$ and $$\widetilde{h}_{m_1,m_2}(n_1,n_2)={\prod_{j=n_1-m_1+1}^{n_1} \sqrt{j}\prod_{i=n_2-m_2+1}^{n_2}\sqrt{i}} \,\,\,h(n_1,n_2)\,1_{n_1\ge m_1}1_{n_2\ge m_2}\,,$$
	where we introduced the notation $1_{x\ge m}$ for the indicator function on the set $\{x:\,x\ge m\}$, and where by convention we take $\prod_{i=a}^b=1$ when $a>b$.
\end{lem}
\begin{proof}
	We define $K\coloneqq za_1^{\ell_1}a_2^{\ell_2}h(N_1,N_2)(\ad_1)^{k_1}(\ad_2)^{k_2}+\overline{z}a_1^{k_1}a_2^{k_2}h(N_1,N_2)(\ad_1)^{\ell_1}(\ad_2)^{\ell_2}$ and represent it in the $2$-mode Fock basis:
	\begin{equation*}
		\begin{aligned}
			K&=\sum_{n_1,n_2}\,h(n_1,n_2)\,  \big( z a_1^{\ell_1}a_2^{\ell_2}\,\ketbra{n_1,n_2}{n_1,n_2}(a_1^\dagger)^{k_1}(a_2^\dagger)^{k_2}+\overline{z} a_1^{k_1}a_2^{k_2}\ketbra{n_1,n_2}{n_1,n_2}(a_1^\dagger)^{\ell_1}(a_2^\dagger)^{\ell_2} \big)\\
			&=\sum_{\substack{n_1\geq m_1\\n_2\geq m_2}} g_{\substack{\ell_1,\ell_2\\k_1,k_2}}(n_1,n_2)\\
			&\qquad \qquad \left(z\ketbra{n_1-\ell_1,n_2-\ell_2}{n_1-k_1,n_2-k_2}+\overline{z}\ketbra{n_1-k_1,n_2-k_2}{n_1-\ell_1,n_2-\ell_2}\right)\,, 
		\end{aligned}
	\end{equation*}
	where 
	\begin{equation*}
		\begin{aligned}
			&g_{\substack{\ell_1,\ell_2\\k_1,k_2}}(n_1,n_2)\\
			&\coloneqq h(n_1,n_2) \sqrt{n_1\dots(n_1-\ell_1+1)n_1\dots (n_1-k_1+1)n_2\dots (n_2-\ell_2+1)n_2\dots (n_2-k_2+1)}\,. 
		\end{aligned}
	\end{equation*}
	By assumption, since $\min\{\ell_1,k_1\}=\min\{\ell_2,k_2\}=0$, we have that $g_{\substack{\ell_1,\ell_2\\k_1,k_2}}(n_1,n_2)=\widetilde{h}_{m_1,m_2}(n_1,n_2)$ for $n_1\ge m_1$, $n_2\ge m_2$, and 
	
	\begin{align*}
		&K=\sum_{n_1,n_2\in \mathbb{N}}\widetilde{h}_{m_1,m_2}(n_1+m_1,n_2+m_2)\\
		&\qquad \qquad\left(z\ketbra{n_1+k_1,n_2+k_2}{n_1+\ell_1,n_2+\ell_2}+\overline{z}\ketbra{n_1+\ell_1,n_2+\ell_2}{n_1+k_1,n_2+k_2}\right)\,.
	\end{align*}
	Next, we consider the constituents of the above sum individually. Note that the operator
	\begin{equation}\label{eq:block-matrix}
		\begin{aligned}
			z\ketbra{n_1+k_1,n_2+k_2}{n_1+\ell_1,n_2+\ell_2}+\overline{z}\ketbra{n_1+\ell_1,n_2+\ell_2}{n_1+k_1,n_2+k_2}
		\end{aligned}
	\end{equation}
	\begin{equation*}
		\left(\begin{array}{*{6}{c}}
			\tikzmark{left}0 &0 &\ast &0 &0 &0 \\               
			0 &0 &0 &\ast &0 &0 \\      
			\ast &0 &0\tikzmark{right} &0 &\ast &0 \\\DrawBox[thick]
			0 &\ast &0 &\tikzmark{left}0 &0 &\ast \\
			0 &0 &\ast &0 &0 &0 \\
			0 &0 &0 &\ast &0 &0\tikzmark{right} \\\DrawBox[thick]
		\end{array}\right)\,.
	\end{equation*}
	can be embedded into an operator on a two-dimensional space of the form 
	\begin{equation*}
		z\ketbra{e_1}{e_2} + \overline{z} \ketbra{e_2}{e_1} \, , 
	\end{equation*}
	where $\ket{e_1}$ and $\ket{e_2}$ are orthonormal vectors. For $\ket{e_1}=\ket{e_2}$, $z+\overline{z}\leq2|z|$ shows 
	\begin{equation*}
		z\ketbra{e_1}{e_2} + \overline{z} \ketbra{e_2}{e_1} \le |z|(\ketbra{e_1}{e_1} + \ketbra{e_2}{e_2})\, . 
	\end{equation*}
	In the case $\ket{e_1}\neq\ket{e_2}$, we have
	\begin{center}
		\begin{tabular}{ c c }
			Eigenvalue & Eigenvectors \\[0.5ex]\hline
			$|z|$ & $\ket{\psi}=\frac{1}{\sqrt{2}|z|}(|z|\ket{e_1}+z\ket{e_{2}})$\\
			$-|z|$ & $\ket{\varphi}=\frac{1}{\sqrt{2}|z|}(|z|\ket{e_1}-z\ket{e_{2}})$
		\end{tabular}.
	\end{center}
	so that 
	\begin{equation*}
		z\ketbra{e_{2}}{e_1}+\overline{z}\ketbra{e_1}{e_{2}}=|z| \ketbra{\psi}{\psi} - |z| \ketbra{\varphi}{\varphi}\leq|z|\ketbra{\psi}{\psi}\leq|z| (\ketbra{e_1}{e_1} + \ketbra{e_2}{e_2})\,.
	\end{equation*}
	This allows us to estimate
	\begin{align*}
		K &\le \sum\limits_{n_1, n_2 \in \N} \widetilde h_{m_1, m_2}(n_1 + m_1, n_2 + m_2) |z| (\ketbra{n_1 + k_1, n_2 + k_2}{n_1 + k_1, n_2 + k_2}\\
		&\hspace{7cm} + \ketbra{n_1 + l_1, n_2 + l_2}{n_1 + l_1, n_2 + l_2})\\
		&\le 2 |z| \widetilde h_{m_1, m_2}(N_1 + m_1 I, N_2 + m2 I)
	\end{align*}
	employing the monotonicity of $\widetilde h_{m_1, m_2}$ in both arguments in the last step.
\end{proof}

\section{Inequalities for power functions}
Many bounds in Sections \ref{sec:examples-sobolev-preserving-semigroup} and \ref{sec:example-perturbation-bounds} can be deduced from bounds on real-valued functions acting on the spectrum of the number operator $N$. Especially, the following functions, first introduced in \Cref{eq:f-g-l-function}, will require special attention: Let $l,k\in\N$, $f(x)=(x+1)^{k/2} 1_{x\ge -1}$, and 
\begin{equation}\label{eq-appx:f-g-l-function}
	g_l(x) = \begin{cases}
		f(x) - f(x - l) & x \ge l-1;\\
		f(x) & l-1 > x \ge 0;\\
		0 & 0 > x\,.
	\end{cases}
\end{equation}
\begin{lem}\label{appx-lem:monotonicity-g}
	Let $g_l$ be defined in \Cref{eq-appx:f-g-l-function} for $l,k\in\N$. Then, for all $k\geq2$ and $x\in\R$
	\begin{align}
		g_l(x)&\leq g_{l+1}(x)\label{appx-eq:monotonicity-k}\,,\\
		g_l(x-l)&\leq g_l(x)\label{appx-eq:monotonicity-x}\,.
	\end{align}
\end{lem}
\begin{proof}
	By the monotonicity and non-negativity of $f(x)=(x+1)^{k/2} 1_{x\ge -1}$,
	\begin{equation*}
		\begin{rcases}
			x\ge l-1 & f(x)-f(x-l)\\
			l-1>x\geq 0 & f(x)\\
			0>x& 0
		\end{rcases}
		=g_l(x)\leq g_{l+1}(x)=
		\begin{cases}
			f(x)-f(x-(l+1)) & x\ge l \\
			f(x) & l>x\geq 0 \\
			0 & 0>x
		\end{cases}\,,
	\end{equation*}
	which proves Inequality \ref{appx-eq:monotonicity-k}. For Inequality \ref{appx-eq:monotonicity-x}, we consider the following cases:
	\begin{equation*}
		\begin{rcases}
			f(x-l)-f(x-2l)\\
			f(x-l)\\
			0 \\
			0
		\end{rcases}
		=g_l(x-l)\leq g_l(x)=
		\begin{cases}
			f(x)-f(x-l) & \qquad x\geq 2l-1\\
			f(x)-f(x-l) & \qquad 2l-1>x\geq l-1 \\
			f(x) & \qquad l-1>x\geq 0 \\
			0 & \qquad 0>x
		\end{cases}\,.
	\end{equation*}
	For $x<l-1$, the inequalities are clear by the non-negativity of $f$. The case $l-1\leq x<2l-1$ follows by
	\begin{equation*}
		2\frac{f(x-l)}{f(x)}=2\left(1-\frac{l}{x+1}\right)^{k/2}\leq2\left(1-\frac{1}{2}\right)^{k/2}=1
	\end{equation*}
	and the last case $x\geq 2l-1$ follows by monotonicity of $g_l$:
	\begin{equation*}
		\frac{2}{k}g_l'(x-l)=(x-l)^{k/2-1}-(x-2k)^{k/2-1}\geq0\,.\vspace{-2ex}
	\end{equation*}
\end{proof}

Next, we prove upper and lower bounds for $g_l$:
\begin{lem}\label{lem:upper-lower-bound-gl}
	Let $g_l:\R\rightarrow\R_{\geq0}$ be defined in \Cref{eq-appx:f-g-l-function} for $l\in\N$. Then, for all $x\in\R$ and $k\in\N$,
	\begin{equation*}
		\begin{rcases}
			x\geq l-1 & (x+1)^{k/2-1}\frac{kl}{2}-1_{k\geq3}(x+1)^{k/2-2}\,\,\frac{(kl)^2}{8} \\
			x\geq l-1 & (x+1)^{k/2-1}l\\
			l-1>x\geq0&(x+1)^{k/2}\\
			0>x&0
		\end{rcases}
		\leq g_l(x)
	\end{equation*}
	and 
	\begin{equation*}
		g_l(x)\leq 
		\begin{cases}
			\frac{kl}{2}\left(1+1_{k=1}\right)(x+1)^{k/2-1} & x\geq 0\\
			(x+1)^{k/2}&x\geq0 \\
			0 & 0>x
		\end{cases}\,.
	\end{equation*}
\end{lem}
\begin{proof}
	The case $k=0$ is trivial. We start with the upper bounds. By monotonicity of $g_l$, it is enough to prove the first upper bound just for $x\geq l-1$. For $k=1$,
	\begin{equation*}
		g_l(x)=(x+1)^{-1/2}\frac{l}{2}\int_0^1\left(1-s\frac{l}{x+1}\right)^{-1/2}ds\leq(x+1)^{-1/2}\frac{l}{2}\int_0^1\left(1-s\right)^{-1/2}ds=(x+1)^{-1/2}l.
	\end{equation*}
	For $k\geq2$,
	\begin{equation*}
		g_l(x)=\frac{k}{2}\int_0^l(x+1-s)^{k/2-1}ds\leq\frac{kl}{2}(x+1)^{k/2-1}\,,
	\end{equation*}
	which finishes the proof of the first upper bound. The other two bounds are obvious by definition. Next, we consider the lower bounds. The case $x<l-1$ is trivial so we are left with proving
	\begin{equation*}
		(x+1)^{k/2-1}\frac{kl}{2}-\delta_{k\geq3}(x+1)^{k/2-2}\frac{(kl)^2}{8}\leq g_l(x)
	\end{equation*}
	for $x\geq l-1$. For $k=1$, the integral representation can be lower bounded as
	\begin{equation*}
		\begin{aligned}
			g_l(x)&=(x+1)^{-1/2}\frac{l}{2}\int_{0}^1\left(1-s\frac{l}{x+1}\right)^{-1/2}ds\geq(x+1)^{-1/2}\frac{l}{2}\,.
		\end{aligned}
	\end{equation*}
	For $k=2$, it is again easy to calculate the quantity $g_l(x)=l$, and for $k=3$
	\begin{equation*}
		\begin{aligned}
			g_l(x)&=(x+1)^{1/2}\frac{3l}{2}\int_{0}^1\left(1-s_1\frac{l}{x+1}\right)^{1/2}ds_1\\
			&=(x+1)^{1/2}\frac{3l}{2}-(x+1)^{-1/2}l^2\frac{3}{4}\iint_{0}^1s_1\left(1-s_1s_2\frac{l}{x+1}\right)^{-1/2}ds_2ds_1\\
			&\geq (x+1)^{1/2}\frac{3l}{2}-(x+1)^{-1/2}l^2\frac{3}{4}\iint_{0}^1s_1\left(1-s_1s_2\right)^{-1/2}ds_2ds_1\\
			&=(x+1)^{1/2}\frac{3l}{2}-(x+1)^{-1/2}\frac{l^2}{2}.
		\end{aligned}
	\end{equation*}
	Finally, the case $k\geq4$ is given by
	\begin{equation*}
		\begin{aligned}
			g_l(x)&=(x+1)^{k/2-1}\frac{kl}{2}\int_{0}^1\left(1-s_1\frac{l}{x+1}\right)^{k/2-1}ds_1\\
			&=(x+1)^{k/2-1}\frac{kl}{2}-(x+1)^{k/2-2}l^2\frac{k(k-2)}{4}\iint_{0}^1s_1\left(1-s_1s_2\frac{l}{x+1}\right)^{k/2-2}ds_2ds_1\\
			&\geq(x+1)^{k/2-1}\frac{kl}{2}-(x+1)^{k/2-2}l^2\frac{k(k-2)}{4}\int_{0}^1s_1ds_1\\
			&\geq(x+1)^{k/2-1}\frac{kl}{2}-(x+1)^{k/2-2}\frac{(kl)^2}{8}\,
		\end{aligned}
	\end{equation*}
	which proves the first non-trivial lower bound for $x\geq l-1$. Next, we consider  
	\begin{equation*}
		g_l(x)\geq(x+1)^{k/2-1}l\,.
	\end{equation*}
	The inequality is obvious for $k<2$ by the same idea as before and for $k\geq2$
	\begin{equation*}
		\begin{aligned}
			g_l(x)&=(x+1)^{k/2-1}\frac{kl}{2}\int_{0}^1\left(1-s_1\frac{l}{x+1}\right)^{k/2-1}ds_1\\
			&\geq(x+1)^{k/2-1}\frac{kl}{2}\int_{0}^1\left(1-s_1\right)^{k/2-1}ds_1\\
			&\geq(x+1)^{k/2-1}l
		\end{aligned}
	\end{equation*}
	which ends the proof. 
\end{proof}

\begin{lem}\label{lem:bounds-ccr-l-product}
	Let $l\in\N$ and $x\geq l$, then
	\begin{equation*}
		\begin{aligned}
			(x+1)^l-\frac{(l+1)l}{2}(x+1)^{l-1}&\leq&((x+1)-l)\cdots ((x+1)-1)&\leq&(x+1)^l\\
			(x+1)^l&\leq&(x+1)\cdots (x+1+(l-1)&\leq&l!(x+1)^l
		\end{aligned}
	\end{equation*}
\end{lem}
\begin{proof}
	To prove Lemma \ref{lem:bounds-ccr-l-product}, we redefine $y=x+1$ and rewrite the first product as
	\begin{equation*}
		p_l(y)\coloneqq(y-l)\cdots (y-1)\eqqcolon y^l-\frac{(l+1)l}{2}y^{l-1}+r_{l-2}(y)
	\end{equation*}
	where $r_{l-2}$ is a polynomial of degree $l-2$. The proof idea is to show that $r_{l-2}(y)$ is non-negative for all $y\geq l+1$, which proves the inequality. The non-negativity of the polynomial $r_{l-2}$ can be proven by induction over $l$: The statement is directly clear for $l=1$ and $l=2$. Next, we assume that $r_{l-2}$ is non-negative for all $x\geq l+1$ and show that $r_{l-1}$ is for all $x\geq l+2$. 
	\begin{equation*}
		\begin{aligned}
			p_{l+1}(y)&=(y-(l+1))p_l(y)\\
			&=(y-(l+1))\left(y^l-\frac{(l+1)l}{2}y^{l-1}+r_{l-2}(y)\right)\\
			&=y^{l+1}-\frac{(l+1)(l+2)}{2}y^{l}+\frac{(l+1)^2l}{2}y^{l-1}+(y-(l+1))r_{l-2}(y)\\
			&=y^{l+1}-\frac{(l+1)(l+2)}{2}y^{l}+r_{l-1}(y).
		\end{aligned}
	\end{equation*}
	For the second product $(x+1)\cdots ((x+1)+l-1)$ the lower bound is clear and the upper bound follows by 
	\begin{equation*}
		(x+1)(x+2)\cdots ((x+1)+l-1)=(l-1)!\left(\frac{x}{1}+1\right)\cdots \left(\frac{x}{l}+1\right)\leq l!(x+1)^l.
	\end{equation*}
\end{proof}

\section{Technical lemmas for the quantum Sobolev spaces}
\begin{lem}[Continuity of $G(z)$]\label{lem:continuity-G}
	Let $k_0 < k_1 \in \R_+$ and $T: W^{k_j, 1} \to W^{k_j, 1}$, be a linear map with $\norm{T}_{W^{k_j, 1} \to W^{k_j, 1}} \le M_j$, bounded by $M_j \ge 0$ for $j = 1, 2$ respectively. Further let $\theta \in [0, 1]$ and $k_\theta = (1-\theta) k_0 +  \theta k_1$ and $x \in \cT_f$, then the map
	\begin{align*}
		G: S &\coloneqq \{ z \in \C \::\; 0 \le \Re(z) \le 1\} \to \cT_{1, \operatorname{sa}} \\
		z & \mapsto G(z) = (N + \1)^{\frac{k(z)}{4}} T\Big((N + \1)^{\frac{k_\theta - k(z)}{4}} x (N + \1)^{\frac{k_\theta - k(z)}{4}}\Big) (N + \1)^{\frac{k(z)}{4}}
	\end{align*}
	with $k(z) = (1-z) k_0 + z k_1$, is well-defined, uniformly bounded and continuous.
\end{lem}
\begin{proof}
	To prove the claim, we decompose $G$ using the following auxiliary functions:
	\begin{equation}\label{eq:G_1}
		\begin{aligned}
			G_1: S \times W^{k_1, 1} &\to \cT_{1, \operatorname{sa}}\\
			(z, y) &\mapsto (N + \1)^{\frac{k(z)}{4}} y (N + \1)^{\frac{k(z)}{4}}
		\end{aligned}
	\end{equation}
	and 
	\begin{equation}\label{eq:G_2}
		\begin{aligned}
			G_2: S &\to \cT_f \subset W^{k_1, 1}\\
			z &\mapsto (N + \1)^{\frac{k_\theta - k(z)}{4}} x (N + \1)^{\frac{k_\theta - k(z)}{4}} \, . 
		\end{aligned}
	\end{equation}
	We clearly have that $G_1(z, \cdot):W^{k_1, 1} \to \cT_{1, \operatorname{sa}}$ is a bounded linear map for all $z \in S$, since
	\begin{equation}\label{eq:bound-G_1}
		\norm{G_1(z, y)}_1 = \norm{(N + \1)^{\frac{\Re(k(z))}{4}} y (N + \1)^{\frac{\Re(k(z))}{4}}}_1 \le \norm{(N + \1)^{\frac{k_1}{4}} y (N + \1)^{\frac{k_1}{4}}}_1 = \norm{y}_{W^{k_1, 1}}
	\end{equation}
	where we used that $k_0 \le \Re(k(z)) \le k_1$ and $(N + \1)^{i\frac{\Im(k(z))}{4}}$ is a unitary that can be absorbed into the norm. Next, we will show that $G_1(\cdot, y): S \to \cT_{1, \operatorname{sa}}$ is continuous for all $y \in W^{k_1, 1}$. For that first note that, for $y \in \cT_f$, the claim follows directly from the continuity of $z \mapsto (n + 1)^{\frac{k(z)}{4}}$ with $n \in \N$ as a map from $S$ to $\C$. This is because all the involved operators can be considered finite dimensional using a cut-off of the Fock-basis. For a general $y \in W^{k_1, 1}$, we find $\{y_n\}_{n \in \N} \subset \cT_f$, s.t. $y_n \to y$ in $W^{k_1, 1}$, hence for all $n \in \N$
	\begin{align*}
		\lim\limits_{z \to z_0} &\norm{G_1(z, y) - G_1(z_0, y)}_1 \\
		&\le \lim\limits_{z \to z_0} \norm{G_1(z, y - y_n)}_1 + \norm{G_1(z, y_n) - G_1(z_0, y_n)}_1 + \norm{G_1(z_0, y_n - y)}_1\\
		&\le \lim\limits_{z \to z_0} \norm{G_1(z, y_n) - G_1(z_0, y_n)}_1 + 2 \norm{y - y_n}_{W^{k_1, 1}}\\
		&\le 2 \norm{y - y_n}_{W^{k_1, 1}} \, , 
	\end{align*}
	where we used \Cref{eq:bound-G_1}. Taking the limit $n \to \infty$ concludes the claim that $G_1(\cdot, y): S \to \cT_{1, \operatorname{sa}}$ is continuous for all $y \in W^{k_1, 1}$. We further have that $G_2$ as a map from $S$ to $W^{k_1, 1}$ is continuous, since $x \in \cT_f$ and the maps $z \mapsto (n + 1)^{\frac{k(z)}{4}}$ for $n \in \N$ are continuous as maps $S \to \C$. This suffices since $x \in \cT_f$, hence all involved operators can be made finite dimensional via a cut-off in the Fock-basis again.\par 
	We can now write 
	\begin{equation*}
		G(z) = G_1(z, T(G_2(z)))
	\end{equation*}
	where $T(G_2(z)) \in W^{k_1, 1}$ as $T:W^{k_1, 1} \to W^{k_1, 1}$ and $G_2(z) \in \cT_f \subset W^{k_1, 1}$ for all $z \in S$. This not only gives us that $G$ is well-defined but also allows us to get
	\begin{align*}
		\norm{G(z)}_1 &= \norm{G_1\left(z, T(G_2(z))\right)}_1 \\
		&\le \norm{T(G_2(z))}_{W^{k_1, 1}} \\
		&\le \norm{T}_{W^{k_1, 1} \to W^{k_1, 1}} \norm{G_2(z)}_{W^{k_1, 1}}\\
		&\le \norm{T}_{W^{k_1, 1} \to W^{k_1, 1}} \norm{(N + \1)^{\frac{k_\theta - k_0}{4}} x (N + \1)^{\frac{k_\theta - k_0}{4}}}_{W^{k_1, 1}}
	\end{align*}
	where we again used \Cref{eq:bound-G_1}, giving us a bound independent of $z$. Further, using again \Cref{eq:bound-G_1} we can conclude continuity, since 
	\begin{align*}
		\lim\limits_{z \to z_0} \norm{G(z) - G(z_0)}_1 &\le \lim\limits_{z \to z_0} \norm{G_1\left(z, T\left\{G_2(z) - G_2(z_0)\right\}\right)}_1 \\
		&\hspace{2cm} + \lim\limits_{z \to z_0}\norm{G_1\left(z, T\left(G_2(z_0)\right)\right) - G_1\left(z_0, T\left(G_2(z_0)\right)\right)}_1\\
		&\le \lim\limits_{z \to z_0}\norm{T}_{W^{k_1, 1} \to W^{k_1, 1}} \norm{G_2(z) - G_2(z_0)}_{W^{k_1, 1}} \\
		& \hspace{2cm} + \lim\limits_{z \to z_0} \norm{G_1\left(z, T\left(G_2(z_0)\right)\right) - G_1\left(z_0, T\left(G_2(z_0)\right)\right)}_1\\
		&= 0
	\end{align*}
	where in addition we used the continuity of $G_1(\cdot, y): S \to \cT_{1, \operatorname{sa}}$ for $y \in W^{k_1, 1}$ and $G_2:S \mapsto W^{k_1, 1}$.
\end{proof}

\begin{lem}[Differentiability of G(z)]\label{lem:differentiability-G}
	Let $k_0 < k_1 \in \R_+$, $T: W^{k_j, 1} \to W^{k_j, 1}$, be a linear map with $\norm{T}_{W^{k_j, 1} \to W^{k_j, 1}} \le M_j$, bounded by $M_j \ge 0$ for $j = 1, 2$ respectively. Further let $\theta \in [0, 1]$ and $k_\theta = (1-\theta) k_0 +  \theta k_1$ and $x \in \cT_f$, then the map
	\begin{align*}
		G: S &\coloneqq \{ z \in \C \::\; 0 \le \Re(z) \le 1\} \to \cT_{1, \operatorname{sa}} \\
		z & \mapsto G(z) = (N + \1)^{\frac{k(z)}{4}} T\left((N + \1)^{\frac{k_\theta - k(z)}{4}} x (N + \1)^{\frac{k_\theta - k(z)}{4}}\right) (N + \1)^{\frac{k(z)}{4}}
	\end{align*}
	with $k(z) = (1-z) k_0 + z k_1$, is holomorphic on $\mathring{S} \coloneqq \{z \in \C \; : \; 0 < \Re(z) < 1\}$.
\end{lem}
\begin{proof}
	To prove the claim, we follow a similar strategy as with Lemma \ref{lem:continuity-G}. We will again use the auxiliary functions \Cref{eq:G_1} and \Cref{eq:G_2}. We begin by showing that for a fixed $y \in W^{k, 1}$, $G_1(\cdot, y):S \to \cT_{1, \operatorname{sa}}$ is holomorphic on $\mathring{S}$ and initially even simplify to the case $y \in \cT_f$. In this setting, all operators involved can be assumed to be linear maps on a finite subspace, by just taking a cut-off in the Fock-basis as we did before. This allows us to Taylor expand around $z_0 \in \mathring{S}$
	\begin{align*}
		(N + \1)^{\frac{k(z)}{4}} y (N + \1)^{\frac{k(z)}{4}} = G_1(z_0, y) +  G_1'(z_0, y) (z - z_0) + \int\limits_{[z_0, z]} G_1''(\omega, y)(\omega - z_0) \,  d\omega
	\end{align*}
	where the integral is a path integral along the line segment $[z_0, z]$ and
	\begin{align*}
		G_1'(z_0, y) &= \frac{k_0 - k_1}{4}\, \left(\log(N + \1) (N + \1)^{\frac{k(z_0)}{4}} y (N + \1)^{\frac{k(z_0)}{4}}\right. \\
		&\hspace{3cm} \left. + (N + \1)^{\frac{k(z_0)}{4}} y (N + \1)^{\frac{k(z_0)}{4}} \log(N + \1)\right)
	\end{align*}
	and 
	\begin{align*}
		G_1''(\omega, y) &= \left(\frac{k_0 - k_1}{4}\right)^2\left(\log^2(N + \1)(N + \1)^{\frac{k(\omega)}{4}} y (N + \1)^{\frac{k(\omega)}{4}}\right.\\
		&\hspace{3cm} + 2 \log(N + \1)(N + \1)^{\frac{k(\omega)}{4}} y (N + \1)^{\frac{k(\omega)}{4}} \log(N + \1) \\
		&\hspace{6cm} +\left. (N + \1)^{\frac{k(\omega)}{4}} y (N + \1)^{\frac{k(\omega)}{4}} \log^2(N + \1)\right)
	\end{align*}
	are linear in $y$. From this representation, we can immediately deduce holomorphy of $G_1(\cdot, y):S \to \cT_{1, \operatorname{sa}}$ at $z_0 \in \mathring{S}$ and hence on all of $\mathring{S}$. To lift holomorphy from $y \in \cT_f$ to $y \in W^{k_1, 1}$, we note that for $z_0 \in \mathring{S}$ there exists $C_{z_0} \ge 0$ such that for $y \in \cT_f$
	\begin{equation}\label{eq:bound-G'1}
		\norm{G'(z_0, y)}_1 \le C_{z_0} \norm{y}_{W^{k_1, 1}}
	\end{equation}
	and further for $\omega \in B_\varepsilon(z_0) \coloneqq \{z \in \C \;:\; |z - z_0| < \varepsilon\} \subset \mathring{S}$ there exists $C_{\varepsilon, z_0} \ge 0$ such that
	\begin{equation}\label{eq:bound-G''1}
		\norm{G''(\omega, y)}_1 \le C_{\varepsilon, z_0} \norm{y}_{W^{k_1, 1}} \, . 
	\end{equation}
	We will only show that given $\omega$ as above, 
	\begin{equation}\label{eq:boundedness-subterms}
		\norm{\log^2(N + \1) (N + \1)^{\frac{k(\omega)}{4}} y (N + \1)^{\frac{k(\omega)}{4}} y (N + \1)^{\frac{k(\omega)}{4}}}_1 \le \tilde{C}_{\varepsilon, z_0} \norm{y}_{W^{k_1, 1}} \, . 
	\end{equation}
	Using the same reasoning for the other terms of \Cref{eq:bound-G'1} and \Cref{eq:bound-G''1} in combination with triangle inequality immediately gives the claims. Note first that we can reduce $k(\omega)$ to its real part since the imaginary part only produces a unitary $(N + \1)^{i \Im(k(\omega))}$ that can be absorbed into the norm. We call the real part $r(\omega)$ for now. Since $\omega \in B_\varepsilon(z_0) \subset \mathring{S}$ we find a $\delta_\varepsilon > 0$ independent of $\omega$, such that $|r(\omega) - k_1| < \delta_\varepsilon$ or more precisely $r(\omega) - k_1 \le - \delta_\varepsilon$. Hence using Hölder's inequality, we can deduce
	\begin{align*}
		&\norm{\log^2(N + \1) (N + \1)^{\frac{k(\omega)}{4}} y (N + \1)^{\frac{k(\omega)}{4}} y (N + \1)^{\frac{k(\omega)}{4}}}_1\\
		&\le \norm{\log^2(N + \1) (N + \1)^{\frac{r(\omega) - k_1}{4}}}_\infty \norm{(N + \1)^{\frac{r(\omega) - k_1}{4}}}_\infty \norm{y}_{W^{k_1, 1}}\\
		&\le \norm{\log^2(N + \1) (N + \1)^{-\frac{\delta_\varepsilon}{4}}}_\infty \norm{(N + \1)^{-\frac{\delta_\varepsilon}{4}}}_\infty \norm{y}_{W^{k_1, 1}}\\
		&\le \norm{\log^2(N + \1) (N + \1)^{-\frac{\delta_\varepsilon}{4}}}_\infty \norm{y}_{W^{k_1, 1}}
	\end{align*}
	where we used that $x \mapsto e^{k x}$ for $k \ge 0$ is monotone and further that $(N + \1)^{-\frac{\delta_\varepsilon}{4}}$ is a contraction. Lastly, we have that $x \mapsto \frac{\log^2(x + 1)}{(x + 1)^{\frac{\delta_\varepsilon}{4}}}$ is a bounded function for $x \ge 0$ with a bound we call $\tilde{C}_{\delta_\varepsilon}$. This allows us to estimate $\norm{\log^2(N + \1) (N + \1)^{-\frac{\delta_\varepsilon}{4}}}_\infty \le \tilde C_{\delta_\varepsilon}$, which concludes \Cref{eq:boundedness-subterms} and therefore also \Cref{eq:bound-G'1} and \Cref{eq:bound-G''1}.\par 
	For a general $y \in W^{k_1, 1}$ and $z_0 \in \mathring{S}$ \Cref{eq:bound-G'1} allows us to conclude that $G'_1(z_0, y) \in \cT_{1, \operatorname{sa}}$ is well defined. Further, for $z \in B_\varepsilon(z_0)$ and $(y_n)_{n \in \N} \subset \cT_f$ with $y_n \to y$ in $W^{k_1, 1}$, we have for all $n \in \N$
	\begin{equation}
		\begin{aligned}
			\norm{\frac{G_1(z, y_n) - G_1(z_0, y_n)}{z - z_0} - G_1'(z_0, y_n)}_1 &\le \frac{1}{|z - z_0|} \int\limits_{[z_0, z]} \norm{G''_1(\omega, y_n)}_1 |(\omega - z_0) d\omega| \\
			&\le C_{\varepsilon, z_0} |z - z_0| \norm{y_n}_{W^{k_1, 1}}
		\end{aligned}
	\end{equation}
	where we used the expansion and \Cref{eq:bound-G''1}. Now we can take the limit $n \to \infty$ on both sides, as all objects involved are stable w.r.t.~that limit (using Lemma \ref{lem:continuity-G} and \Cref{eq:bound-G'1}). We get
	\begin{equation}
		\norm{\frac{G_1(z, y) - G_1(z_0, y)}{z - z_0} - G_1'(z_0, y)}_1 \le  C_{\varepsilon, z_0} |z - z_0| \norm{y}_{W^{k_1, 1}}
	\end{equation}
	which immediately lets us deduce holomorphy of $G_1(\cdot, y):S \to \cT_{1, \operatorname{sa}}$ on $\mathring{S}$ for $y \in W^{k_1, 1}$.\par
	For $G_2:S \to W^{k_1, 1}$ the holomorphy immediately follows from the fact that $x \in \cT_f$, which again allows reducing the analysis to a finite-dimensional subspace by taking a cut-off in the Fock basis again. Lastly, we have that $T(G_2(z)) \in W^{k_1, 1}$ for all $z \in S$, which finally gives us that for $z_0 \in \mathring{S}$ and for $z \in B_\varepsilon(z_0) \subset \mathring{S}$
	\begin{align*}
		&\norm{\frac{G(z) - G(z_0)}{z - z_0} - (G_1'(z_0, T(G_2(z_0))) + G_1(z_0, T\{G_2'(z_0)\})} \\
		&\le \norm{\frac{G_1(z, T(G_2(z_0))) - G_1(z_0, T(G_2(z_0)))}{z - z_0} - G_1'(z_0, T(G_2(z_0)))}_1 \\
		&\hspace{1cm} + \norm{T}_{W^{k_1, 1} \to W^{k_1, 1}}\norm{\frac{G_2(z) - G_2(z_0)}{z - z_0} - G'_2(z_0)}_{W^{k_1, 1}}
	\end{align*}
	where we used linearity of $G_1(z, \cdot)$, $G_1'(z, \cdot)$ and $T$. In addition, we used the bound on $G_1(z, \cdot)$ from \Cref{eq:bound-G_1} and $G'_2$ to denote the derivative of $G_2$. Now the differentiability of $G_1(\cdot, y)$ and $G_2$ at $z_0$ immediately gives the differentiability of $G$ at $z_0$, which concludes the proof as $z_0 \in \mathring{S}$ was arbitrary.
\end{proof}

\section{Technical lemmas for the generation theorem}
\begin{lem}\label{lem:(n+1)-(n+1)-properties}
	For $d \ge 0$ and $\varepsilon > 0$, define the operator 
	\begin{equation*}
		\cI_{d, \varepsilon}: \cT_f \to \cT_f, \qquad x \mapsto \cI_{d, \varepsilon}(x) \coloneqq - \varepsilon\{(N + \1)^{4d}, x\} \, . 
	\end{equation*}
	For $\lambda \ge 0$, we have that $\lambda - \cI_{d, \varepsilon}:\cT_f \to \cT_f$ is bijective. While for all $k \in \R_+$ and $x \in \cT_f$ one further has
	\begin{align}
		& \norm{(\lambda - \cI_{d, \varepsilon})^{-1}(x)}_{W^{k, 1} } \le \frac{1}{\lambda + 2 \varepsilon} \norm{x}_{W^{k, 1}}\,;\tag{1}\label{item:(1)}\\
		&  \norm{\cI_{d, \varepsilon} \circ (\lambda - \cI_{d, \varepsilon})^{-1}(x)}_{W^{k, 1}} \le {2}\norm{x}_{W^{k, 1}} \tag{2}\label{eq:inequality-resolvent-(n+1)4d}\, . 
	\end{align}
\end{lem}
\begin{proof}
	For $\lambda \ge 0$ define the following linear operator
	\begin{align*}
		(\lambda - \cI_{d, \varepsilon})^{-1}:\cT_f &\to \cT_f,\\
		x = \sum\limits_{\text{finite}} x_{nm} \ketbra{n}{m} &\mapsto (\lambda - \cI_{d, \varepsilon})^{-1}(x) \coloneqq \sum\limits_{\text{finite}} x_{nm} \frac{1}{\varepsilon(n + 1)^{4d} + \varepsilon(m + 1)^{4d} + \lambda} \ketbra{n}{m} \, .
	\end{align*}
	or alternatively 
	\begin{align*}
		(\lambda - \cI_{d, \varepsilon})^{-1}(x) &= \int\limits_{0}^\infty e^{-(\varepsilon(N + \1)^{4d} + \lambda/2) s} x e^{-(\varepsilon(N + \1)^{4d} + \lambda/2) s} ds\\
		&=  \int\limits_{0}^\infty  \sum\limits_{\text{finite}} e^{-(\varepsilon(n + \1)^{4d} + \lambda/2)s} x_{nm} e^{-(\varepsilon(m + \1)^{4d} + \lambda/2)s} \ketbra{n}{m} ds \, .
	\end{align*}
	The integral representation allows us to deduce that $(\lambda - \cI_{d, \varepsilon})^{-1}$ preserves positivity. Using the first expression it is straightforward to show that this map is indeed the inverse to $\lambda - \cI_{d, \varepsilon}:\cT_f \to \cT_f$. The bound $\norm{(\lambda - \cI_{d, \varepsilon})^{-1} x}_{W^{k, 1} } \le \frac{1}{\lambda + 2 \varepsilon} \norm{x}_{W^{k, 1}}$ can be shown, using the integral representation and Hölder inequality:
	\begin{align*}
		\norm{(\lambda - \cI_{d, \varepsilon})^{-1}(x)}_{W^{k, 1}} &\le \int\limits_{0}^\infty \norm{e^{-\varepsilon(N + \1)^{4d} + \lambda/2)s} (N + \1)^{k/4} x (N + \1)^{k/4} e^{-\varepsilon(N + \1)^{4d} + \lambda/2)s}}_1\\
		&\le \int\limits_{0}^\infty \norm{e^{-\varepsilon(N + \1)^{4d} + \lambda/2)s}}_\infty^2 ds \norm{x}_{W^{k, 1}}\\
		&= \int\limits_{0}^\infty e^{-(2\varepsilon + \lambda)s} \norm{x}_{W^{k, 1}} = \frac{1}{\lambda + 2\varepsilon} \norm{x}_{W^{k, 1}} \, .
	\end{align*}
	Issues arising from the unbounded nature of $N$ can be ignored in the above estimations, as we can take a finite cut-off in the Fock basis due to $x \in \cT_f$. For \eqref{eq:inequality-resolvent-(n+1)4d}, we have that on $\cT_f$, $-\cI_{d, \varepsilon}\circ (\lambda - \cI_{d, \varepsilon})^{-1} = \1 - \lambda (\lambda - \cI_{d, \varepsilon})^{-1}$, i.e.~for $x \in \cT_f$
	\begin{equation}
		\norm{(\1 - \lambda (\lambda - \cI_{d, \varepsilon})^{-1})x}_{W^{k, 1}} = \norm{-\cI_{d, \varepsilon}\circ (\lambda - \cI_{d, \varepsilon})^{-1} (x)}_{W^{k, 1}}
	\end{equation}
	where the LHS can be upper bounded by $(1 + \frac{\lambda}{\lambda + 2\varepsilon})\norm{x}_{W^{k, 1}} \le 2 \norm{x}_{W^{k, 1}}$ using \eqref{item:(1)}. This proves the last claim.
\end{proof}

\begin{lem}\label{lem:boundedness-polynomials}
	For $p \in \C[X,Y]$ a polynomial of degree $d$ and 
	\begin{equation*}
		A: \cH_f \to \cH, \qquad \ket\psi = \sum\limits_{\text{finite}}\psi_n \ket{n} \mapsto p(a, a^\dagger)\ket{\psi} = \sum\limits_{\text{finite}}\psi_n p(a, a^\dagger)\ket{n} \, ,
	\end{equation*}
	we get that for all $k \in \R_+$ and $d' \ge d$
	\begin{align*}
		B_1: \cH_f \to \cH, \qquad \ket\psi \mapsto (N + \1)^{k} A (N + \1)^{-k - d'} \ket\psi\\
		B_2: \cH_f \to \cH, \qquad \ket\psi \mapsto (N + \1)^{-k - d'} A (N + \1)^{k} \ket\psi
	\end{align*}
	are bounded and therefore can be uniquely extended to a bounded map on $\cH$, with the same bound.
\end{lem}
\begin{proof}
	Since the proof for $B_1$ and $B_2$ are almost completely analogous, we will only show it here for $B_1$. The canonical commutation relation allows us to rewrite $A$ as a finite linear combination of monomials of the form $(a^\dagger)^i N^j$ and $a^i N^j$ with $i + j/2 \le d$. Now by triangle inequality for the norm on $\cH$ and since the sum of these monomials comprising $A$ are finite, for the claim to be true it suffices to show that 
	\begin{equation*}
		(N + \1)^k(a^\dagger)^i N^j (N + \1)^{-k - d}:\cH_f \to \cH, \qquad (N + \1)^ka^i N^j (N + \1)^{-k - d}:\cH_f \to \cH
	\end{equation*}
	are bounded, and hence can be uniquely extended to a bounded map on $\cH$. We only give the argument for the first map, since it is almost completely analogous to the second one. Let $\ket\psi = \sum\limits_{n = 0}^M\psi_n \ket{n} \in \cH_f$, then
	\begin{align*}
		\ket{\varphi} \coloneqq (N + \1)^k(a^\dagger)^i N^j (N + \1)^{-k - d'} \ket{\psi} &= \sum\limits_{n = 0}^M\psi_n (N + \1)^k(a^\dagger)^i N^j (N + \1)^{-k - d'}\ket{n}\\
		&= \sum\limits_{n = 0}^M\psi_n \frac{(n + 1 + j)^k}{(n + 1)^k} \frac{n^j \prod_{l = 1}^i \sqrt{n + l}}{(n + 1)^{d'}} \ket{n + i} \, . 
	\end{align*}
	Hence
	\begin{align*}
		\norm{\varphi}^2 &= \sum\limits_{n = 0}^M \frac{(n + 1 + j)^{2k}}{(n + 1)^{2k}} \frac{n^{2j} \prod_{l = 1}^i (n + l)}{(n + 1)^{2d'}} |\psi_n|^2\\
		&\le \sum\limits_{n = 0}^M \frac{(n + 1 + d)^{2k}}{(n + 1)^{2k}} \frac{(n + d)^{2j}(n + d)^i}{(n + 1)^{2d}} |\psi_n|^2\\
		&\le \sum\limits_{n = 0}^M d^{2k} d^{2d} |\psi_n|^2\\
		&= d^{2(k + d)} \norm{\psi}^2 \, ,
	\end{align*}
	where we used that $i + j/2 \le d$ and $d \le d'$. Hence $(N + \1)^k (a^\dagger)^i N^j(N + \1)^{-k - d'}:\cH_f \to \cH$ is bounded by $d^{k + d}$ and can be uniquely extended to a bounded linear map on $\cH$. This concludes the claim.
\end{proof}

\begin{lem}\label{lem:infinitesimal-boundedness-W-k-1}
	Let $K \in \N$. For $i = 1, \hdots, K$ let $p_{i, 1}, p_{i, 2} \in \C[X,Y]$ polynomials of degree $d_{i, 1}$, $d_{i, 2}$ such that 
	\begin{equation*}
		\cA: \cT_f \to \cT_f, \qquad x \mapsto \cA(x) = \sum\limits_{i = 1}^KA_{i, 1}x A_{i, 2} = \sum\limits_{i = 1}^K p_{i, 1}(a, a^\dagger) x \, p_{i, 2}(a, a^\dagger)
	\end{equation*}
	where the action of $p_{1/2}(a, a^\dagger)$ on $x$ is defined via the action of $a$ and $a^\dagger$ on $\ketbra{n}{m}$. We then have that for all $k \ge 0$, $d \ge \max\limits_{i = 1, \hdots, K}\max\{d_{i, 1}, d_{i, 2}\}$ there exists $C_k \ge 0$, s.t. for all $\varepsilon \ge 0$ and $\forall x \in \cT_f$
	\begin{equation*}
		\norm{\cA(x)}_{W^{k, 1}} \le \varepsilon\norm{\{(N + \1)^{4d}, x\}}_{W^{k, 1}} + \frac{C_k}{\varepsilon} \norm{x}_{W^{k, 1}} \, . 
	\end{equation*}
\end{lem}
\begin{proof}
	Let $k \in \R_+$. The first step is to show that there exists $c_k \ge 0$, s.t. for all $x \in \cT_f$
	\begin{equation}\label{eq:proof-first-stage-inequality}
		\norm{\cA(x)}_{W^{k, 1}} \le c_k \norm{(N + \1)^d x (N + \1)^d}_{W^{k, 1}} \, . 
	\end{equation}
	The argument reduces to showing that for $i = 1, \hdots, K$ and $x \in \cT_f$
	\begin{equation*}\label{eq:proof-first-stage-inequality-substage}
		\norm{(N + \1)^{k/4} A_{i, 1} x A_{i, 2} (N + \1)^{k/4}}_1 \le c_{i, k} \norm{(N + \1)^d x (N + \1)^d}_{W^{k, 1}}
	\end{equation*}
	since the sum comprising $\cA$ is finite. Note that the trace norm on the LHS is the one on the trace-class operators since the argument might not necessarily be self-adjoint. For $x \in \cT_f$, we have
	\begin{align*}
		&\norm{(N + \1)^{k/4} A_{i, 1} x A_{i, 2} (N + \1)^{k/4}}_1\\
		&\hspace{0.5cm} = \|(N + \1)^{k/4} A_{i, 1} (N + \1)^{-k/4 - d} (N + \1)^{k/4 + d}x\\
        &\qquad\qquad\qquad\qquad(N + \1)^{k/4 + d}  (N + \1)^{-k/4 - d} A_{i, 2} (N + \1)^{k/4}\|_1\\
		&\hspace{0.5cm}\le \norm{(N + \1)^{k/4} A_{i, 1} (N + \1)^{-k/4 - d}}_\infty \norm{x}_{W^{k+4d, 1}} \norm{(N + \1)^{-k/4 - d} A_{i, 2} (N + \1)^{k/4}}_\infty\\
		&\hspace{0.5cm}\le c_{i, k} \norm{(N + \1)^d x (N + \1)^d}_{W^{k, 1}} \, 
	\end{align*}
	where we used Lemma \ref{lem:boundedness-polynomials} to argue that the operators involved are bounded and we can employ Hölder's inequality to split them off. Subsequently, we replaced the operator norms with the constant $c_{i, k}$. Now in the second step we show that for all $\varepsilon > 0$ and $x \in \cT_f$
	\begin{equation*}
		\norm{(N + \1)^d x (N + \1)^d}_{W^{k, 1}} \le \varepsilon \norm{\{(N + \1)^{4d}, x\}}_{W^{k, 1}} + \frac{1}{4\varepsilon} \norm{x}_{W^{k, 1}} \, . 
	\end{equation*}
	Combining this with \Cref{eq:proof-first-stage-inequality} then immediately provides the claim.
	Therefore, let $\lambda > 0$ and $x \in \cT_f$ with $x \ge 0$. We find
	\begin{align*}
		&\norm{(N + 1)^d(\lambda - \cI_{d, 1})^{-1}(x) (N + \1)^d}_{W^{k, 1}}\\
		&\hspace{2cm} = \tr[(\lambda - \cI_{d, 1})^{-1}\{(N + \1)^{k/4 + d} x (N + \1)^{k/4 + d}\}]\\
		&\hspace{2cm} = \int\limits_{0}^\infty \tr[e^{-(2(N + \1)^{4d} + \lambda) s}(N + \1)^{2d} (N + \1)^{k/4} x (N + \1)^{k/4}] ds\\
		&\hspace{2cm} = \tr\Big[\int\limits_{0}^\infty e^{-(2(N + \1)^{4d} + \lambda) s}(N + \1)^{2d} ds\,  (N + \1)^{k/4} x (N + \1)^{k/4}\Big]\\
		&\hspace{2cm} = \tr[\frac{(N + \1)^{2d}}{(N + \1)^{4d} + \lambda}(N + \1)^{k/4} x (N + \1)^{k/4}] \, ,
	\end{align*}
	where we used the map $\cI_{d, 1}$ from Lemma \ref{lem:(n+1)-(n+1)-properties}, the integral representation of its resolvent $(\lambda - \cI_{d, 1})^{-1}$ and that this resolvent preserves positivity. Further, we applied the cyclicity of the trace and conveniently suppressed issues that might arise from the unbounded nature of $N$ by taking a cut-off in the Fock basis. This is justified by $x \in \cT_f$. Now we can use that $(N + \1)^{k/4} x (N + \1)^{k/4} \ge 0$ to bound the RHS of the above chain of inequalities to get
	\begin{align*}
		\norm{(N + 1)^d(\lambda - \cI_{d, 1})^{-1}(x) (N + \1)^d}_{W^{k, 1}} &\le \sup\limits_{s \ge 1} \frac{s}{s^2 + \lambda} \tr[(N + \1)^{k/4} x (N + \1)^{k/4}]\\
		&= \sup\limits_{s \ge 1} \frac{s}{s^2 + \lambda} \norm{x}_{W^{k, 1}} \, .
	\end{align*}
	For a general $x \in \cT_f$, we can set $y = (N + \1)^{k/4} x (N + \1)^{k/4}$ decompose into $y = y_+ - y_-$ the positive and negative part of $y$ respectively and then set $x_\pm = (N + \1)^{-k/4} y_\pm (N + \1)^{-k/4}$. We clearly have that $x_\pm \in \cT_f$, $x = x_+ - x_-$ and $x_\pm \ge 0$ as $(N + \1)^{-k/4} \cdot (N + \1)^{-k/4}$ preserves positivity. This allows us to apply what we have shown above to obtain
	\begin{align*}
		\|(N + 1)^d(\lambda - \cI_{d, 1})^{-1}&(x) (N + \1)^d\|_{W^{k, 1}} \\
        &\le \norm{(N + 1)^d(\lambda - \cI_{d, 1})^{-1}(x_+) (N + \1)^d}_{W^{k, 1}}\\
		&\hspace{2cm} + \norm{(N + 1)^d(\lambda - \cI_{d, 1})^{-1}(x_-) (N + \1)^d}_{W^{k, 1}}\\
		&\le \sup\limits_{s \ge 1} \frac{s}{s^2 +\lambda}(\norm{x_+}_{W^{k, 1}} + \norm{x_-}_{W^{k, 1}})\\
		&= \sup\limits_{s \ge 1} \frac{s}{s^2 +\lambda} (\norm{y_+}_1 + \norm{y_-}_1)\\
		&= \sup\limits_{s \ge 1} \frac{s}{s^2 +\lambda} \norm{y}_1\\
		&= \sup\limits_{s \ge 1} \frac{s}{s^2 +\lambda} \norm{x}_{W^{k, 1}}
	\end{align*}
	Lastly we can use the bijectivity of $(\lambda - \cI_{d, 1})$ on $\cT_f$ (q.v.~Lemma \ref{lem:(n+1)-(n+1)-properties}) and triangle inequality to conclude 
	\begin{equation}
		\norm{(N + \1)^d x (N + \1)^d}_{W^{k, 1}} \le \sup\limits_{s \ge 0} \frac{s}{s^2 + \lambda} \norm{\cI_{d, 1}(x)}_{W^{k, 1}} + \lambda \sup\limits_{s \ge 0} \frac{s}{s^2 + \lambda} \norm{x}_{W^{k, 1}}\,.
	\end{equation}
	Choosing $\lambda = \frac{1}{4\varepsilon^2}$, we find that $\sup\limits_{s \ge 1} \frac{s}{s^2 + \lambda} < \varepsilon$, hence
	\begin{equation}
		\norm{(N + \1)^d x (N + \1)^d}_{W^{k, 1}} \le \varepsilon \norm{\{(N +\1)^{4d}, x\}}_{W^{k, 1}} + \frac{1}{4\varepsilon} \norm{x}_{W^{k, 1}}
	\end{equation}
\end{proof}

\begin{lem}[Interpolation Lemma]\label{lem:interpolation-lemma}
	Let $k_0 < k_1 \in \R_+$ and $(\cL, \cD(\cL))$ an operator on $W^{k_j, 1}$, $j = 0, 1$ respectively. Further, assume that the closure of $(\cL, \cD(\cL))$ defines a strongly continuous semigroup $(\cP_t^j)_{t \ge 0}$ with 
	\begin{equation}
		\norm{\cP_t^{j}}_{W^{k_j, 1} \to W^{k_j, 1}} \le M_j e^{\omega_j t} \quad \forall t \ge 0 \, 
	\end{equation}
	in both spaces, respectively. Then for $\theta \in [0, 1]$ and $k_\theta = \theta k_1 + (1 - \theta) k_0$ the following are true
	\begin{enumerate}
		\item\label{item:int-pol-1} The closure of $(\cL, \cD(\cL))$ defines a strongly continuous semigroup on $W^{k_\theta, 1}$ with 
		\begin{equation}\label{eq:interpolation-semigroup-bound}
			\norm{\cP^\theta_t}_{W^{k_\theta, 1} \to W^{k_\theta, 1}} \le M_0^{1 - \theta} M_1^\theta e^{(\omega_{k_0}(1 - \theta) + \omega_{k_1} \theta) t} \quad \forall t \ge 0\, . 
		\end{equation}
		\item\label{item:int-pol-2} $(\cP^\theta_t)_{t \ge 0}$ agrees with $(\cP_t^j)_{t \ge 0}$ on $W^{k_j, 1} \cap W^{k_\theta, 1}$ for $j = 1, 2$.
	\end{enumerate}
\end{lem}
\begin{proof}
	We begin with the second claim and only cover $k_j = k_0$ as the other case only requires minor changes that are left to the reader. Let $\theta \in (0, 1)$ and the closure of $(\cL, \cD(\cL))$ the generator of $(\cP_t^0)_{t \ge 0}$ and $(\cP_t^\theta)_{t \ge 0}$ on the respective spaces. Since $k_0 < k_\theta$ and hence $W^{k_\theta, 1} \Subset W^{k_0, 1}$, we have that the closure of $(\cL, \cD(\cL))$ in $W^{k_\theta, 1}$ agrees with the restriction of the closure in $W^{k_0, 1}$. As the semigroup is completely determined by its generator, we find that the semigroups agree on $W^{k_\theta, 1}$.\newline
	For the first claim, note that the semigroups $(\cP_t^j)_{t \ge 0}$, $j = 1, 2$ agree on $W^{k_1, 1}$ by \Cref{item:int-pol-2}, which allows us to employ the Stein-Weiss theorem for Bosonic Sobolev spaces (Theorem \ref{thm:stein-weiss}) to conclude \Cref{eq:interpolation-semigroup-bound}. It remains to check that the families of bounded maps $(\cP_t^\theta)_{t \ge 0}$ are strongly continuous semigroups generated by the closure of $(\cL, \cD(\cL))$ on $W^{k_\theta, 1}$. We have that $\cP_0^\theta = \1$  and $\cP_t^\theta \cP_s^\theta = \cP_{t + 1}^\theta$ $\forall t, s \ge 0$ as a consequence of $\cP_t^0 |_{W^{k_\theta, 1}} = \cP_t^\theta$ $\forall t$ (this equality holds by Theorem \ref{thm:stein-weiss}). The strong continuity follows, due to $\cP^1_t = \cP_t^\theta|_{W^{k_1, 1}}$, as for $x \in W^{k_1, 1}$
	\begin{equation*}
		\lim\limits_{t \to 0} \norm{\cP_t^\theta(x)  - x}_{W^{k_\theta, 1}} = \lim\limits_{t \to 0} \norm{\cP^1_t(x) - x}_{W^{k_\theta, 1}} \le \lim\limits_{t \to 0} \norm{\cP^1_t(x) - x}_{W^{k_1, 1}} = 0
	\end{equation*}
	where we used $W^{k_1, 1} \Subset W^{k_\theta, 1}$ and that $\cP^1_t$ is a strongly continous semigroup on $W^{k_1, 1}$. For general $x \in W^{k_\theta, 1}$, we find $(x_n)_{n \in \N} \subset W^{k_\theta, 1}$ converging to $x$ in $W^{k_\theta, 1}$ and for all $n \in \N$
	\begin{align*}
		\lim\limits_{t \to 0} \norm{\cP_t^\theta(x)  - x}_{W^{k_\theta, 1}} &\le \lim\limits_{t \to 0}\left[ (1 + M_0^{1 - \theta} M_1^\theta e^{(\omega_{k_0}(1 - \theta) + \omega_{k_1} \theta) t}) \norm{x - x_n}_{W^{k_\theta, 1}} + \norm{\cP_t^\theta(x_n)  - x_n}\right]\\
		&\le (1 + M_0^{1 - \theta} M_1^\theta)\norm{x - x_n}_{W^{k_\theta, 1}}
	\end{align*}
	which concludes the strong continuity. It remains to argue that the closure of $(\cL, \cD(\cL))$ on $W^{k_\theta, 1}$ is indeed the generator of $(\cP_t^\theta)_{t \ge 0}$. By \cite[Sec. II.2.3]{Engel.2000}, we find that the restriction of the generator $(\hat{\cL}, \cD(\hat{\cL}))$ of $(\cP_t^\theta)_{t \ge 0}$ to $W^{k_1, 1}$ is the generator of $(\cP_t^1)_{t \ge 0}$ with $(\cL, \cD(\cL))$ being a core for this restricted generator on $W^{k_1, 1}$ by assumption. This in particular means that for $\lambda > \omega_k$, $\lambda - \cL:\cD(\cL) \to W^{k_1, 1}$ has a dense range in $W^{k_1, 1}$, which further allows us to conclude that $\lambda - \cL:\cD(\cL) \to W^{k_\theta, 1}$ has a dense range in $W^{k_\theta, 1}$. Now as it is $\omega_{k_\theta}$-quasi dissipative (being the restriction of the generator of the semigroup $(\cP_t^\theta)_{t \ge 0}$) we can conclude that indeed $(\cL, \cD(\cL))$ is closable in $W^{k_\theta, 1}$ with the closure $(\hat{\cL}, \cD(\hat{\cL}))$ (c.f. \cite[Proposition 3.14]{Engel.2000}). 
\end{proof}

\begin{lem}[Approximation Lemma]\label{lem:approximation-lemma}
	Let $K \in \N$. For $i = 1, \hdots, K$ let $p_{i, 1}, p_{i, 2} \in \C[X,Y]$ polynomials of degree $d_{i, 1}$, $d_{i, 2}$ and $\{a_{i, n}\}_{n \in \N} \subset \C$ convergent sequences with limits $a_i \in \C$ such that $\{(\cA_n, \cD(\cA_n) = \cT_f)\}_{n \in \N}$ is an operator sequence with
	\begin{equation}
		\cA_n :\cT_f \to \cT_f, \quad x \mapsto \cA_n(x) \coloneqq \sum\limits_{i = 1}^K a_{i, n} A_{i, 1} \,x\, A_{i, 2} \coloneqq \sum\limits_{i = 1}^K a_{i, n} p_{i, 1}(a, a^\dagger) \,x\, p_{i, 2}(a, a^\dagger) \, 
	\end{equation}
    with $A_{i,1}=p_{i,1}(a,\ad)$ and $A_{i,2}=p_{i,2}(a,\ad)$ (see Lemma \ref{lem:infinitesimal-boundedness-W-k-1}). If for all $k \in \R_+$ there exists $M_k, \omega_k$ such that for all $n \in \N$ the closure of $(\cA_n, \cD(\cA_n))$ generates a strongly continuous semigroup $(\cP_t^n)_{t \ge 0}$ on $W^{k, 1}$ with
	\begin{equation}\label{eq:uniform-bound-semigroups}
		\norm{\cP_t^n}_{W^{k, 1} \to W^{k, 1}} \le M_k e^{\omega_k t} \quad \forall t \in \R \, ,
	\end{equation}
	then the closure of $(\cA, \cD(\cA) = \cT_f)$, the pointwise limit of $(\cA_n, \cD(\cA_n))$, defines a strongly continuous semigroup on $W^{k, 1}$ for $k \ge 0$ as well. We further get that the semigroups generated by the closure of $(\cA_n, \cD(\cA_n))$ converge uniformly (in time) on compact intervals to the semigroup generated by the closure of $(\cA, \cD(\cA))$ and that \Cref{eq:uniform-bound-semigroups} also holds for the limiting semigroup.
\end{lem}
\begin{proof}
	Let $k \in \R_+$. To prove the lemma, we first note that $(\cA, \cD(\cA))$ is densely defined and the pointwise limit of $\{(\cA_n, \cD(\cA_n))\}_{n \in \N}$. To employ the second Trotter-Kato approximation theorem, which implies the claim (see the version in \cite[Thm.~III.4.9]{Engel.2000}), we need to show that there exists $\lambda > 0$ such that $(\lambda - \cA, \cD(\cA))$ has dense range in $W^{k, 1}$. We will do so by showing that the closure of the range contains $\cT_f$ which is a dense subset of $W^{k, 1}$. Therefore let $\lambda > \max\{\omega_k, \omega_{k + 4d}\}$ (with $\omega_\cdot$ from \Cref{eq:uniform-bound-semigroups} and $d$ the maximal degree of the polynomials but at least one, i.e.~$d = \max\limits_{i = 1, \hdots, K} \max\{d_{i, 1}, d_{i, 2}, 1\}$). By assumption, we have that for all $n \in \N$ the operator $(\lambda - \cA_n, \cD(\cA_n))$ has dense range in $W^{k + 4d, 1}$, meaning in particular that 
    for any $\xi \in \cT_f$ we find a sequence $\{x_{n, m}\}_{m \in \N}$ which is convergent in $W^{k + 4d, 1}$ and further
	\begin{equation}
		\lim\limits_{m \to \infty} \norm{(\lambda - \cA_n)(x_{n, m}) - \xi}_{W^{k + 4d, 1}}  = 0 \, . 
	\end{equation}
	In addition, we can choose the sequence such that for all $m \in \N$, $\norm{(\lambda - \cA_n)(x_{n, m})}_{W^{k + 4d, 1}} \le \norm{\xi}_{W^{k + 4d,1}} + 1$. Due to the $\omega_{k + 4d}$-quasi dissipativity of $\cA_n$ (it is a generator of a strongly continuous semigroup with a bound given in \Cref{eq:uniform-bound-semigroups}) this immediately implies $\norm{x_{n, m}}_{W^{k + 4d, 1}} \le M_{k + 4d}\frac{\norm{\xi}_{W^{k + 4d, 1}} + 1}{\lambda - \omega_{k + 4d}} =: c_{\xi}$, i.e.~the set $\{x_{n, m}\}_{n, m \in \N}$ is bounded in $W^{k + 4d, 1}$. We now have that for $n, m \in \N$
	\begin{equation}\label{eq:dense-range-argument-inequality}
		\begin{aligned}
			\norm{(\lambda - \cA)(x_{n, m}) - \xi}_{W^{k, 1}} &\le \norm{(\lambda - \cA_n)(x_{n, m}) - \xi}_{W^{k, 1}} + \norm{(\cA - \cA_n)(x_{n, m})}_{W^{k, 1}}\\
			&\le \norm{(\lambda - \cA_n)(x_{n, m}) - \xi}_{W^{k, 1}} + \sum\limits_{i = 1}^K c_{i, k} |a_i - a_{i, n}|  \norm{x_{n, m}}_{W^{k + 4d, 1}}\\
			&\le \norm{(\lambda - \cA_n)(x_{n, m}) - \xi}_{W^{k, 1}} + \sum\limits_{i = 1}^K |a_i - a_{i, n}| c_{i, k} c_\xi\\
			&\le  \norm{(\lambda - \cA_n)(x_{n, m}) - \xi}_{W^{k, 1}} + C\sum\limits_{i = 1}^K |a_i - a_{i, n}| \, .
		\end{aligned}
	\end{equation}
	In the first line we used triangle inequality, in the second one the explicit form of $\cA$ and $\cA_n$ and then that there exists $c_{k, i} \ge 0$ such that 
	\begin{align*}
		\norm{(N + \1)^{k/4}A_{i, 1} x_{n, m} A_{i, 2} (N + \1)^{k/4}}_1\le c_{k, i} \norm{(N + \1)^d x_{n, m} (N + \1)^d}_{W^{k, 1}}
	\end{align*}
	as in the proof of Lemma \ref{lem:infinitesimal-boundedness-W-k-1}. Lastly we used the uniform bound $c_\xi$ and set $C = \max_{i = 1, \hdots, K} c_{i,k} c_\xi$. By a proper choice of a subsequence of $\{x_{n, m}\}_{n, m \in \N}$, we get that the RHS of \Cref{eq:dense-range-argument-inequality} vanishes. Since $\{x_{n, m}\}_{n, m \in \N}$ is bounded in $W^{k + 4d, 1}$ it is in particular precompact in $W^{k, 1}$ (as of the compact embedding of the Sobolev spaces), meaning we can further choose the aforementioned sequence to be convergent in $W^{k, 1}$. Let us call it $\{y_n\}_{n \in \N} \subset \cT_f$. To summarise, for the chosen $\lambda$ and $\xi \in \cT_f$ arbitrary we have constructed a sequence $\{y_n\}_{n \in \N} \subset \cD(\cA)$ which is convergent in $W^{k, 1}$ and further $\{(\lambda - \cA)(y_n)\}_{n \in \N}$ converges to $\xi$ in $W^{k, 1}$. Hence the closure of the range of $(\lambda - \cA, \cD(\cA))$ contains $\cT_f$ a dense subset of $W^{k, 1}$, which concludes the proof. 
\end{proof}

\begin{rmk*}
	In the above lemma, it suffices to assume that the semigroups are Sobolev preserving, as one can interpolate between the sequence elements to obtain semigroups for $k \in \R_+$.
\end{rmk*}

\addtocontents{toc}{\protect\setcounter{tocdepth}{2}}

\end{document}